\documentclass[10pt]{article}
\usepackage{hyperref}
\usepackage{amssymb}
\usepackage{amsmath}
\input{liemacs10.sty} 
\numberwithin{equation}{section}

\usepackage{color}
\def\red{\textcolor[rgb]{0.00,0.00,0.00}}

\def\blue{\textcolor[rgb]{0.00,0.00,0.00}}

\renewcommand{\up}{{\mathop{\uparrow}}}
\renewcommand{\down}{{\mathop{\downarrow}}}
\usepackage[usenames,dvipsnames]{xcolor}
\usepackage{cancel}

\newcommand{\sH}{{\sf H}}
\newcommand{\sN}{{\sf N}}
\newcommand{\sK}{{\sf K}}
\newcommand{\sV}{{\tt V}}

\newcommand{\RR}{\mathbb{R}}
\newcommand{\CC}{\mathbb{C}}

\newcommand{\PSU}{\mathrm{PSU}}

\renewcommand{\bO}{\mathbb O}

\newcommand{\Pin}{\mathop{{\rm Pin}}\nolimits}

\newcommand{\Out}{\mathop{{\rm Out}}\nolimits}

\newcommand{\Conj}{\mathop{{\rm Conj}}\nolimits}
\newcommand{\Inv}{\mathop{{\rm Inv}}\nolimits}

\renewcommand{\phi}{\varphi}

\newcommand{\Stand}{\mathop{{\rm Stand}}\nolimits}

\newcommand{\Mob}{{\rm\textsf{M\"ob}}}
\newcommand{\ucW}{{\uline \cW}}
\newcommand{\uW}{{\uline W}}

\newcommand{\uG}{{\underline{G}}}

\renewcommand\mlabel{\label} 

\newtheorem{assumption}{Assumption}

\begin{document}

\title{Covariant homogeneous nets of standard subspaces }

\author{
{\bf Vincenzo Morinelli}\\
Dipartimento di Matematica, Universit\`a di Roma ``Tor
Vergata''\\ 
E-mail: {\tt morinell@mat.uniroma2.it}
\and
{\bf Karl-Hermann  Neeb}\\
Department  Mathematik, FAU Erlangen-N\"urnberg, \\ 
E-mail: {\tt  neeb@math.fau.de}}

\date{}
\maketitle 

\begin{abstract}
Rindler wedges are fundamental localization regions in AQFT. They are determined by the one-parameter group of boost symmetries fixing the wedge. The algebraic canonical construction of the free field provided by Brunetti-Guido-Longo (BGL) arises from the wedge-boost identification, the BW property and the 
PCT Theorem. 

In this paper we generalize this picture in the following way. Firstly, given a $\mathbb Z_2$-graded Lie group we define a (twisted-)local poset of abstract wedge regions. We classify (semisimple) Lie algebras 
supporting abstract wedges and study special wedge configurations. This allows us to exhibit an analog of the Haag-Kastler one-particle net axioms for such general Lie groups without referring to any specific spacetime. This set of axioms supports a first quantization net obtained by generalizing the BGL construction. The construction is possible for a large family of 
Lie groups and provides several new models. We further comment on orthogonal wedges and extension of symmetries.

\end{abstract}
\tableofcontents

\section{Introduction}

Quantum Field Theory (QFT) lives in a tension between the locality principle and the underlying group of symmetries characterizing the theory. On one hand, 
it is a physical principle that every interesting quantity of a theory 
should be deducible by local measurements, namely---in the language of 
Algebraic Quantum Field Theory (AQFT)---by the structure 
of the local algebras (see e.g.~\cite{Ha96}). 
On the other hand, the symmetries of a theory provide a 
feature to describe physical objects,  a ``key to nature’s secrets,'' 
as it happens in the standard model \cite{We05, We11}.

In AQFT, models are specified by a net of von Neumann algebras associated to 
causally complete spacetime regions satisfying fundamental quantum and relativistic principles, such as 
isotony, locality, covariance, positivity of the energy, and 
existence of a vacuum state. An important bridge between the geometry 
and the algebraic structure is the Bisognano--Wichmann (BW) 
property of (A)QFT  claiming that  the modular group of the algebra 
associated to any Rindler wedge $W$ inside 
Minkowski spacetime with respect to the vacuum state 
implements unitarily the covariant one-parameter group of boosts fixing the wedge $W$. As a consequence, the algebraic structure of the model, through the Tomita--Takesaki theory, contains the information about the symmetry group acting on the model. Starting with the BW property, 
one can enlarge the symmetry group of a QFT \cite{GLW, MT18}, find new relations among field theories \cite{GLW, LMPR19, MR}, establish 
proper relations among spin and statistics \cite{GL95}, 
and compute entropy in QFT \cite{LX, Witt}. For 
recent results on this property we refer to \cite{Gu19, DM}.

Particles are field-derived concepts that can be described as unitary positive energy representations of the symmetry group. They are building blocks 
to construct Quantum Field Theories. The operator-valued 
distribution $\Phi_U$ defining the free field 
associated to any  particle $U$ is not provided by a canonical construction, see e.g.~\cite{BGL02,LMR16}. On the other hand, 
the von Neumann algebra net generated by $\Phi_U$ satisfies the Bisognano--Wichmann property and the PCT Theorem\footnote{The spacetime reflection 
$j_1(t,x_1,\ldots,x_n)=(-t,-x_1,\ldots, x_n) $ is implemented by the 
modular conjugation corresponding to the standard right wedge $W_1$.}. 
These properties provide the tools for 
a canonical construction of the free algebra net \cite{BGL02}: 
Segal's second quantization gives the vacuum representation 
of the Weyl algebra on the Fock space associated 
with the one-particle Hilbert space.  The Araki lattice of von Neumann algebras is uniquely determined by the local one-particle structure 
encoded in the lattice of closed real subspaces, 
the first quantization \cite{Ar63}.  
As a result of the Tomita--Takesaki modular theory for real subspaces, the set of real states for a particle $U$ localized in a wedge region is uniquely determined by the couple 
$(e^{-2\pi K_W}, U(j_W))$ where  $U(j_W)$ is the antiunitary 
implementation of the wedge reflection and $K_W$ is the generator of the 
one-parameter group of boosts associated to the wedge $W$. 
They satisfy the Tomita relation $U(j_W) e^{2\pi K_W}U(j_W)= e^{-2\pi K_W}$. The 
one-particle states and the local algebra associated to bounded causally 
complete regions are obtained by wedge state spaces and algebra intersection, 
respectively. 

Conversely, every pair $(x,\sigma)$, consisting of an element 
of the Poincar\'e--Lie algebra and an involution $\sigma$ satisfying 
$\Ad(\sigma)x = x$ specifies for every (anti-)unitary representation 
$(U,\cH)$ of the Poincar\'e group a pair $(\Delta, J) 
= (e^{2\pi i \partial U(x)}, U(\sigma))$ that in turn defines a standard 
subspace $\sV \subeq \cH$.  
This construction, called the {\it BGL construction}, was introduced in
 \cite{BGL02} and allows us to 
observe: \textit{The algebraic construction of the free fields is uniquely determined by its symmetries and the correspondence 
between spacetime regions and their relative position with symmetries}. 
In this sense, due to the one-to-one correspondence between 
boosts and the corresponding wedges, one should be able to 
specify the underlying symmetry structure of 
a quantum field theory without any reference to the spacetime. Then one can reconstruct the spacetime features, such 
as locality and region inclusions from the symmetry group. 

With this claim in mind,  we generalize the above picture as follows. Given 
a suitable Lie group~$G$, we first define an abstract wedge space. 
We then endow the wedge space with a $G$-action, 
a notion of causal complement and 
an order structure. 
  Eventually, starting from an (anti-)unitary representation of a 
graded Lie group $G$, we construct the analogue of the BGL 
one-particle net by the abstract setting.
 
We now collect the motivation and additional explanations of 
the fundamental structure we will use. In order to obtain 
a one-particle net by the Tomita--Takesaki theory we need to start with a 
graded Lie group $G=G^\up\rtimes \Z_2$, such as 
the improper M\"obius group $\PGL_2(\R)$ 
or the proper Poincar\'e group $\cP_+$. 
For the moment,  we assume that $Z(G^\up)=\{e\}$ and that $G^\up$ is connected. 

The key features of our approach are the following: \\
$\bullet$ \textit{Abstract boost generator.} The abstract one-parameter group 
of boosts are generated by elements $x$ 
in the Lie algebra $\fg$ of $G$ defining a three grading $\fg=\g_1 \oplus \g_0 \oplus \g_{-1}$ in the adjoint representation by 
$\g_j = \ker(\ad x - j \id_\g)$. 
To see how this complies with the well known models, see 
Examples~\ref{ex:models}. 
We call such elements $x\in\g$ Euler elements 
because they corresponds to the linear Euler vector field 
on the open embedding $\g_1 \into G/{\bf P}, x \mapsto \exp(x) {\bf P}$, 
where ${\bf P} \subeq G$ is the connected subgroup corresponding to the 
Lie algebra $\g_0 + \g_{-1}$. For more on the underlying geometry of 
theses spaces, we refer to  \cite{BN04}. 

$\bullet$ \textit{The wedge reflection} is obtained by 
analytic continuation of the one-parameter group of boosts 
associated to the wedge at $i\pi$. For instance, on Minkowski space,  the wedge reflection 
$j_1=\Lambda_1(i\pi)$ is obtained by analytic extension of the one-parameter group of boosts in the first direction $\Lambda_1(t)=\exp(\sigma_1t)$ where 
$(\sigma_i)_{i=1,2,3}$ are the Pauli matrices.  In our general setting, the reflection $\sigma$, called Euler involution, associated to an Euler element $x$ is determined by the analytic continuation 
of the one-parameter group in the adjoint representation of the Lie algebra 
via $\Ad(\sigma)=e^{\pi i \ad x}$ (see \eqref{eq:eul}). 

$\bullet$ \textit{Euler wedge.} An Euler wedge is defined as a 
couple $W=(x_W, \sigma_W)$ of an Euler element and the related Euler involution.
 The need to use the couple is to 
implement the $G$-action 
on the wedge space (see \eqref{eq:homact}) and to establish 
 the relation with the standard subspaces $\sV$ and the 
corresponding modular objects $(\Delta_\sV, J_\sV)$. 
We further remark that,  in principle, it is not necessary to  assume that the involution $\sigma_W$ satisfies 
\[ \Ad(\sigma_W)=e^{\pi i \ad x_W} \] 
in the adjoint representation, 
but only satisfying the proper commutation relation 
$\Ad(\sigma_W) x_W =  x_W$, cf.~Proposition \ref{prop:2.1}.  
 
$\bullet$ \textit{$G^\uparrow$-covariance}. There is an action 
of the group $G$ on the wedge space 
given by an adjoint action on both components 
that takes care of the grading (see \eqref{eq:cG-act}). 
In this way the language of Euler wedges is consistent with the one of the standard subspaces, cf.~Section~\ref{sec:1}.
 
$\bullet$ \textit{Locality.} Complementary wedges correspond to 
inverted one-parameter groups of boosts. For instance dilations associated to causally complementary intervals in chiral theory or boosts associated to complementary 
wedges are inverse to each other.  On the abstract wedge space this is captured by defining the complementary wedge of  $W=(x,\sigma)$ 
by $W'=(-x, \sigma)$.

 $\bullet$ \textit{Isotony.}  By the existence of a (positive) 
invariant 
cone $C$ in the Lie algebra $\fg$, it is possible to define a wedge endomorphism 
semigroup defining the wedge inclusion relation. Given an Euler wedge $W=(x,\sigma)$, the generators in the positive cone lying in the subspaces 
$\fg_{\pm1}$ define proper wedge inclusions as each of them generates with $x$ 
a translation-dilation group (isomorphic to the affine 
group of the real line); see \cite{Bo92, Wi92, Wi93} and in particular 
\cite{Bo00}.  
This is the case of wedge endomorphisms in Minkowski spacetime given by lightlike shifting or M\"obius  transformations mapping an interval into itself as the translations  
do for the half-lines. 
\textit{These properties define a local partially ordered set of wedges that can support key features of an AQFT structure.}
 
It is important to note that the wedge space \textbf{only  depends 
on the Lie group and its Lie algebra}, and the order structure on 
the invariant cone $C \subeq \g$. 
The relations among the wedges specify  
the abstract spacetime structure to a large extent. 
For example, $\PSL_2(\R)$ is the symmetry  group  for the 2-dimensional de Sitter spacetime and for the chiral circle. If one considers $\PSL_2(\R)$ 
with the trivial cone in $\fsl_2(\R)$---no proper 
inclusions of wedges---then it describes a QFT on de Sitter spacetime; if one considers 
$C\subset \fsl_2(\R)$ as in $\eqref{eq:cone}$, inclusion relations 
among wedges arise, and we obtain the wedge space on $\bS^1$. 
\begin{footnote}{In \cite{GL03} it is used that the 
$2$-dimensional de Sitter space 
$\dS^2 \cong \SO_{1,2}(\R)^\up/\SO_{1,1}(\R)^\up$ 
has the same abstract wedge space as the circle $\SO_{1,2}(\R)/{\bf P}$ 
to set up a $\dS/CFT$ correspondence.}
\end{footnote}
This correspondence between isotony and and positivity of the energy was also studied in \cite{GL03}; see also 
\cite{Bo92, Bo00, Wi92, Wi93} and \cite{NO17, Ne19, Ne19b}. 
For recent classification results for the triples
 $(\g,x,C)$, we refer to \cite{Oeh20, Oeh20b}. 
  
There is more interesting structure on the abstract wedge space:
  
$\bullet$ \textit{Orthogonal wedges}: We call two abstract wedges  
$W_1 = (x_1, \sigma_1)$ and $W_2 = (x_2,\sigma_2)$ orthogonal if 
$\sigma_1(x_2)=-x_2$, i.e., $W_2$ is reflected into its complement~$W_2'$. 
Examples of orthogonal wedges are coordinate wedges on Minkowski spacetime\footnote{For instance  $W_i$ and $W_j$ for $i\neq j$, where $W_i=\{(t,x)\in\RR^{1+s}:|t|<x_i\}$}, or the upper and the right half-circle in chiral theories 
on $\bS^1$. 
This notion, which immediately generalizes to the abstract setting, plays 
 a central role in  spin-statistics relations \cite{GL95} and the 
nuclearity property in conformal field theory \cite{BDL}.

$\bullet$ \textit{Symmetric wedges.} A wedge $W$ is called 
symmetric  if there exists $g\in G^\up$, such that $g.W=~W'$. 
For instance, any couple of wedge regions,  in $1+s$-dimensional 
 Minkowski spacetime with $s\geq2$, are transformed one into the other by the 
action of the Poincar\'e group $G^\up=\cP_+^\uparrow$.  
 On the other hand, in $1+1$-dimensional Minkowski space, the right and the left wedges are not symmetric. Indeed 
\[ W_{R}=\{(t,x)\in\RR^{1+1}: |t|< x\} \quad \mbox{ and } \quad 
W_{L}=\{(t,x)\in\RR^{1+1}: |t|<- x\} \] 
 belong to disjoint transitive families with respect to the 
$\cP_+^\uparrow$-action. Further examples of symmetric wedges are intervals in 
conformal theories on the circle. Half-lines in the real line are not symmetric wedges with respect to the translation-dilation group.  A transitive family of wedges has the feature that algebras associated to complementary wedges are - by covariance - unitary equivalent. On the other side, there is  no contradiction in having a $G^\up$-covariant net of von Neumann algebras on a transitive family of non-symmetric wedges with trivial algebras associated to the 
family of complements.

In the first part of the paper we define and investigate the 
abstract structure we have described. 
When the center $Z(G^\up)$ is non trivial, for instances when covering groups are considered, a generalized notion of complementary wedges has to be introduced. Indeed, while Euler elements are uniquely determined as generators of one-parameter 
groups in $G^\up$, several involutions  $\sigma$ 
satisfying $\Ad(\sigma)=e^{\pi i \ad x}$ 
can be associated to the same Euler element~$x$. 
In an analogous way, different  wedge complements 
can be labeled by central elements. 
We classify wedge orbits and define a notion of a central wedge complement.  
Furthermore, if $W'$ does not belong to the $G^\up$-orbit of $W$, a new action of $G$ on the wedge space is defined. This happens for instance in fermionic nets.

Having specified the abstract structures, we are prepared to 
answer the following question:\\
\textit{``Which  Lie algebras/groups support such a structure?" } 
To this end, we first classify Euler elements in real simple Lie algebras 
in Theorem~\ref{thm:classif-symeuler}. 
The key point of this classification is that 
Euler elements are conjugate under inner automorphisms  
to elements in any given Cartan subspace of hyperbolic elements. 
Here the restriction to simple Lie algebras is not 
restrictive because any symmetric Euler element is contained in a 
semi-simple Lie subalgebra. Furthermore, an  Euler 
element is symmetric if and only if it is contained in 
an $\fsl_2(\R)$-subalgebra 
(see Theorem~\ref{thm:automaticsym} for these results).  
As a consequence, there is a large family of real 
Lie algebras supporting such wedge structures which properly contains the well known models.  Note that, 
for a Lie algebra $\fg$ containing an Euler element 
$x \in \g$, there always exists a graded Lie group $G$  
with Lie algebra $\fg$ and a corresponding Euler 
wedge $(x,\sigma)$. 

The second part of the paper is devoted to  nets of standard subspaces.\\\textit{Is it possible to construct one-particle models supporting this abstract setting?} 
Starting with a $G^{\uparrow}$-orbit $\cW_+$ in the wedge space, 
we describe a set of axioms which, for the well known models,  reflect 
fundamental quantum and relativistic 
principles corresponding to the one-particle Haag--Kastler axioms. 
This set of axioms is fulfilled by extending the BGL construction to every graded Lie group~$G$, supporting a 
suitable wedge space. 
A twisted locality relation among complementary wedges 
is introduced in order to relate central complementary wedges. 

\textit{Do we get any new models out of this general construction?} The answer is affirmative. All the simple Lie algebras whose 
restricted root system appears in Theorem \ref{thm:classif-symeuler} 
correspond to a graded Lie group with a non-trivial wedge space. 
There are for instance 
Lie algebras of type $E_7$ that do not correspond to any known models. In 
this context the Jordan spacetimes of 
G\"unaydin \cite{Gu93, Gu00, Gu01} and the simple spacetime manifolds 
in the sense of Mack--de Riese \cite{MdR07} are homogeneous 
spaces of simple hermitian Lie groups whose Lie algebras contain 
Euler elements, and the corresponding abstract wedges correspond to domains 
in these causal manifolds. These Lie groups have many (anti-)unitary 
representations, some of them with positive energy with 
respect to a non-trivial invariant cone $C$ in the Lie algebra. 
As a consequence, they support many one-particle nets 
\cite{NO20} and second quantization 
models of von Neumann algebras whose physical meaning has to be investigated.

The structure of this paper is as follows: In Section~2 the wedge space 
is defined and its properties are studied. 
A  number of examples are discussed in detail 
to show how  the abstract setting applies to the known models and realizes 
the well known structure. In Section~3 we study the Euler elements in 
Lie algebras. We relate orthogonal and symmetric wedges and provide a classification of Lie algebras supporting (symmetric) Euler elements.
In Section~4 we apply this structure to define and construct one-particle nets associated to graded Lie groups supporting a wedge structure. We further stress new models, orthogonal wedges and  extension of symmetries. 
\red{We hope this paper is approachable for the 
Lie Theory community as well as 
the Algebraic Quantum Field Theory community. }

\subsection*{Acknowledgments}
VM was supported by the European Research Council Advanced Grant 669240 QUEST and by Indam from March 2019 to February 2020.
VM acknowledges the MIUR Excellence Department Project awarded to
the Department of Mathematics, University of Rome ``Tor Vergata'', CUP E83C18000100006 and the University of Rome ``Tor Vergata'' funding scheme ``Beyond Borders'', CUP E84I19002200005.

\nin KHN acknowledges support by DFG-grant NE 413/10-1.

\section{The abstract setting} \mlabel{sec:1} 

In this section we develop an abstract perspective on 
wedge domains in spacetimes, phrased completely in group theoretic terms. 
As wedge domains are supposed to correspond to standard subspaces 
in Hilbert spaces, we orient our approach on how standard subspaces 
are parametrized. 

Let $\Stand(\cH)$ denote the set of standard subspaces of the complex 
Hilbert space~$\cH$. 
In Section~\ref{sec:4} we shall see that every standard subspace 
$\sV$ determines a pair $(\Delta_\sV, J_\sV)$ of modular objects 
and that $\sV$ can be recovered from this pair by 
$\sV = \Fix(J_\sV \Delta_\sV^{1/2})$. 
This observation can be used to obtain a representation theoretic 
parametrization of $\Stand(\cH)$: 
each standard subspace $\sV$ specifies a continuous homomorphism 
\begin{equation}
  \label{eq:uv-rep-1}
 U^\sV \: \R^\times \to \AU(\cH)\quad \mbox{ by } \quad 
U^\sV(e^t) := \Delta_\sV^{-it/2\pi}, \quad 
U^\sV(-1) := J_\sV.
\end{equation}
We thus obtain a bijection between $\Stand(\cH)$ and 
the set $\Hom_{\rm gr}(\R^\times, \AU(\cH))$ of continuous 
morphisms of graded topological groups. 

The space $\Stand(\cH)$ carries three important features: 
\begin{itemize}
\item[$\bullet$] an order structure, defined by  set inclusion
\item[$\bullet$] a duality operation $\sV \mapsto \sV' = \{ \xi \in \cH 
\: (\forall v \in \sV)\, \Im \la \xi, v \ra = 0\}$
\item[$\bullet$] the action of $\AU(\cH)$ as a symmetry group. 
\end{itemize}
The order structure is hard to express in terms of the modular 
groups (see \cite{Ne19b} for some first steps in this direction), 
but the duality operation corresponds to inversion
\begin{equation}
  \label{eq:duality-onepar}
U^{\sV'}(r) = U^{\sV}(r^{-1}) \quad \mbox{ for }\quad r \in \R^\times,
\end{equation}
and the action of $\AU(\cH)$ translates into 
\begin{equation}
  \label{eq:symact}
 U^{g\sV}(r) = g U^{\sV}(r^{\eps(g)}) g^{-1} 
\quad \mbox{ for } \quad g \in \AU(\cH), r \in \R^\times,
\end{equation}
where $\eps(g) = 1$ if $g$ is unitary and $\eps(g) = -1$ otherwise. 
So unitary operators $g \in \U(\cH)$ simply act by conjugation, 
but antiunitary operators also involve inversion. In particular, 
$J_\sV \sV = \sV'$ corresponds to 
\[  U^{\sV'}(r) = J_\sV U^{\sV}(r^{-1}) J_\sV = U^{\sV}(r^{-1})  
\quad \mbox{ for }\quad r \in \R^\times.\] 
We now develop the corresponding structures by replacing 
$\AU(\cH)$ by a finite dimensional graded Lie group. 

\subsection{Group theoretical setting}\label{sect:G}

The basic ingredient of our approach is a finite dimensional 
{\it graded Lie group} 
$(G,\eps_G)$, i.e., $G$ is a Lie group and $\eps_G \:  G \to \{\pm 1\}$ 
a continuous homomorphism. We write 
\[ G^\up = \eps_G^{-1}(1) \quad \mbox{ and } \quad G^\down = \eps_G^{-1}(-1),\] 
so that $G^\up \trile G$ is a 
normal subgroup of index $2$ and $G^\down = G \setminus G^\up$. 
We also fix a pointed closed convex cone $C \subeq \g$ satisfying 
\begin{equation}
  \label{eq:Cinv}
\Ad(g) C = \eps_G(g) C \quad \mbox{ for } \quad g \in G.
\end{equation}
As we shall see in the following, for graded Lie groups, it is more natural 
to work with the {\it twisted adjoint action} 
\begin{equation}
  \label{eq:adeps}
  \Ad^\eps \: G \to \Aut(\g), \qquad 
\Ad^\eps(g) := \eps_G(g) \Ad(g),
\end{equation}
so that \eqref{eq:Cinv} actually means that $C$ is invariant under 
the twisted adjoint action. The cone $C$ will play a role 
in specifying an order structure. It is related to positive spectrum 
conditions on the level of unitary representations. 
We also allow $C=\{0\}$. For instance, the Lie algebra 
$\g = \so_{1,d}(\R)$ of the Lorentz group $G = \OO_{1,d}(\R)$, 
the isometry group of de Sitter space time $\dS^d$, 
 contains no non-trivial invariant cone. 

\subsection{The space $\Hom_{\rm gr}(\R^\times,G)$ and 
abstract wedges} 

In this section we define the fundamental objects we will need in the forthcoming discussion.
We write $\Hom_{\rm gr}(\R^\times,G)$ for the space of 
continuous morphisms of graded Lie groups $\R^\times \to G$, where
$\R^\times$ is endowed with its canonical grading by $\eps(r) := \sgn(r)$.
On this space $G$ acts by 
\begin{equation}
  \label{eq:homact}
 (g.\gamma)(r) := g \gamma(r^{\eps_G(g)}) g^{-1}, 
\end{equation}
where the twist is motivated by formula \eqref{eq:duality-onepar}.
Elements of $G^\up$ simply act by conjugation. 

Since we are dealing with Lie groups, we also have the following simpler 
description of the space $\Hom_{\rm gr}(\R^\times,G)$ by the set  
\[ \cG := \{ (x,\sigma)\in\g \times G^\down \: \sigma^2 = e, \Ad(\sigma)x = x\}.\] 

\begin{prop} \label{prop:2.1}The map 
\begin{equation}
  \label{eq:Psi}
\Psi\: \Hom_{\rm gr}(\R^\times,G) \to 
\cG, \quad \gamma \mapsto (\gamma'(0), \gamma(-1)) 
\end{equation} 
is a bijection. It is equivariant with respect to the action of $G$ 
on $\cG$  by 
\begin{equation}
  \label{eq:cG-act}
 g.(x,\sigma) := (\Ad^\eps(g)x, g\sigma g^{-1}).
\end{equation}
\end{prop}
\nin Note that center $Z(G^\up)$ of $G^\up$ acts trivially on the Lie algebra 
but it may act non-trivially on involutions in $G^\down$.

\begin{rem} \mlabel{rem:2.2}
For every involution $\sigma \in G^\down$, 
the involutive automorphism $\sigma_G(g) := \sigma g \sigma$ 
defines the structure of a symmetric Lie group $(G^\up,\sigma_G)$,  
and $G \cong G^\up \rtimes \{\id, \sigma\}$, 
so that we can translate between  $G$ as a graded Lie group 
and the pair $(G^\up,\sigma_G)$, without loosing 
information. 
\end{rem} 

To indicate the analogy of elements 
of $\cG$ with the wedge domains in QFT, we shall often 
denote the elements of $\cG$ by $W = (x,\sigma)$. 
\begin{defn}
  \mlabel{def:1.1} 
(a) We assign to 
$W = (x,\sigma) \in \cG$ the one-parameter  group 
\begin{equation}
  \label{eq:deflambdaw}
 \lambda_W \: \R \to G^\up \quad \mbox{ by } \quad 
\lambda_W(t) := \exp(t x)
 \end{equation} 
Then we have  the graded homomorphism 
\[ \gamma_W \: \R^\times \to G, \quad 
\gamma_W(e^t) := \lambda_W(t), \qquad 
\gamma_W(-1) := \sigma.\] 
Note that $\Psi(\gamma_W) = W$ in terms of~\eqref{eq:Psi}. 
\end{defn}

\begin{defn} \mlabel{def:euler}
 (a) We call an element $x$ of the finite dimensional 
real Lie algebra $\g$ an 
{\it Euler element} if $\ad x$ is diagonalizable with 
$\Spec(\ad x) \subeq \{-1,0,1\}$, so that the eigenspace 
decomposition with respect to $\ad x$ defines a $3$-grading 
of~$\g$: 
\[ \g = \g_1(x) \oplus \g_0(x) \oplus \g_{-1}(x), \quad \mbox{ where } \quad 
\g_\nu(x) = \ker(\ad x - \nu \id_\g)\] 
(see \cite{BN04} for more details 
on Euler elements in more general Lie algebras). 
Then $\sigma_x(y_j) = (-1)^j y_j$ for $y_j \in \g_j(x)$ 
defines an involutive automorphism of $\g$. 

For an Euler element we write $\cO_x = \Inn(\g)x \subeq \g$ for the 
orbit of $x$ under the group $\Inn(\g) = \la e^{\ad \g} \ra$ of 
inner automorphisms.\begin{footnote}
{For a Lie subalgebra $\fs \subeq \g$, we write 
$\Inn_\g(\fs)= \la e^{\ad \fs} \ra \subeq \Aut(\g)$ for the subgroup 
generated by $e^{\ad \fs}$.}  
\end{footnote}
We say that $x$ is {\it symmetric} if $-x \in \cO_x$. 

We write $\cE(\g)$ for the set of {non-zero} Euler elements in~$\g$ 
and $\cE_{\rm sym}(\g) \subeq \cE(\g)$ for the subset of symmetric Euler elements.

\nin (b) An element $(x,\sigma) \in \cG$ is called an 
{\it Euler couple} or {\it Euler wedge}   
if 
\begin{equation}\label{eq:eul}
\Ad(\sigma)=e^{\pi i \ad x}.\end{equation} 
Then $\sigma$ is called an {\it Euler involution}. 
We write $\cG_E \subeq \cG$ for the subset of Euler couples 
and note that the relation $e^{\pi i \ad x} = e^{-\pi i \ad x}$ implies 
that the subset $\cG_E$ is invariant under the $G$-action.  
\end{defn} 

For an Euler element $x \in\cE(\g)$, the relation 
\eqref{eq:eul} only determines $\sigma$ up to an element 
$z \in G^\up \cap \ker(\Ad)$ 
for which $(\sigma z)^2 = e$, i.e., $\sigma z \sigma = z^{-1}$. 
Note that, if $G^\up$ is connected, then $G^\up \cap \ker(\Ad) = Z(G^\up)$ 
is the center of $G^\up$. 
The couples $(x,\sigma)$ that we have seen in the physics literature 
are all Euler couples (cf.\ \cite[Ex.~5.15]{NO17}). 
This ensures many properties, such as 
the proper relation between spin and statistics, see for instance \cite{GL95}. 

\begin{defn} \mlabel{def:abs-struc}
(a) (Duality operation) For $W = (x,\sigma) \in \cG$, we define 
$W' := (-x,\sigma)$. Under~$\Psi$, this operation corresponds to inverting
the homomorphism $\R^\times \to G$ pointwise.
Note that $(W')' = W$ and $(gW)' = gW'$ for $g \in G$ 
by \eqref{eq:cG-act}.  

\nin (b) (Order structure on $\cG$) We now define an order 
structure on $\cG$ that depends on the invariant cone 
$C$ from \eqref{eq:Cinv}. We associate to $W = (x,\sigma) \in \cG$ 
  \begin{itemize}
\item the Lie wedge 
\[ L_W := L(x,\sigma) := 
C_+(W) \oplus \underbrace{(\g^{\sigma}\cap \ker(\ad x))}_{\g_W :=} \oplus C_-(W), \] 
where 
\[ C_\pm(W) = \pm C \cap \g^{-\sigma} \cap \ker(\ad x \mp \1)
\quad \mbox{ and } \quad \g^{\pm \sigma} 
:= \{ y \in \g \: \Ad(\sigma)(y) = \pm y \}. \] 
\item $\g(W) := L_W - L_W$, the Lie algebra generated by $L_W$.
\item the semigroup associated to the triple $(C,x,\sigma)$: 
\[ \cS_W := \exp(C_+(W)) G^\up_{W} \exp(C_-(W)) 
= G^\up_{W} \exp\big(C_+(W) + C_-(W)\big),\] 
where 
\[ G^\up_{W} = \{g \in G^\up \: g.W = W \} 
= \{ g \in G^\up\: \sigma_G(g) = g, \Ad(g)x = x\} \] 
is the stabilizer of 
$W = (x,\sigma)$ in $G^\up$ (cf.\ \cite[Thm.~3.4]{Ne19b}).
\begin{footnote}{In \cite{Ne19b} it is shown that the 
different descriptions as a product of two sets 
(polar decomposition) and a product of two abelian subsemigroups 
and a group yield the same set $\cS_W$ which 
actually is a subsemigroup.} 
\end{footnote}
\item the subgroups $G^\up(W) := \la \exp \g(W) \ra G^\up_{W}$ and 
$G(W) := G^\up(W) \{e,\tau\}$ with Lie algebra $\g(W)$. 
  \end{itemize}
As the unit group of $\cS_W$ is given by 
$\cS_W \cap \cS_W^{-1} = G^\up_{W}$ (\cite[Thm.~III.4]{Ne19b}), the semigroup 
$\cS_W$ defines a $G^\up$-invariant partial order on the orbit 
$G^\up.W \subeq \cG$ by 
\begin{equation}
  \label{eq:cG-ord}
g_1.W \leq g_2.W  \quad :\Longleftrightarrow \quad 
g_2^{-1}g_1 \in \cS_W.
\end{equation}
In particular, $g.W \leq W$ is equivalent to $g \in \cS_W$. 
\end{defn}

We have the following relations among these objects:

\begin{lem} \mlabel{lem:1.1} 
For every $W=(x_W,\sigma_W) \in \cG$, $g \in G$, and $t \in \R$, 
the following assertions hold: 
  \begin{itemize}
\item[\rm(i)] $\lambda_W(t) W = W, \lambda_W(t) W' = W'$ and $\blue{\sigma_W.} W = W'.$ 
\item[\rm(ii)] $\sigma_{W'} = \sigma_W$ and 
$\lambda_{W'}(t) = \lambda_W(-t)$. 
\item[\rm(iii)] $\sigma_W$ commutes with $\lambda_W(\R)$. 
\item[\rm(iv)] $L_{W'} = - L_W$ and $\cS_{W'} = \cS_W^{-1}$. 
\item[\rm(v)] $C_\pm(g.W) = \Ad(g) C_{\pm \eps_G(g)}(W)$,  
$L_{g.W} = \Ad(g)L_W$, and $\cS_{g.W} = g \cS_W g^{-1}$. 
\item[\rm(vi)] For $W_1, W_2 \in \cG$, the relation 
$W_1 \leq W_2$ in $\cG$ implies 
$g.W_1 \leq g.W_2$. 
\end{itemize}
\end{lem}

\begin{prf} (i) For $W = (x,\sigma)\in \cG$, the first two relations 
follow from the fact that $\exp(\R x)$ commutes with $x$ and $\sigma$. 
The second follows from 
$\blue{\sigma_W. W} = \sigma.(x,\sigma) = (-\Ad(\sigma) x, \sigma) = (-x, \sigma) = W'.$ 

\nin (ii) is clear from the definition of $W'$. 

\nin (iii) follows from (i). 

\nin (iv) follows from $C_\pm(W') = - C_{\mp}(W)$. 

\nin (v) The assertion is clear for $g \in G^\up$. 
For $g \in G^\down$, we have $g\sigma \in G^\up$, so that 
\begin{align*}
 C_\pm(g.W) 
&= C_\pm(g\sigma.W') 
= \Ad(g\sigma)C_{\pm}(W')
= -\Ad(g\sigma)C_{\mp}(W)
= \Ad(g)C_{\mp}(W)\\
&= \Ad(g)C_{\pm \eps_G(g)}(W).
\end{align*}
This implies in particular that $L_{g.W} = \Ad(g)L_W$. 
From $G^\up_{g.W} = g G^\up_{W} g^{-1}$, we thus obtain 
$\cS_{g.W} = g \cS_W g^{-1}$. 

\nin (vi) If $W_1 \leq W_2$, then $W_1 = s.W_2$ for $s \in \cS_{W_2}$. 
Then $g.W_1 = gs.W_2 = gsg^{-1}.(g.W_2)$ with 
$gsg^{-1} \in g \cS_{W_2} g^{-1} = \cS_{g.W_2}$ implies 
$g.W_1 \leq g.W_2$. 
\end{prf}

In this discussion we started with a Lie group. 
We remark that one can also start with a Lie algebra as follows: 
Consider a quadruple $(\g, \sigma_\g, h, C)$ of a 
Lie algebra $\g$, an involutive automorphism $\sigma_\g$ of $\g$, 
and a pointed closed convex invariant cone $C \subeq \g$ with 
$\sigma_\g(C) = - C$. Then $\sigma_\g$ integrates to an automorphism 
$\sigma_G$ of  the $1$-connected Lie group $G^\up$ with 
Lie algebra $\g$, so that we obtain all the data required above 
with $G := G^\up \rtimes \{\id_G,\sigma_G\}$.

 For two such quadruples 
$(\g_j, \tau_{\g,j}, h_j, C_j)_{j =1,2}$, a homomorphism 
$\phi \: \g_1 \to \g_2$ of Lie algebras 
is compatible with this structure if 
\[  \phi \circ \tau_{\g,1} = \tau_{\g,2} \circ \phi,\qquad 
\phi(h_1) =h_2 \quad \mbox{ and } \quad \phi(C_1) \subeq C_2.\] 
We thus obtain a category whose objects are the quadruples 
$(\g, \tau_\g, h, C)$ and its morphisms are the compatible 
homomorphisms. 

{A similar category can be defined on the group level, but there 
are some subtle ambiguities concerning the possible extensions 
of the group structure from $G^\up$ to $G$.}

\begin{rem}\mlabel{rem:twists-grad} 
(Twisted extensions of $G^\up$ to $G$) 
We start with a graded group $G$ for which $G^\down$ contains an involution 
$\sigma$, so that $G \cong G^\down \rtimes \{e,\sigma\}$, where 
$\sigma$ acts on $G^\up$ by the automorphism $\sigma_G(g) := \sigma g \sigma$. 
This defines a split group extension 
\[  G^\up \to G \to \Z_2 \] 
and we are now asking for other group extensions 
\[  G^\up \to \hat G \to \Z_2 \] 
for which the elements in $\hat G^\down$ define the same element in the group  
$\Out(G^\up) = \Aut(G^\up)/\Inn(G^\up)$ of outer automorphisms of~$G^\up$. 
These extensions are parametrized 
by the group 
\[ Z(G^\up)^+ := \{ z \in Z(G^\up) \: \sigma_G(z) = z\},\] 
by assigning to $z \in Z(G^\up)^+$ the group structure on 
$G^\up \times \{1,-1\}$ given by 
\begin{equation}
  \label{eq:grpext}
(g,1)(g',\eps') = (gg',\eps'), \quad  
(e,-1)(g',1) = (\sigma_G(g'),-1) \quad \mbox{ and } \quad 
(e,-1)^2 = (z,1).
\end{equation}
We write $\hat G_z$ for the corresponding Lie group. 
Basically, this means that the element $\hat\sigma := (e,-1)$ 
has the same commutation relations with $G^\up$ but its square is 
$z$ instead of~$e$: 
\red{\begin{equation}
  \label{eq:hatsigmarel}
\hat\sigma g \hat\sigma^{-1} = \sigma_G(g) \quad \mbox{ for }\quad  g\in G 
\quad \mbox{ and }\quad \hat\sigma^2 = z.
\end{equation}}
 For two elements $z,z' \in Z(G^\up)^+$, 
the corresponding extensions are equivalent if and only if 
\begin{equation}
  \label{eq:defb}
z^{-1}z' \in B := \{ w \sigma_G(w) \: w \in Z(G^\up)\}.
\end{equation}
This follows from \cite[Thm 18.1.13]{HN12}, combined with 
\cite[Ex.~18.3.5(b)]{HN12}.\\
\nin (a) For $G = \OO_{n}(\R)$, $n > 3$, and $G^\up = \SO_{n}(\R)$, 
the situation depends on the parity of $n$. 
If $n$ is odd, then $Z(G^\up) = \{e\}$ and no twists exist. 
If $n$ is even, then $Z(G^\up) = \{\pm \1\} = Z(G)$. 
Therefore $Z(G^\up)^+ = \{ \pm \1 \} \not= B = \{e\}$. 
We therefore have one twisted group 
$\hat G = \SO_n(\R) \{e,\hat\sigma\}$, 
where $\sigma \in \OO_n(\R)$ corresponds to a hyperplane reflection, 
and $\hat\sigma^2 = - \1$ in $\hat G$.

\nin (b) The same phenomenon occurs for Spin groups. 
Let  $G := \Pin_n(\R) \cong\Spin_n(\R) \rtimes \{e,\sigma\}$, 
where $\sigma$ corresponds to a hyperplane reflection. 
If $n$ is odd, then $Z(\Spin_n(\R)) = \{  e,z\}$ contains two elements, 
and we have a twisted group 
\[ \hat G = \Spin_n(\R) \{ e, \hat\sigma \} 
\quad \mbox{ with } \quad 
\hat\sigma^2 = z \] 
(cf.\ \cite[Rem.~B.3.25]{HN12}). If $n$ is even, then the situation is more 
complicated because the center of $\Spin_n(\R)$ has order~$4$.

\nin (c) For $G = \tilde\Mob \rtimes \{e,\sigma\}$, 
where $\sigma$ corresponds to a reflection $\sigma(x) = -x$ 
on $\R^\infty \cong \bS^1$, 
we have $Z(G^\up) \cong \Z$ and $\sigma_G(z) = z^{-1}$ for $z \in Z(G^\up)$. 
Hence $Z(G^\up)^+ = \{e\}$, so that there are no twists.

\nin(d) 
If $G = \Mob^{(2n)} \rtimes \{e,\sigma\}$, where 
$\Mob^{(2n)}$ is the covering of $\Mob$ of even order, then 
$Z(G^\up) \cong \Z_{2n}$ and $\sigma_G(z) = z^{-1}$ for $z \in Z(G^\up)$. 
Therefore $Z(G^\up)^+ = \{ e, \gamma\}$, where $\gamma$ is the 
unique non-trivial involution in $Z(G^\up)$ and $B = \{e\}$. 
Hence there exists a non-trivial twist 
\red{$\hat G = G^\up\{ e,\hat\sigma\}$} with $\hat\sigma^2 = \gamma$.

\nin(e) As we shall see in Example~\ref{ex:2.13} 
below, it may happen that, for the twisted 
groups $\hat G_z$, the coset $\hat G_z^\down$ contains no involutions. 
In this example $G^\up = \SL_2(\R)$ and 
\red{$G = G^\up \{e,\gamma\}$} with $\gamma^2 = -  \1$.

In general, elements in $\hat G_z^\down$ are of the form 
$g\hat\sigma$ with \red{$g \in G^\up$, and then}
\begin{equation}
  \label{eq:ztrivrel}
(g\hat\sigma)^2 = g \hat\sigma g \hat\sigma
 = g \sigma_G(g) \hat\sigma^2 = g \sigma_G(g) z.
\end{equation}
Hence $\hat G_z^\down$ contains an involution if and only if 
\[ z \in 
\red{\{ \sigma_G(g)^{-1} g^{-1} \: g \in G^\up\} }
=  \{ g \sigma_G(g) \: g \in G^\up\}.\] 
If $z = g \sigma_G(g)$ for some $g \in G^\up$, then 
conjugating with $g$ implies that $g$ and $\sigma_G(g)$ commute. 

The discussion in Example~\ref{ex:2.13} shows that 
\red{\eqref{eq:ztrivrel}  is not satisfied}
for $z = -\1$ and the Euler involution of $G^\up = \SL_2(\R)$. 
For any odd degree covering $\SL_2(\R)^{(2k+1)} \to \SL_2(\R)$,  
the central involution is mapped onto $-\1$, so that 
this \red{observation} carries over to odd coverings of $\SL_2(\R)$. 

The situation changes if we consider 
$G^\up = \SL_2(\red{\C})$ instead. Then 
$g := i \pmat{1 & 0 \\ 0 & -1}$ satisfies 
$g^2 = -\1$, so that the group  
\red{$\hat G  = G^\up \{ \1,\hat\sigma\}$} with 
$\hat\sigma^2 = - \1$ contains the non-trivial involution 
$g\hat\sigma \in \hat G^\down$. As this involution is central, 
$\hat G \cong \hat G^\up \times \Z_2$ is a direct product.
\end{rem}

\subsection{The abstract wedge space, some fundamental examples} 

\begin{defn} {\rm(The abstract wedge space)} 
{From here on, we always assume that $\cG \not=\eset$, i.e., that 
$G^\down$ contains an involution $\sigma$. Then 
\[ G \cong G^\up \rtimes \{\id,\sigma\} \]
(cf.\ Remark~\ref{rem:2.2}). }
For a fixed couple $W_0 = (h,\sigma) \in \cG$, the orbits 
\[ \cW_+ (W_0):= G^\up.W_0 \subeq \cG  \quad \mbox{ and } \quad 
\cW(W_0) := G.W_0 \subeq \cG  \] 
are called the {\it positive} and the {\it full wedge space containing $W_0$}. 
\end{defn}

\begin{rem} \mlabel{rem:stab}  
(a) As $\sigma.W_0 = (-h,\sigma) = W_0'$, we have 
$\cW(W_0) = \cW_+(W_0) \cup \cW_+(W_0')$, and 
$\cW(W_0)$ coincides with $\cW_+(W_0)$ if and only if 
$W_0' = (-h,\sigma) \in \cW_+(W_0)$. 
This is equivalent to the existence of an element 
$g \in G^\up$ with $g.W_0 = W_0'$, i.e., 
$g \in (G^\up)^\tau$ with $\Ad(g)h = - h$. 

\nin (b) If $W_0$ is an Euler couple, then $\cW(W_0)$ is a family of Euler 
couples, and we shall see below that in this case we have $\cW(W_0) = \cW_+(W_0)$ 
in many  important cases. 
\end{rem}

We collect some fundamental examples, starting from the low dimensional cases, that we shall refer to throughout the paper.

\begin{exs} \mlabel{ex:models} (a) The smallest 
example is the abelian group $G = \R \times \{\pm 1 \}$, 
where $G^\up = \R$, $C = \{0\}$ 
and $L = \g$. For $W_0 = (h,\sigma)$ with $h = 1$ and $\sigma= (0,1)$, we then have 
the one-point set $\cW_+ = \{ (h,\sigma)\}$, 
and $\cW = \{ (h,\sigma), (-h,\sigma)\}$. \\

\nin  (b) \textbf{The affine group} $G :=\Aff(\R)\cong \R \rtimes \R^\times$ 
of the real line is two-dimensional. 
Its elements are denoted $(b,a)$, and they act by 
$(b,a)x = ax + b$ on the real line.
The identity component 
$G^\up = \R \rtimes \R^\times_+$   
acts by orientation preserving maps, and $G^\down$ consists of 
reflections $r_p(x) = 2p-x$, $p \in \R$.\\
Let $\zeta(t) = (t,1)$ and $\delta(t) = (0,e^t)$ 
be the translation and dilation one-parameter groups, respectively. 
We write $\lambda = (0,1) \in \g = \R \rtimes \R$ for the infinitesimal 
generator of $\delta$, which is an Euler element. 
Therefore $W := (\lambda, r_0)$ is an Euler couple.

The cone $C = \R_+ \times \{0\} \subeq \g$ 
satisfies the invariance condition \eqref{eq:Cinv} 
and the corresponding semigroup $\cS_W$ 
is 
\[ \cS_W = [0,\infty) \rtimes \R^\times_+ 
= \{ g = (b,a) \: g.0 = b \geq 0\} 
= \{ g \in G^\up \: g\R_+ \subeq \R_+\}. \] 
Therefore the map 
\[ \cW_+(W) \ni {g.(\lambda, r_0) }\mapsto gW \] 
defines an order preserving 
bijection between the abstract wedge space $\cW_+(W) \subeq \cG$ 
and the set $\cI_+(\R) = \{ (t,\infty) \: t \in \R\}$ of 
of lower bounded open intervals in $\R$. Accordingly, we may write 
$W_{(t,\infty)}=(\Lambda_{(t,\infty)},r_t) 
:= \zeta(t) W = (\Ad(\zeta(t))\lambda, r_t)$ for 
$t \in \R$. 
Acting with reflections, we also obtain the couples 
\[ W_{(-\infty,t)}:=(\Lambda_{(-\infty,t)},r_t) = r_t.W_{(t,\infty)} = (-\Ad(\zeta(t))\lambda, r_t)\] 
corresponding to past pointing half-lines $(-\infty,t)\subset\R$. 
We thus obtain a bijection between the full wedge space 
$\cW(W)$ and the set $\cI(\R)$ of open semibounded intervals in~$\R$. 
We shall denote with $\delta_I$ the one-parameter group of dilations 
with generator $\lambda_I$ corresponding to the half line~$I$. 

The set $\cE(\g) = \Ad(G^\up)\{\pm \lambda\}$ 
of {non-zero} Euler elements in $\g$ consists of two $G^\up$-orbits 
and, for each non-zero Euler element 
$\pm\Ad(\zeta(t))\lambda \in \cE(\g)$, 
the reflection $r_t$ is the unique partner for which 
$(\pm\Ad(\zeta(t))\lambda, r_t) \in \cG$. Accordingly, 
Euler couples in $\cG$ are in one-to-one correspondence 
with semi-infinite open intervals in $\R$. 

\nin (c) 
\textbf{The M\"obius group} $G :=\Mob_2 := \PGL_2(\R)\cong \GL_2(\R)/\R^\times$ 
acts on the compactification $\oline\R = \R \cup \{\infty\}$ 
of the real line by 
\[ g.x := \frac{a x + b}{cx + d} \quad \mbox{ on } \quad 
 \oline\R:= \R \cup \{\infty\}, \qquad 
\mbox{ for }\quad g = \pmat{a & b \\ c & d}\in \GL_2(\R).\] 
We write $G^\up = \Mob \cong \PSL_2(\R)$ for the subgroup of orientation 
preserving maps. 
The Cayley transform 
\[ C:\overline\RR \to \bS^1 := \{z\in\CC: |z|=1\}, \quad 
C(x) := \frac{i-x}{i+x}, \qquad C(\infty) := -1, \] 
is a homeomorphism, identifying $\oline\R$ with the circle. 
Its inverse is the stereographic map 
\[ C^{-1}:\bS^1 \to \oline \R, \quad 
C^{-1}(z) = i\frac{1-z}{1+z}.\] 
It maps the upper semicircle $\{ z \in \bS^1 \: \Im z > 0\}$ 
to the positive half line $(0,+\infty)$. 
The Cayley transform intertwines the action of $\Mob$ on $\oline\R$ 
with the action of $\PSU_{1,1}(\C) =\SU_{1,1}(\C)/\{\pm \1\}$, given by  
\[  \pmat{\alpha & \beta \\ \oline\beta & \oline\alpha}.z 
:= \frac{\alpha z + \beta}{\oline\beta z + \oline\alpha}
\quad \mbox{ for } \quad z \in \bS^1, 
\pmat{\alpha & \beta \\ \oline\beta & \oline\alpha} 
\in \SU_{1,1}(\C).\] 

The three-dimensional Lie group $\Mob$ is 
generated by the following one-parameter subgroups:
 \begin{itemize}
 \item Rotations: $\rho(\theta)(x) = \frac{\cos(\theta/2)x + \sin(\theta/2)}
{-\sin(\theta/2)x + \cos(\theta/2)}$ for $\theta \in \R$; 
note that $C(\rho(\theta)x) = e^{i\theta} C(x)$. 
 \item Dilations: $\delta(t)(x)=  e^t x$ for $t \in \R$. 
 \item Translation: $\zeta(t) x = x+t$ for $t \in \R$. 
 \end{itemize}
In the circle picture  $\delta$ and $\zeta$ 
will be denoted by $\delta_{\cap}$ and $\zeta_\cap$, 
referring to the upper semicircle with endpoints $\{-1,1\} = 
C(\{0,\infty\})$.  
Note that $-1$ is the unique fixed point of $\zeta_\cap$ 
and one of the two fixed points $\{\pm 1\}$ of 
$\delta_\cap$. On the circle, 
$\rho(\pi)$ maps $1$ to $-1$ and exchanges the upper and the lower semicircle. 
Accordingly, $\zeta_\cup=\rho(\pi)\zeta\rho(\pi) $ is the subgroup 
of conjugated translations fixing the point $1\in \bS^1$.

We write $\textbf{K} = \rho(\R)$, $\textbf{A} = \delta(\R)$, 
$\bN^+ = \zeta(\R)$ and $\bN^- = \zeta_\cup(\R)$ 
for the corresponding one-dimensional subgroups of $\Mob$, 
and $\mbP^+ = \bA\bN^+ = \Mob_\infty$, $\mbP^- := \bA \bN^- = \Mob_0$ 
for the stabilizer groups of $\infty$ and $0$ in $\Mob$.
We observe that $\oline \R \cong \Mob/\mbP^-$ and that 
the circle group $K= \PSO_2(\R)$ acts simply transitively on $\oline\R$. 

On the compactified line, the point reflection $\tau(x) = -x$ in $0$ 
acts on the Lie algebra by 
\begin{equation}
  \label{eq:Adtau}
\Ad(\tau)\pmat{a & b \\ c & -a} 
= \pmat{-1 & 0 \\ 0 & 1}
\pmat{a & b \\ c & -a} 
\pmat{-1 & 0 \\ 0 & 1}
= \pmat{a & -b \\ -c & -a}.
\end{equation}
Note that $\tau \in G^\down$. 

The infinitesimal generator 
$h := \pmat{\frac{1}{2}  & 0 \\ 0 & -\frac{1}{2} }$ of $\delta$ 
is an Euler element and  $W := (h,\tau)$ is an Euler couple. 
Since $\Mob_2 \cong \PGL_2(\R) \cong \Aut(\fsl_2(\R))$, 
for any Euler couple $(x,\tau)$, the involution $\tau$ 
is determined by the requirement that it acts on $\g=\fsl_2(\R)$ 
by $e^{\pi i \ad x}$. We conclude that the action of 
$G^\up = \Mob$ on the set of Euler couples is transitive, i.e., 
$\cG_E = G^\up.(h,\tau)$. 

To see the geometric side of Euler couples, 
let us call a non-dense, non-empty open connected subset 
$I \subeq \bS^1$ an {\it interval} and 
write $\cI(\bS^1)$ for the set of intervals in $\bS^1$. 
It is easy to see that $\Mob$ acts transitively on $\cI(\bS^1)$. 
To determine the stabilizer of an interval, we consider 
the upper half circle, which corresponds to the half line 
$(0,\infty) \subeq \oline\R$. 
Each element $g \in \Mob$ mapping $(0,\infty)$ onto itself fixes 
$0$ and $\infty$. Since it is completely determined by the image 
of a third point, it is of the form $\delta(t)$ if $g.1 = e^t$. 
Therefore the stabilizer of $(0,\infty)$ in $\Mob$ is 
the subgroup $\delta(\R)$, which coincides with 
the stabilizer of $h$ under the adjoint action. 
This already shows that $\cW_+(W)$ and $\cI(\bS^1)$ 
are isomorphic homogeneous spaces of $\Mob$.
In particular, we can associate to an interval $I = g(0,\infty)$ 
the reflection $\tau_I=g \tau g^{-1}$ and the one-parameter group 
$\delta_I := g \delta g^{-1}$. 
Note that $\tau_I$ is an orientation reversing involution mapping 
$I$ to the complementary open interval~$I'$. 
We write $x_I := \Ad(g)h$ for the infinitesimal generator 
of $\delta_I$, so that the assignment 
$I \mapsto x_I$ defines an equivariant
bijection $\cI(\bS^1) \to \cE(\g)$. 
The  anticlockwise orientation of $\bS^1$, which 
can also be considered as a causal structure, is used here to 
pick the sign of $x_I$ in such a way that the flow $\delta_I$ 
is counter clockwise (future pointing) on $I$. Accordingly, 
$x_{I'} = -x_I$ corresponds to the complementary interval~$I'$.

To identify the natural order on the abstract wedge space $\cG_E = \cW_+(W)$, 
we consider for $X = \pmat{a & b \\ c & -a} \in \g = \fsl_2(\R)$ 
the corresponding fundamental vector field 
\[ {V_X}(x) = \frac{d}{dt}\Big|_{t = 0} \exp(tX).x 
= (a - d)x + b - c x^2 = b + 2 ax - c x^2.\] 
This shows that 
\begin{equation}\label{eq:cone} C := \{ X \in \g \: {V_X} \geq 0\} 
= \Big\{  X= \pmat{a & b \\ c & -a} \: 
b \geq 0, c \leq 0, a^2 \leq -bc\Big\}
\end{equation}
is a pointed generating invariant cone in $\g$. 
The Lie wedge specified by the triple $(h,\tau, C)$ is 
\[ L_W = L(h,\tau,C) 
= \underbrace{\R_+ \pmat{0 & 1 \\ 0 & 0}}_{C_+} \oplus \R h \oplus 
\underbrace{\R_+ \pmat{0 & 0 \\ -1 & 0}}_{C_-} 
= \Big\{ \pmat{a & b \\ c & -a} \: a \in \R, b\geq 0, c\leq 0\Big\}.\] 
We further have $G(W) = G^\up$, and the associated semigroup is 
\[ \cS_W = \exp(C_+) \exp(\R h) \exp(C_-) 
= \{ g \in G^\up \: g(0,\infty) \subeq (0,\infty) \}.\]
Therefore the map 
\begin{equation}
  \label{eq:iso1}
\cG_E =  \cW_+(W) = \cW(W) \to \cI(\bS^1), \quad 
g.W \mapsto g (0,\infty) 
\end{equation}
defines an order preserving  bijection between the 
abstract wedge space $\cW(W)$ and the ordered set~$\cI(\bS^1)$. \\

\nin (d) We now consider \textbf{the universal covering of the M\"obius group}~$\Mob$. 
Concretely, we put $G := \tilde\Mob \rtimes \{\1,\tilde\tau\}$, where 
$\tilde\tau$ acts on $\tilde\Mob$ by integrating $\Ad(\tau)$ from 
\eqref{eq:Adtau} to an automorphism of $\tilde\Mob$. 
The group $G$ is a graded Lie group and $G^\up := \tilde\Mob$ 
is its identity component. 
We have a covering homomorphism 
$q_G: G \to \Mob_2$ whose kernel $Z(\tilde\Mob) \cong \Z$ is discrete cyclic. 
We write $\tilde \rho$, $\tilde \delta$, $\tilde\zeta$  and $\tilde \zeta_\cup$ 
for the canonical lifts of the one-parameter groups 
$\rho$, $\delta$, $\zeta$, $\zeta_\cup$ of $\Mob$, 
$\tilde \mbP^+ := \tilde\delta(\R) \tilde\zeta(\R)$, 
and $\tilde \mbP^- := \tilde\delta(\R) \tilde\zeta_\cup(\R)$. 

The action of $\Mob$ on $\bS^1$ lifts canonically 
to an action of the connected group 
$G^\up = \tilde\Mob$ on the universal covering $\tilde{\bS^1} \cong \R$, 
where we fix the covering map 
$q_{\bS^1} \: \R \to \tilde{\oline\R}$, defined by 
$q_{\bS^1}(\theta) = \tilde\rho(\theta).0$, 
which corresponds to the map $\theta \mapsto e^{i\theta} 
= C(\tilde\rho(\theta).0)$ in the circle picture. 
We may thus identify $\tilde{\bS^1}$ with the homogeneous space 
$\tilde\Mob/\tilde\mbP^-\cong\R$. As conjugation with $\tilde\tau$ on 
$\tilde\Mob$ preserves the subgroup $\tilde\mbP^-$, it also acts on 
$\tilde{\bS^1}$. From \eqref{eq:Adtau} it follows that 
it simply acts by the point reflection $\tilde\tau.x = - x$ in the base 
point~$0$. We also note that $Z := \ker(q_G)  = \tilde\rho(2\pi \Z)$ is 
the group of deck transformations of the covering 
$q_{\bS^1}$, which acts by 
\begin{equation}
  \label{eq:deck}
 \tilde\rho(2\pi n).x = x + 2 \pi n \quad \mbox{ for } \quad n \in \Z.
\end{equation}

We call  a non-empty interval $I \subeq \R$ {\it admissible} 
if its length is strictly smaller than $2\pi$ 
and write $\cI(\R)$ for the set of admissible intervals. 
An interval $I \subeq \R$ is admissible if and only if 
there exists an interval $\uline I \in \cI(\bS^1)$ 
such that $I$ is a connected component of $q_{\bS^1}^{-1}(\uline I)$. 
The group $Z$ 
acts transitively on the set of these connected components. 
As $\Mob$ acts transitively on $\cI(\bS^1)$, 
it follows that the group $\tilde\Mob$ acts transitively on the set 
$\cI(\R)$, and that 
composition with $q_{\bS^1}$ yields an equivariant covering map 
\begin{equation}
  \label{eq:iso2a}
\cI(\R) \cong \tilde\Mob/\tilde\delta(\R) \to  \cI(\bS^1)
\cong \Mob/\delta(\R), \quad 
I \mapsto q_{\bS^1}(I). 
\end{equation}

We further have: 
\begin{itemize}
 \item The group $\tilde\mbP^+ = \tilde\delta(\R) \tilde\zeta(\R)$ 
fixes the points $\{(2k+1)\pi\: k\in\Z\}$. 
\item For $I \in \cI(\bS^1)$, let 
$\tilde\delta_I$ be the lift of the one-parameter group $\delta_I$.  
Then $\tilde\delta_I$ 
preserves every interval in the preimage $q_{\bS^1}^{-1}(I)$. 
\item The inverse images of $\tau \in \Mob_2$ 
in $\tilde\Mob_2$ 
are the elements $\tilde\tau_n := \tilde\rho(2\pi n)\tilde\tau$, 
$n \in \Z$. These are involutions, acting by 
\begin{equation}
  \label{eq:taun-act}
  \tilde\tau_n(x)= 2\pi n -x \quad \mbox{ for } \quad x\in\R
\end{equation}
which is a point reflection in the point $\pi n$.   
All pairs $(h, \tilde\tau_n)$ are Euler couples in $\cG(\tilde\Mob_2)$, 
and from the discussion of the set of Euler couples 
$\cG_E(\Mob_2)$ under (c), we know that the involutions 
$\tilde\tau_n$ exhaust all possibilities for supplementing $h$ to 
an Euler couple. 

There is an interesting difference to the situation for 
$\Mob_2$, where $\Mob$ acts transitively on the set $\cG_E(\Mob_2)$ 
of Euler couples. To see what happens for $\tilde\Mob_2$, 
recall that the stabilizer of the element $(h,\tau)\in \cG_E(\Mob_2)$ 
in $\Mob$ is the subgroup $\delta(\R)$. 
Its inverse image is the group 
\[ \tilde\delta(\R) \tilde\rho(2\pi \Z) \cong \R \times \Z.\] 

An element $g \in \tilde\Mob$ fixes $(h,\tilde\tau_n)$ if and only if 
$\Ad(g)h = h$ and $g \tilde\tau_n g^{-1} = \tilde\tau_n$. 
The first condition is equivalent to $g$ being of the form 
\[ g = \tilde\delta(t) \tilde\rho(2\pi k)
\quad \mbox{  for some } \quad t\in \R, k \in \Z.\] 
The second condition is equivalent to 
$\tilde\tau g \tilde\tau = \tilde\tau_n g \tilde\tau_n = g$, 
which takes the form 
\[ \tilde\delta(t) \tilde\rho(-2\pi k) = \tilde\delta(t) \tilde\rho(2\pi k),\] 
and this is equivalent to $k = 0$. We conclude that the stabilizer 
of $(h,\tilde\tau_n)$ is 
\begin{equation}
  \label{eq:stab-couple}
\tilde\Mob_{(h,\tilde\tau_n)} = \tilde\delta(\R).
\end{equation}
We also note that 
\[ \tilde\rho(\pi k).(h,\tilde\tau_n)
= ( (-1)^k h,  \tilde\rho(\pi k)\tilde\tau_n\tilde\rho(-\pi k))
= ((-1)^k h,  \tilde\rho(2\pi k)\tilde\tau_n)
= ((-1)^k h,  \tilde\tau_{n+k}).\]
We conclude that the group $\tilde\Mob$ does not act transitively 
on the set $\cG_E$ of Euler couples. It has two orbits: 
\begin{equation}
  \label{eq:twoorb}
\cG_E(\tilde\Mob_2) = G^\up.W_0 \dot\cup G^\up.{W_1}
= \cW_{+}(W_0) \dot\cup \cW_{+}({W_1}) 
\quad \mbox{ for } \quad 
W_0 := (h,\tilde\tau_0), W_1 := (h,\tilde\tau_1).
\end{equation}
We also refer to Example~\ref{ex:halfnet} for a discussion of 
this issue from a different perspective. 

\item The subgroup $\tilde\delta(\R)$ preserves every interval which is a non-trivial 
orbit of $\tilde\delta(\R)$, acting on $\R$. If, conversely, $g \in \tilde\Mob$ 
preserves such an interval, then its image in $\Mob$ is contained 
in $\delta(\R)$, so that 
\[ g = \tilde\delta(t) \tilde\rho(2\pi k)\quad \mbox{  for some } \quad t\in \R, 
k \in \Z.\] 
As every open orbit of $\tilde\delta(\R)$ is an interval of length 
$\pi$, the element $g$ can only preserve such an orbit if $k = 0$. 
This shows that $\tilde\Mob_{(h,\tilde\tau_n)}$ also is the stabilizer 
group of any open $\tilde\delta(\R)$-orbit in $\R$. 
We conclude that, for the Euler couple $W_0 = (h, \tilde\tau_0)$, the map 
\begin{equation}
{\Phi} \:  \cW_+(W_0) \to \cI(\R), \quad 
g.(h,\tilde\tau_0) \mapsto g (0,\pi) 
\end{equation}
defines a $G^\up$-equivariant  bijection between the abstract 
wedge space $\cW_+(W_0) \subeq \cG$ 
and the  set $\cI(\R)$ of admissible intervals in~$\R$.
Since the full group $G$ acts on the space $\cI(\R)$ of intervals, 
$\Phi$ can be used to transport this action to a $G$-action 
 on the space $\cW_+(W_0)$, extending 
the action of the subgroup $G^\up$. 
Since 
$\tau_0(0,\pi)=(-\pi,0)=\rho(-\pi)(0,\pi),$ we have 
$$\Phi^{-1}(\tau_0(0,\pi))=\Phi^{-1}(\rho(-\pi)(0,\pi))
=\rho(-\pi).\Phi((0,\pi))^{-1}=(-h,\rho(-2\pi)\tau_0),$$ 
so that $\tau_0.W_0:=(-h,\rho(-2\pi)\tau_0)$. 
By $G^\up$-equivariance of the map $\Phi$, we conclude that the action
of $G^\down$ on $\cW_+(W_0)$ is given by 
\begin{equation}\label{eq:left} g *_{\rho(-2\pi)} (x,\sigma) := (\Ad^\eps(g), \rho(-2\pi) g \sigma g^{-1}) 
\quad \mbox{ for every } \quad g \in G^\down.\end{equation}
Here we use that $\tilde\rho(-2\pi) \in Z(G^\up)$.
Note that we have chosen $(0,\pi)$ to be the image of $W_0$ throught $\Phi$. Further possible actions come from the identifications 
\begin{equation}
{\Phi_n} \:  \cW_+(W_n) \to \cI(\R), \quad 
g.(h,\tilde\tau_n) \mapsto g (0,\pi) 
\quad \mbox{ with } \quad 
W_n=(h,\tau_n), \end{equation}
and one can likewise  see that
\[  g *_{\alpha_n} (x,\sigma) := (\Ad^\eps(g), \alpha_n g \sigma g^{-1}) 
\quad \mbox{ for } \quad g \in G^\down \quad \mbox{ and } \quad 
\alpha_n = \tilde\rho((2n-1)2\pi) \in Z(G^\up), \] 
extends the action of $G^\up$ on $\cW_+(W_n)$ to $G$ and $\Phi=\Phi_0$ for $n=0$ 
(see also \eqref{eq:twistact} and Section~\ref{sect:alftwist} for this kind of action).
\end{itemize}

\nin (e) Let $q \: \Mob^{(n)} \to \Mob$ 
be \textbf{the $n$-fold covering group} of 
$\Mob$ and 
$\rho^{(n)}, \delta^{(n)}, \zeta^{(n)}$ and 
$\zeta_{\cup}^{(n)}$ be the lifts of the corresponding one-parameter groups 
of $\Mob$. We further put $\mbP^{-,(n)} := \delta^{(n)}(\R) \zeta_{\cup}^{(n)}(\R)$, 
so that we obtain an $n$-fold covering 
\[ q_n \: \bS^1_n := \Mob^{(n)}/\mbP^{-,(n)}  \to \bS^1= \Mob/\mbP^{-}, \quad 
g \mbP^{-,(n)} \mapsto q(g) \mbP^{-} \] 
of the circle, and the action of the one-parameter group 
$\rho^{(n)}$ induces a diffeomorphism 
\[ \R/2\pi n \Z \to \bS^1_n, \quad 
[t] \mapsto \rho^{(n)}(t).0 \] 
The set of wedges can be described analogously 
to the case (d), but there is a difference depending on the parity 
of~$n$. If $n$ is even, the group $G^\up$ has two orbits in the set 
$\cG_E$ of Euler couples, but if $n$ is odd, there is only one. 
Indeed, for $n = 2k$, the element $\rho^{(n)}(2\pi k)$ acts as an involution 
on $\bS^1_n$. So it fixes all Euler 
couples $(h,\tilde\tau_n)$, even if it does NOT fix any proper  interval 
 in $\bS^1_n$ (see also Example~\ref{ex:halfnet}). \\

\nin (f) The example arising most prominently in physics is the 
proper \textbf{Poincar\'e group }
\[ G := \cP_+ := \R^{1,d} \rtimes \SO_{1,d}(\R), \qquad 
G^\up := \cP_+^\up := \R^{1,d} \rtimes \SO_{1,d}(\R)^\up.\] 
It acts on $1+d$-dimensional Minkowski space $\R^{1,d}$ 
as an isometry group of the Lorentzian metric given by 
$(x,y)= x_0y_0-\bx\by$ for $x = (x_0, \bx) \in \R^{1,d}$. 
Writing 
\[ V_+ := \{ (x_0, \bx) \in \R^{1,d} \: x_0 > 0, x_0^2 > \bx^2\}\] 
for the open future light cone, the 
grading on $G$ is specified by time reversal, i.e., 
$gV_+ = \eps(x,g) V_+$. In particular $C := \oline{V_+}$ is a pointed 
closed convex cone satisfying \eqref{eq:Cinv}. 
For $d > 1$, this is, up to sign, the only non-zero pointed 
invariant cone in the Lie algebra $\g$. 

The generator $k_{1} \in \so_{1,d}(\R)$ of the Lorentz boost on the 
$(x_0,x_1)$-plane 
\[  k_1(x_0,x_1,x_2, \ldots, x_{d}) = (x_1, x_0, x_2, \ldots, x_{d})\] 
is an Euler element. It combines with the spacetime 
reflection $j_1(x) =(-x_0,-x_1,x_2,\ldots,x_d)$ 
to the Euler couple $(k_1, j_1)$. 
We associate to $(k_1, j_1)$ the spacetime region 
\[ W_1=\{x\in\RR^{1+d}: |x_0|<x_1\}, \] 
the {\it standard right wedge}, 
 and note that $W_1$ is invariant under $\exp(\R k_{1})$. 
It turns out that the semigroup $\cS_{( k_1,j_1)}$ associated to the 
couple $( k_1,j_1)$ in Definition~\ref{def:abs-struc} satisfies 
\begin{equation}
  \label{eq:semiw1}
\cS_{( k_1, j_1)} 
= \{ g \in G \: gW_1 \subeq W_1\} =: \cS_{W_1}
\end{equation}
(see \cite[Lemma~4.12]{NO17}). 
From \eqref{eq:semiw1} it follows that the map 
\begin{equation}
  \label{eq:abs-conc-wedge-poincare}
\cW_+ = \cW = G^\up.(k_1, j_1) \ni g.( k_1, j_1) \mapsto gW_1 
\end{equation}
defines an order preserving 
 bijection between the abstract wedge space $\cW \subeq \cG$ 
and the set of wedge domains in Minkowski space~$\R^{1+d}$. 
For an abstract wedge $W = ( k_W, j_W) \in \cW$, 
the Euler element $k_W$ is the corresponding boost generator. 
For an axial wedge $W_i:=\{x\in\RR^{1+d}:|x_0|<x_i\}$, $i =1,\ldots, n$, 
the corresponding Euler couple will be denoted $(k_i,j_i)$.
\end{exs}

\subsection{Nets of wedges, isotony, central locality and covering groups}

In the following sections we will focus on the description of 
relative positions of  wedges, 
in particular wedge inclusions and the locality principle.
\subsubsection{Wedge inclusion}\label{sect:wedgeinc}

Firstly consider this wedge inclusion configuration called \textit{half-sided modular inclusion}:
\begin{defn}
Let $W_0 = (x,\sigma) \in \cG$ and 
$y \in \pm C$ with $[x,y] = \pm y$. Then $\exp(y) \in \cS_{W_0}$ 
(Definition~\ref{def:abs-struc}(b)), so that 
\[ W_1  := \exp(y).W_0 \leq W_0.\] 
We then call $W_1 \leq W_0$ a {\it $\pm$half-sided modular inclusion}. 
\end{defn}

The next lemma shows that any wedge inclusion can be described in terms of positive and negative half-sided modular inclusions. 

\begin{lem}\label{lem:isoW} If $W_1 \leq W_3$ in $\cG$, then there exists an 
element $W_2 \in \cG$ with $W_1 \leq W_2 \leq W_3$ for which 
the inclusion $W_1 \leq W_2$ is $+$half-sided modular and 
the inclusion $W_2 \leq W_3$ is $-$half-sided modular. 
\end{lem}

\begin{prf} That $W_1 \leq W_3$ means that $W_1 = s W_3$ for some 
\[ s \in \cS_{W_3} = \exp(C_-(W_3)) \exp(C_+(W_3)) G^\up_{W_3}.\] 
Accordingly, we write $s = g_- g_+ g_0$ and 
observe that $W_1 = g_- g_+ W_3$ because $g_0 W_3 = W_3$. 
Put $W_2 := g_- W_3$. Then $W_2 \leq W_3$ and $g_+ W_3 \leq W_3$ 
implies $W_1 = g_- g_+ W_3 \leq g_- W_3 = W_2$. 

Further, the inclusion $W_2 \leq W_3$ is $-$half-sided modular 
because $g_- \in \exp(C_-(W_3))$. 
Likewise the inclusion $g_+ W_3 \leq W_3$ is $+$half-sided modular, 
and therefore $W_1 \leq W_2$ is also $+$half-sided modular.
\end{prf}

\subsubsection{Central locality}\label{sect:alftwist}

For a wedge $W = (x,\sigma)$, the dual wedge 
$W' = (-x,\sigma)$ need not be contained in the orbit 
$\cW_+ = G^\up.W$. If, however, $G^\up$ has a non-trivial central
subgroup $Z$ such that, modulo $Z$, the complement $W'$ is contained in $\cW_+$, 
then we use central elements $\alpha \in Z$ 
to define ``twisted complements'' $W^{'\alpha}$ which are contained in $\cW_+$, 
and this in turn leads to a twisted action of the full group $G$ on $\cW_+$. 
We also obtain on $\cW_+$ a complementation map $W \mapsto W^{'\alpha}$.

{Let $Z \subeq Z(G^\up)$ 
be a closed normal subgroup of $G$,  
and $q \: G \to \uG := G/Z$  be the corresponding surjective morphism of graded Lie groups 
with kernel $Z$. If $Z$ is discrete, then  $q$ is a covering map. 
The morphism of graded Lie groups $q$ induces a natural map 
{
  \begin{equation}
    \label{eq:cge}
  q_\cG \: \cG(  G) \to \uline\cG 
:= \{ (x,\uline\sigma)\in\g \times \uG^\down \: \uline\sigma^2 = e, 
\Ad_\g(\uline\sigma)x = x\},
\quad (x,\sigma) \mapsto (x, q(\sigma)), 
  \end{equation}
where $\Ad_\g \: \uG \to  \Aut(\g)$ denotes the factorized 
adjoint action which exists because $Z = \ker(q)$ acts trivially on $\g$.
It restricts to a map 
\begin{equation}
  \label{eq:cge2}
 \cG_E(G) \to \uline\cG_E 
:= \{ (x,\uline\sigma)\in \cE(\g) \times \uG^\down \: \uline\sigma^2 = e, 
\Ad_\g(\uline\sigma) = e^{\pi i \ad x}\}.
\end{equation}
} As the following example shows, neither of these maps is always surjective.
The main obstruction is that, 
although the differential $\L(q)\: \L(G) \to \L(\uG)$ is surjective, 
there may be involutions $\tau \in \uG^\down$ for which 
no involution $\sigma \in   G^\down$ with $q(\sigma) = \tau$ exists. 
This phenomenon is tightly related to the twisted groups $\hat G_z$ discussed 
in Remark~\ref{rem:twists-grad} because these twists disappear 
for $z \in Z$ in $\hat G/Z \cong G/Z$.}

\begin{ex} \mlabel{ex:2.13}
We consider the graded Lie group 
\[ G := \SL_2(\R)
 \{ \1,\gamma \} \subeq \SL_2(\C), \quad \mbox{ where } \quad 
\gamma := \pmat{ i & 0 \\ 0 & -i} 
\quad \mbox{ satisfies} \quad \gamma^2 = - \1. \] 
It has  two connected component and $G^\up = \SL_2(\R)$.\begin{footnote}
{This is the twisted version of the $2$-fold cover of the extended 
M\"obius groups; see Remark~\ref{rem:twists-grad}(d).}
\end{footnote}
The subgroup $Z := \{\pm \1\}$ is central and the quotient map 
$q \:   G \to \uG :=   G/Z$ 
is a $2$-fold covering. The Euler element 
$x := \frac{1}{2}\pmat{1 & 0 \\ 0 & -1} \in \g = \fsl_2(\R)$ 
combines with the involution $q(\gamma) \in \uG^\down$ to 
the Euler couple $(x,q(\gamma)) \in \uline\cG$.
However, the set $\cG(G)$ is empty because $G^\down$ contains no involution. 
In fact, for $g = \pmat{a & b \\ c & d} \in \SL_2(\R)$, the condition 
that $g\gamma$ is an involution is equivalent to 
\[ \pmat{-a & b \\ c & -d} = \gamma g \gamma = g^{-1} 
= \pmat{ d & -b \\ -c & a}. \] 
This is equivalent to $a = -d$ and $b = c = 0$, contradicting 
that $1 = \det(g) = -a^2$. 
We conclude in particular that the maps 
$\cG(G) \to \uline\cG$ and 
$\cG_E(G) \to \cG_E(\uG)$ are not surjective. 
\end{ex}

We now discuss $G^\up$-orbits in $\cG(G)$. 
{In the examples we have in mind, the central subgroup $Z$ is discrete.} 

\nin\textbf{Involution lifts and central wedge orbit.} 
Each element $\sigma \in  G^\down$ acts in the same way on 
the abelian normal subgroup $ Z$ by the involution 
\[ \sigma_ Z \:  Z \to  Z, \quad \gamma \mapsto \gamma^\sigma
:= \sigma \gamma \sigma\] 
which restricts to an involution $\sigma_ Z \in \Aut( Z)$ 
because $ Z$ is central in $G^\up$ and a normal subgroup of~$G$. 
In the following we shall need the subgroups 
\begin{equation}
  \label{eq:gammasubgrps}
 Z^- := \{ \gamma \in  Z\: \gamma^\sigma = \gamma^{-1}\}
\supeq    Z_1 := \{\gamma^\sigma\gamma^{-1} \: \gamma\in  Z\}.
\end{equation} 
For $\gamma \in  Z^-$, the element $\gamma^2 = (\gamma^\sigma\gamma^{-1})^{-1}$ 
is contained in $ Z_1$, so that the quotient group 
$ Z^-/ Z_1$ 
is an elementary abelian $2$-group, i.e., isomorphic to 
$\Z_2^{(B)}$ for some index set $B$. 

For an involution $\sigma \in   G^\down$ and $\beta \in Z(G^\up)$, the element 
$\beta\sigma \in G^\down$  is an involution if and only if 
$\beta \in  Z^-$. Therefore 
\begin{equation}
  \label{eq:*act}
\alpha*(x,\sigma):= (x,\alpha\sigma) 
\end{equation}
defines an action of $ Z^-$ on $\cG(G)$, commuting with the conjugation action 
of $G^\up$ and satisfying  
  \begin{equation}
    \label{eq:gdowneq}
g.(\alpha * (x,\sigma)) 
= \alpha^{-1} * (g.(x,\sigma)) \quad \mbox{ for }\quad 
g \in G^\down, \alpha \in Z^-.
  \end{equation}
For $W = (x,\sigma) \in \cG(  G)$, the fiber over 
$\uW := (x,q(\sigma))$ is thus given by 
\begin{equation}
  \label{eq:fib}
 Z^-* W:=  \{ (x,\alpha\sigma) \: \alpha \in  Z^- \}.
\end{equation}
The subgroup $ Z\subeq   G^\up$ acts by conjugation on the fiber $Z^-*W$: 
\[ \gamma.(x,\sigma) 
= (x, \gamma \sigma \gamma^{-1}) 
= (x, \gamma (\gamma^\sigma)^{-1} \sigma),\] 
so that the quotient group $ Z^-/ Z_1$ 
parametrizes the $ Z$-conjugation orbits in the fiber $ Z^-*W$.
\begin{footnote}{
Considering $ Z$ as a module $\Z_2$-module 
via the involution $\sigma_ Z$, we have 
$Z^1(\Z_2,  Z)\cong  Z^-$ and 
$B^1(\Z_2,  Z) \cong  Z_1,$ 
so that the cohomology group is 
$H^1(\Z_2,  Z) := 
Z^1(\Z_2,  Z)/B^1(\Z_2,  Z) \cong  Z^-/ Z_1$. 
We refer to \cite[Ex.~18.3.15]{HN12} or 
\cite[Thm.\ IV.7.1]{ML63} for more on group cohomology.
}\end{footnote}
Here is an example.
\begin{ex} \label{ex:halfnet}(a) If $ Z \cong \Z$ and $n^\sigma = -n$, then 
$ Z^- = \Z$ and $ Z_1 = 2 \Z$, so that 
$ Z^-/ Z_1 \cong \Z/2\Z$. 

\nin (b) If $ Z = \Z_n$ and $\oline n^\sigma = -\oline n$, then 
$ Z^- = \Z_n$ and $ Z_1 = 2 \Z_n$, so that 
\[   Z^-/ Z_1\cong
\begin{cases}
\Z/2\Z &\text{ if } n \ \text{ is } \text{ even}  \\ 
\{0\} &\text{ if } n \ \text{ is } \text{ odd.}
\end{cases} \] 
\end{ex}

\nin\textbf{Wedge $G^\up$-orbits.} 
{Let $W = (x,\sigma) \in \cG_E(G)$ and $\uline W = (x,q(\sigma)) 
\in \uline\cG$.}
{In general the group $  G^\up$ does not act transitively on 
the inverse image of the  orbit $\ucW_+ := 
\uG^\up.\uW\subeq  \uline\cG$ under $q_\cG$. 
We now describe how this set decomposes into orbits.}
By the transitivity of the $\uG^\up$-action on $\ucW_{+}$, it suffices to consider the 
orbits of the stabilizer 
\[ G^\up_{\uW}=\{g\in G^\up: q(g).\underline W=\underline W\} \] 
on the fiber $Z^- * W$. 
That $g \in   G^\up$ fixes $\uW$ implies 
in particular that {$g\sigma g^{-1}\sigma  = g (g^\sigma)^{-1} \in  Z$.}
This leads to a homomorphism 
\begin{equation}
  \label{eq:def-delta}
\partial \:   G^\up_\uW \to  Z^-, \quad 
g \mapsto g (g^\sigma )^{-1}
\quad \mbox{ with }\quad g.(x,\sigma) 
= (\Ad(g)x, g\sigma g^{-1}) 
= (x, \partial(g)\sigma) .
\end{equation}
As $ Z \subeq   G^\up_{\uW}$, the image $ Z_2 := \partial(  G^\up_{\uW})$ 
is a subgroup containing $ Z_1$.

\begin{ex} 
(An example where $Z_1 \not= Z_2$) We consider the group 
$G = \tilde\Mob \rtimes \{\1,\tilde\tau\}$ from Example~\ref{ex:models}(d) 
and the canonical homomorphism 
\[ q \: G \to \uline G := \SL_2(\R) \rtimes \{\1,\sigma\}, \quad 
\sigma := \pmat{-1 & 0 \\ 0 & 1}\] 
whose  kernel is the central subgroup $Z := 2 Z(G^\up) 
\subeq Z(G^\up) \cong \Z$ of 
index two. 
Now $W = (h,\tilde\tau) \in \cG(G)$ is an Euler couple mapped to 
$\uline W = (h,\sigma) \in \uline\cG$. As $z^{\tilde \tau} = z^{-1}$ for 
every $z \in Z$, we have $Z = Z^-$ and $Z_1 = 2Z$ is a subgroup 
of index~$2$. To calculate $Z_2$, we observe that 
\[ \uG^\up_{\uline W} = \uG_h = \exp(\R h) \{ \pm \1\} \quad \mbox{ and } \quad 
G^\up_{\uline W} = \exp(\R h) Z(G^\up).\] 
We conclude that 
\[ Z_2 = \partial\big(G^\up_{\oline W}\big) = \partial(Z(G^\up)) =  
2 Z(G^\up) = Z^- \not= Z_1.\] 

The situation changes if we consider $Z=Z(G^\up)$ and the center-free 
group $\uG = \Mob \rtimes \{\1,\tau\}$ 
instead. Then $Z = Z^- = Z(G^\up)$ and $Z_1 = Z_2 = 2 Z$.
\end{ex}

As the $G^\up$ orbits in 
$q_\cG^{-1}(\uline G^\up.\uline W) = q_\cG^{-1}(\uline \cW_+)$ 
correspond to the $G^\up_{\uline W}$-orbits in the fiber 
$q_\cG^{-1}(\uline W) = Z^- * W$, we obtain the following lemma.

\begin{lemma} The quotient group $Z^-/Z_2$ parametrizes the set of 
$G^\up$-orbits in $q_\cG^{-1}(\ucW_{+})$. 
\end{lemma}

\nin{\textbf{$\alpha$-twisted complement.} The following definition 
generalizes the notion of complementary wedge given in Definition \ref{def:abs-struc} (a).}
\begin{defn} For $\alpha \in{ Z^-}$, we define the 
{\it $\alpha$-twisted  complement}  of 
$W = (x,\sigma) \in \cG(  G)$ by 
\[ (x,\sigma)^{'\alpha} := {(-x,\alpha\sigma)}.\] 
\end{defn}
We will refer to  couples of the form $W^{'\alpha}$ as 
complementary wedges. 
We consider $W^{'\alpha}$ as a ``complement'' of $W$ 
{because $q_\cG$} maps $W^{'\alpha}$ to $W'$ 
(see item (a) below).

\begin{lem}\mlabel{lem:twistact}
For each $\alpha \in Z^-$, the 
 $\alpha$-twisted complementation $W \mapsto W^{'\alpha}$ 
satisfies: 
\begin{itemize}
\item[\rm(a)] For $\alpha \in Z^-$, $W^{'\alpha}$ is mapped by $q_\cG$ onto the complement $W' = (-x, q(\sigma))$ of $\uW = (x,q(\sigma))$. 
\item[\rm(b)] The $\alpha$-twisted complementation is not involutive if $\alpha^2 \not=e$. 
\item[\rm(c)] The map ${}^{'\alpha} \: \cG(  G) \to \cG(  G), 
(x,\sigma) \mapsto (-x, \alpha\sigma)$ 
is $  G^\up$-equivariant.
\item[\rm(d)] In terms of the action \eqref{eq:*act} of $Z^-$ 
on $\cG(G)$, we have
\begin{equation}
  \label{eq:compl-alphaact}
 W^{'\alpha} = \alpha * W' \quad \mbox{ for } \quad 
W \in \cG(G), \alpha \in Z^-.
\end{equation}
\item[\rm(e)] The prescription 
\begin{equation} \label{eq:twistact}
 g *_\alpha (x,\sigma) 
:=\begin{cases}
g.(x,\sigma) & \text{ for } g \in G^\up \\ 
g.(\alpha^{-1} * (x,\sigma)) 
= \alpha * (g.(x,\sigma)) & \text{ for } g \in G^\down.
\end{cases}
\end{equation}
defines an action of $G$ on $\cG(G)$. 
This action  satisfies 
\begin{equation}
  \label{eq:alpha-act2}
 W^{'\alpha} = \sigma *_\alpha W \quad \mbox{ for } \quad 
W  = (x,\sigma) \in \cG(G), \alpha \in  Z^-.
  \end{equation}
If $W^{'\alpha}\in G^\up.W$, then $\cW_+ = G^\up.W$ is invariant 
under the full group $G$ with respect to the $\alpha$-twisted 
action. 
\item[\rm(f)] There exists an $\alpha \in Z^-$ with 
$W^{'\alpha} \in \cW_+$ if and only if $\uline W' := (-x,q(\sigma)) 
\in G^\up.\uline W$. If this is the case, then 
$W^{'\beta} \in \cW_+$ for $\beta \in Z^-$ if and only if 
$\beta^{-1}\alpha \in Z_2$. 
{In this case, the twisted actions of $g \in G^\down$ are related by 
$g*_\beta = (\beta \alpha^{-1}) * g *_\alpha$.}
\end{itemize}
\end{lem}

\begin{prf} (a) and (b) are easy to see. \\
\nin (c) follows from $\alpha \in Z(  G^\up)$ 
and the $G^\up$-equivariance of the complementation map. \\
\nin (d) {is immediate from the definition of 
$\alpha * W$.}\\
\nin (e) That the prescription defines an action follows easily 
from the fact that 
$g_1 *_\alpha (g_2 *_\alpha W)= (g_1 g_2).W$ for $g_1, g_2 \in G^\down$ 
(cf.\ \eqref{eq:gdowneq})). 
The relation \eqref{eq:alpha-act2} follows from 
$\sigma.W = \sigma.(x,\sigma) = (-x,\sigma)$. 
For the last statement, we note that by \eqref{eq:alpha-act2}, 
the relation $W^{'\alpha} \in \cW_+$ implies 
\[ G *_\alpha \cW_+ = \cW_+ \cup \sigma *_\alpha \cW_+ 
=  \cW_+ \cup G^\up.W^{'\alpha} 
=  \cW_+ \cup G^\up \cW_+ = \cW_+.\]
\nin (f) As $q_\cG(\cW_+) = \uline\cW_+ = G^\up.\uline W$ 
and $q_\cG(W^{'\alpha}) = \uline W'$, the inclusion 
$W^{'\alpha} \in \cW_+$ implies that $\uline W' \in \uline\cW_+$. 
If, conversely, $\uline W' \in \uline\cW_+$, then there exists a 
$g \in G^\up$ with 
\[ (-x, q(\sigma)) = g.(x,q(\sigma)) = (\Ad(g)x, q(g\sigma g^{-1})), \] 
so that $\alpha := g\sigma g^{-1} \sigma \in \ker(q) = Z$ satisfies 
\[ \cW_+ \ni g.W = g.(x,\sigma) 
= (-x, g\sigma g^{-1})
= (-x, \alpha \sigma) = \alpha * W' = W^{'\alpha}.\] 
Now suppose that $W^{'\alpha} = \alpha * W' \in \cW_+$. Then 
$W^{'\beta} = \beta * W' \in \cW_+$ is equivalent to 
$\beta\alpha^{-1}* W^{'\alpha} = W^{'\beta} \in \cW_+$, and this is 
equivalent to $\beta^{-1}\alpha * \cW_+ = \cW_+$. 
Next we observe that the relation $\beta\alpha^{-1}* W \in \cW_+$ 
is equivalent to the existence of some 
$g \in G^\up_{\uline W}$ with $g.W = (x,\beta^{-1}\alpha \sigma)$, 
which means that $\beta\alpha^{-1} \in Z_2 = \partial(G^\up_{\uline W})$. 
\end{prf}

\begin{ex} \mlabel{ex:pgl2tilde} 
We show that for $G = \tilde\Mob \rtimes \{\1,\tilde\tau\}$ 
as in Example~\ref{ex:models}(d), 
we have to use twisted complements to obtain a $G^\up$-orbit 
in $\cG_E(G)$ invariant under complementation. 
We have already seen that 
$\cG_E(\tilde\Mob)$ contains two $G^\up$-orbits, 
represented by the couples 
$W_0 = (h,\tilde\tau)$ and $W_1 = (h,\tilde\tau_1)$. 
The complement $W_0' = (-h, \tilde\tau)$ satisfies 
\[ \tilde\rho(\pi)W_0' 
= (h, \tilde\rho(\pi)\tilde\tau\tilde\rho(-\pi))
= (h, \tilde\rho(2\pi)\tilde\tau)
= (h, \tilde\tau_1) = W_1,\] 
so that complementation exchanges the two $G^\up$-orbits 
in $\cG_E(\tilde\Mob)$. On the other hand, 
for the action $*_\alpha$ defined in \eqref{eq:twistact}, 
the full group $G$ preserves both $G^\up$-orbits.

Since $\Ad(\rho(-\pi))h = -h$, 
the element $g := \tilde \rho(-\pi)$ can be used to define a suitable
 $\alpha$-twisted conjugation as follows. We note that 
\[\alpha :=   g (g^{\tilde\tau})^{-1} = \tilde \rho(-\pi) \tilde \rho(-\pi) 
= \tilde \rho(-2\pi) \] 
is a generator of $Z := Z(\tilde\Mob) = Z^-$. 
We now have 
\[ W_0^{'\alpha} 
= (-h, \alpha\tilde\tau) 
= \tilde\rho(-\pi).(h, \tilde\tau) 
= \tilde\rho(-\pi).W_0 \in G^\up.W_0.\] 
Thus $\cG_E(\tilde\Mob_2)$ consists of two 
$G^\up$-orbits, none of which is invariant under complementation, 
but both are invariant under $\alpha$-complementation.  An analogous 
computation leads to the same picture for even coverings of~$\Mob$, 
in particular for the fermionic case.
\end{ex}

\section{Euler elements and $3$-graded Lie algebras}

In this section we exhibit a general relation between two notions that are 
a priori unrelated: complementary and orthogonal wedges. 
For the sake of simplicity we consider 
in this introductory part 
the case of the Poincar\'e group $G=\cP_+$ on $\RR^{1+2}$ 
(cf.\ Example~\ref{ex:models}). 
We have seen that if $W=( k_W,j_W)$ is a wedge of the group $G$, then $W'=(- k_W,j_W)$ is the opposite wedge. The $\pi$-spatial rotation 
$\rho(\pi)$ takes $W$ onto $W'$ and vice versa. 
Thus there exists a group element $g\in G^\uparrow=\cP_+^\uparrow$ such that 
$\Ad(g) k_W =- k_W$, and in this sense $ k_W$ is 
{\it symmetric}. 
This ensures a symmetry between a wedge and its opposite 
wedge, which corresponds to its causal complement in Minkowski spacetime.
 
Typical pairs of orthogonal  wedges are the coordinate wedges 
\begin{equation}\label{eq:axialwedge}
 W_i=\{(t,x)\in\R^{1+2}: |t|<x_i\}\equiv( k_i,j_i) \in \cG_E(G)
\quad \mbox{  for } \quad i=1,2.
\end{equation}
The importance of this couple of wedges comes by the clear geometric relation: the wedge reflection of $W_1$ acts on the orthogonal wedge  as
\[ j_1.W_2 = W_2 \quad \mbox{ resp.} \quad  \Ad (j_1)( k_2)= - k_2.\]
In \cite{GL95} the authors study the orthogonality relation in order to 
extend the unitary covariance representation of 
the Poincar\'e group $\cP_+^\uparrow$ to an (anti-)unitary 
representation of the graded group $\cP_+$ and 
establish the Spin--Statistics Theorem. 
In this extension process, orthogonal Euler wedges 
play a crucial role. This point will be discussed from our 
abstract perspective in Section~\ref{sect:sl2int} below.

In this section we  will see how, in our setting, the existence of a 
symmetric Euler element in the Lie algebra ensures 
the existence of an orthogonal pair.  
For symmetric Euler elements, the orthogonality 
relation for Euler elements is symmetric, and orthogonal pairs of Euler elements 
generate a subalgebra isomorphic to $\fsl_2(\R)$ in~$\fg$.

\subsection{Preliminaries on Lie algebras 
and algebraic groups}

In this subsection we collect some basic facts on 
finite dimensional real Lie algebras and on real algebraic groups 
(see \cite{HN12} for Lie algebras 
and \cite{Ho81} for algebraic groups).
 
A Lie algebra $\g$ is called {\it simple} if $\g$ and $\{0\}$ 
are the only ideals of $\g$. It is called 
{\it semisimple} if 
it is a direct sum of simple ideals $\g = \g_1 \oplus \cdots \oplus \g_n$. 
On the other side of the spectrum, 
we have {\it solvable} Lie algebras. These are the ones 
for which the derived series defined by 
$D^0(\g) := \g$ and $D^{n+1}(\g) := [D^n(\g), D^n(\g)]$ 
satisfies $D^N(\g) = \{0\}$ 
for some $N \in \N$. Here 
\[ [\g,\g] = \Spann \{ [x,y] \: x,y \in \g \} \] 
is the {\it commutator algebra} of $\g$. 

The fundamental theorem on the Levi decomposition asserts that,
if $\fr$ is the maximal solvable ideal of~$\g$, then there exists a 
semisimple subalgebra $\fs$ (a Levi complement), such that 
\[ \g \cong \fr \rtimes \fs \] 
is a semidirect sum, i.e., a vector space direct sum of the ideal 
$\fr$ and the subalgebra $\fs$.

A key feature in the structure theory of semisimple real Lie algebras 
is the concept  of a compactly embedded subalgebra. 
A subalgebra $\fk \subeq \g$ is said to be {\it compactly embedded} 
if the subgroup $\Inn_\g(\fk) = \la e^{\ad \fk} \ra \subeq \Aut(\g)$ has 
compact closure. 
We write $\Inn(\g) := \Inn_\g(\g)$ for the subgroup of 
{\it inner automorphisms} of~$\g$.

An element $x \in \g$ is called 
\begin{itemize}
\item {\it elliptic}, if $\ad x$ is semisimple with purely imaginary 
eigenvalues, which is equivalent to  the one-dimensional Lie 
subalgebra $\R x$ being compactly embedded. 
\item {\it hyperbolic}, if $\ad x$ is diagonalizable. 
\item {\it nilpotent}, if $\ad x$ is nilpotent, i.e., 
$(\ad x)^n = 0$ for some $n \in \N$. 
\end{itemize}
The Cartan--Killing form 
\[ \kappa \: \g \times \g \to \R, \quad 
\kappa(x,y) := \tr(\ad x \ad y) \] 
is a symmetric bilinear form on $\g$ invariant under the automorphism 
group $\Aut(\g)$. Recall that a finite dimensional real Lie algebra 
is semisimple if and only if $\kappa$ is non-degenerate 
(Cartan's criterion). 
Note that $\kappa(x,x) = \tr( (\ad x)^2) \geq 0$ 
if $x$ is hyperbolic and $\kappa(x,x) \leq 0$ if $x$ is elliptic. \\

In the proof of Proposition~\ref{prop:1.1} below we shall use some 
results from the theory of linear algebraic groups. We now recall the basic 
concepts. If $V$ is a finite dimensional real vector space, 
then $\GL(V)$ denotes the group of linear automorphisms of $V$. 
Any polynomial function on the linear space $\End(V)$ defines a 
function on the group $\GL(V)$ and we call a subgroup 
$G \subeq \GL(V)$ {\it algebraic} if it is the zero set of a family 
of polynomial functions $p_j \: \End(V) \to \R$. 
An algebraic group $G$ is said to be 
\begin{itemize}
\item {\it reductive}, if 
each $G$-invariant subspace $V_1 \subeq V$ has a $G$-invariant 
linear complement $V_2$. 
\item {\it unipotent}, if there exists a flag of linear subspaces 
\[ F_0 = \{0\} \subeq F_1 \subeq \cdots \subeq F_n = V \] 
such that $(g-\1)F_j \subeq F_{j-1}$ for $j = 1,\ldots, n$ and 
$g \in G$. 
\end{itemize}
In this context one has a decomposition theorem 
(the Levi decomposition), asserting that 
every algebraic subgroup $G \subeq \GL(V)$ is a semidirect product 
$G \cong U \rtimes L$, where $U$ is unipotent and $L$ is reductive. 
Moreover, for every reductive  subgroup $L_1 \subeq G$ there exists an 
element $g \in G$ with $gL_1 g^{-1} \subeq L$ 
(\cite[Thm.~VIII.4.3]{Ho81}).

\subsection{Symmetric and orthogonal Euler elements} 
\mlabel{app:b.1}

\begin{defn}  A pair $(h,x)$ of Euler elements is called {\it orthogonal} 
if $\sigma_h(x) = - x$ (cf.\ Definition~\ref{def:euler}). 
\end{defn}

\begin{prop}
  \mlabel{prop:1.1} The following assertions hold: 
\begin{itemize}
\item[\rm(i)] An Euler element $h \in \g$ is symmetric, i.e., 
$-h \in \cO_h$, if and only 
if $h$ is contained in a Levi complement $\fs$ and 
$h$ is a symmetric Euler element in~$\fs$.
\item[\rm(ii)] If $\g = \fr \rtimes \fs$ is a Levi decomposition. 
\begin{itemize}
\item[\rm(a)]  If $h \in \g$ is a symmetric Euler element, then 
$\cO_h = \Inn(\g)(\cO_h \cap \fs) = \cO_{q(h)}$, 
where $q \: \g \to \fs$ is the projection map. 
\item[\rm(b)] Two symmetric Euler elements are conjugate under $\Inn(\g)$ 
if and only if their images in $\fs$ are conjugate under~$\Inn(\fs)$. 
\end{itemize}
\end{itemize}
\end{prop}

\begin{prf} (i) As $\cO_h \subeq h + [\g,\g]$ follows from the 
invariance of the affine subspace $h + [\g,\g]$ under 
$\Inn(\g)$, the relation 
$-h \in\cO_h$ implies $h \in [\g,\g]$. Let $\g = \fr \rtimes \fs$ 
be a Levi decomposition of $\g$. As $\fs = [\fs,\fs]$, 
the commutator algebra is adapted to this decomposition: 
\[ [\g,\g] = [\fr + \fs,\fr + \fs] = [\g,\fr] + \fs \cong [\fg,\fr] 
\rtimes \fs. \] 
Now $h$ is an Euler element in the ideal 
$[\g,\g] = [\g, \fr] \rtimes \fs$. 
This is the Lie algebra of an algebraic group 
for which $[\g,\fr]$ is the Lie algebra of the unipotent radical 
and $\fs$ the Lie algebra of a reductive complement 
(\cite[Thm.~VIII.3.3]{Ho81}). 
As the algebraic group generated by $\exp(\R \ad h)$ is 
reductive, the conjugacy of Levi decompositions 
(\cite[Thm.~VIII.4.3]{Ho81}) implies that $\ad h$ is contained in some 
Levi complement $\ad \fs$ of $\ad([\g,\g]) = [\ad \g,\ad \g]$. 
Replacing $h$ by another element in $\cO_h$, we 
may thus assume that $h \in \fz(\g) + \fs$ for some 
Levi complement $\fs$ of~$\g$. 
Then $\fr$ and $\fs$ are $\ad h$-invariant, so that the 
$\ad h$-eigenspaces of the restrictions satisfy 
\[ \fr = \fr_1(h) + \fr_0(h) + \fr_{-1}(h) \quad \mbox{ and } \quad 
\fs = \fs_1(h) + \fs_0(h) + \fs_{-1}(h), \] 
and define $3$-gradings of $\fr$ and~$\fs$.
Further $\g_{\pm 1}(h) \subeq [h,\g] \subeq [\g,\g]$ and 
$\fs = [\fs,\fs] \subeq [\g,\g]$ imply that 
$\g =  \fr_0(h) + [\g,\g]$. As $[\g,\g]$ is an ideal and $\fr_0(h)$ 
a subalgebra of $\g$, the subgroup $\Inn_\g([\g,\g])$ of $\Inn(\g)$ is normal, 
and   $\Inn(\g) = \Inn_\g([\g,\g]) \Inn(\fr_0(h))$. 
As $\Inn(\fr_0(h))$ fixes $h$, this in turn shows that 
$\cO_h = \Inn_\g([\g,\g])h = \Inn_\g([\g,\fr]) \Inn_\g(\fs)h$. 
Writing $h = h_z + h_s$ with $h_z \in \fz(\g)$ and $h_s \in \cE(\fs)$, we thus 
find $x \in [\g,\fr]$ and $s \in \Inn_\g(\fs)$ such that
\begin{footnote}{Here we use that the Lie algebra $[\g,\fr]$ is nilpotent, 
so that the exponential function of the corresponding group 
$\Inn_\g([\g,\fr])$ is surjective, see \cite{HN12}.}  
\end{footnote}

\[ -h_z- h_s = -h = e^{\ad x} s.h = h_z + e^{\ad x} s.h_s.\] 
Applying the Lie algebra 
homomorphism $q$ to both sides, \red{we derive from 
$q(h_z) = 0$ and $q \circ e^{\ad x} = q$ that} $-h_s = s.h_s$, and therefore 
\[e^{\ad x} h_s = h_s + 2h_z.\] 
We conclude that the unipotent linear map $e^{\ad x}$ preserves 
the linear subspace $\R h_s + \R h_z$, and this implies that 
$\ad x = \log(e^{\ad x})$ also has this property. 
We thus arrive at 
\[   [h,x] = [h_s, x] \subeq \R h_s + \R h_z \subeq \g_0(h),\] 
so that we must have $x \in \g_0(h) = \g_0(h_s)$, which in turn leads to 
$0 = e^{\ad x}h_s  - h_s = 2h_z$, i.e., $h = h_s\in \fs$. 

To prove the second assertion of (i), we observe that the 
homomorphism $q \: \g \to \fs \cong \g/\fr$ satisfies 
\begin{equation}
  \label{eq:orbproj}
q(\cO_x) = \cO^\fs_{q(x)} \quad \mbox{ for } \quad x \in \g.
\end{equation}
Hence $q(\cE_{\rm sym}(\g)) \subeq \cE_{\rm sym}(\fs)$. 
If, conversely, $h \in \cE_{\rm sym}(\fs)$, then we clearly 
have $-h \in \Inn_\g(\fs)h \subeq \Inn(\g)h$, so that $h \in \cE_{\rm sym}(\g)$.

\nin (ii)(a) As $\cO_h$ intersects $\fs$ by (i), 
$q(\cO_h) \cap \cO_h \not=\eset$, and since $\Inn(\fs)$ acts transitively 
on $q(\cO_h)$ by \eqref{eq:orbproj}, we obtain $q(\cO_h) \subeq \cO_h$ 
and thus $q(\cO_h) = \cO_h \cap \fs$. This further leads to 
\[ \cO_h = \Inn(\g)(\cO_h \cap \fs) = \Inn(\g)q(\cO_h) 
= \Inn(\g) \cO^\fs_{q(h)} = \cO_{q(h)}.\] 

\nin (ii)(b) follows immediately from (a). 
\end{prf}

Proposition~\ref{prop:1.1} reduces the description of symmetric Euler elements 
up to conjugation by inner automorphisms to the case of simple Lie algebras. 

\begin{rem} Suppose that $\g$ is a 
finite dimensional Lie algebra containing a pointed generating 
invariant cone $C$.  If $\g$ is not reductive, then 
$C \cap \fz(\g) \not=\{0\}$ 
(\cite[Thm.~VII.3.10]{Ne99}). 
If $\tau = \sigma_h$ is an involution defined by a 
symmetric Euler element $h$, then $\tau$ fixes every central 
element, so that we cannot have $\tau(C) = - C$ 
if $\g$ is not reductive. 
\end{rem}

\begin{exs} \mlabel{exs:1.1} (a) If $\fs$ is a semisimple Lie algebra and 
$h \in \fs$ an Euler element, then it also is an Euler element in the 
semidirect sum $T\fs := |\fs| \rtimes \fs$, where 
$|\fs|$ is the linear subspace underlying $\fs$, endowed with the 
$\fs$-module structure defined by the adjoint representation. 

\nin (b) In the simple Lie algebra $\g := \fsl_n(\R)$, we write 
$n \times n$-matrices as block $2 \times 2$-matrices 
according to the partition $n = k + (n-k)$. Then 
\[ h_k := \frac{1}{n}\pmat{(n-k) \1_k & 0 \\ 0 & -k\1_{n-k}} \] 
is diagonalizable with the two eigenvalue $\frac{n-k}{n} = 1 - \frac{k}{n}$ 
and $-\frac{k}{n}$. 
Therefore $h_k$ is an Euler element whose $3$-grading is given by 
\begin{align*}
\g_0(h) &= \Big\{ \pmat{a & 0 \\ 0 & d} \: 
a \in \gl_k(\R), d \in \gl_{n-k}(\R), \tr(a) + \tr(d) = 0\Big\},  \\ 
\g_1(h) &= \pmat{ 0 & M_{k,n-k}(\R)\\ 0 & 0}, \quad 
\g_{-1}(h) \cong \pmat{0 & 0 \\ M_{n-k,k}(\R) & 0}.
\end{align*}
\end{exs}

\begin{ex} \mlabel{ex:1.4} For $\g = \fsl_2(\R)$, the Euler element 
\[ h := \frac{1}{2}\pmat{1 & 0 \\ 0 & -1} \quad \mbox{ satisfies } \quad 
\sigma_h\pmat{a & b \\ c &d} 
= \pmat{a & -b \\ -c &d}.\] 
Any element in $\Fix(-\sigma_h)$ is of the form 
$x = \pmat{0 & b \\ c & 0},$ 
and it is an Euler element if and only if $bc = - \det(x) = \frac{1}{4}$. 
If $g \in \SL_2(\R)$ commutes with $h$, then it is diagonal, i.e., 
$g = \pmat{a & 0 \\ 0 & a^{-1}}$, and thus 
\[ \Ad(g)\pmat{0 & b \\ c & 0}
= \pmat{0 & a^2 b \\ a^{-2}c & 0}.\] 
We thus obtain two representatives 
\[ x_\pm = \pm \frac{1}{2} \pmat{0 & 1 \\ 1 & 0}\] 
of conjugacy classes of orthogonal pairs $(h,x)$ of Euler elements 
for $\fsl_2(\R)$. 
The involution corresponding to $x_\pm$ is given by 
\[ \sigma_{x_\pm}\pmat{a & b \\ c &d} = 
 e^{\pi i x_\pm} \pmat{a & b \\ c &d} e^{-\pi i x_\pm} 
= \pmat{0 & i \\ i & 0}\pmat{a & b \\ c &d} 
\pmat{0 & -i \\ -i & 0}
= \pmat{d & c \\ b &a},\]  
which shows in particular that 
\begin{equation}
  \label{eq:sym2}
 \sigma_{x_\pm}(h) = -h. 
\end{equation}
\end{ex}

As a consequence of the preceding discussion, we see that the orthogonality 
relation on $\cE(\fsl_2(\R))$ is symmetric: 
\begin{lem} \mlabel{lem:sl2}
If $(x,y)$ is an orthogonal pair of Euler elements in 
$\fsl_2(\R)$, then $\sigma_y(x) = -x$, so that 
$(y,x)$ is also symmetric.   
\end{lem}

\begin{ex} For $\g = \gl_2(\R)$, the Euler element 
\[ h := \pmat{1 & 0 \\ 0 & 0} \quad \mbox{ satisfies } \quad 
 \sigma_h\pmat{a & b \\ c &d} = 
 \pmat{-1 & 0 \\ 0 & 1}\pmat{a & b \\ c &d} 
\pmat{-1 & 0 \\ 0 & 1}
= \pmat{a & -b \\ -c &d},\] 
and we see, as for $\fsl_2(\R)$, that 
the orthogonal Euler elements are given by 
\[ x_\pm = \pm \frac{1}{2} \pmat{0 & 1 \\ 1 & 0} 
\quad \mbox{ with }\quad  \sigma_{x_\pm}\pmat{a & b \\ c &d} 
= \pmat{d & c \\ b &a}.\]  
This shows that 
\begin{equation}
  \label{eq:sym3}
   \sigma_{x_\pm}(h) \not= -h. 
\end{equation}
Therefore $\gl_2(\R)$ contains a pair $(h,x)$ 
of orthogonal Euler elements for which 
$\sigma_x(h) \not= -h$. From $h \not\in [\g,\g]$ it immediately 
follows that $h$ is not symmetric.  
We shall see in Theorem~\ref{thm:automaticsym}  below that this pathology 
of the orthogonality relation on the set of Euler elements does not 
occur for symmetric Euler elements. 
\end{ex}

\begin{ex} For $\g = \fsl_3(\R)$, the Euler element 
\[ h_1 := \frac{1}{3}\pmat{
2 & 0  & 0 \\ 
0 & -1  & 0 \\ 
0 & 0  & -1} \quad \mbox{ satisfies } \quad  \sigma_{h_1}\pmat{a & b \\ c &d} 
= \pmat{a & -b \\ -c &d},\] 
where we write matrices as $2 \times 2$-block matrices according to the 
partition $3 = 1 + 2$. 
Up to conjugacy under the centralizer of $h_1$, 
the symmetric matrices in $\Fix(-\sigma_{h_1})$ are represented by 
\[ x = \pmat{0 & 0 & a \\ 0 & 0 & 0 \\ a & 0 & 0}.\] 
These matrices have three different eigenvalues, 
so that \red{$\ad x$ has five eigenvalues}, and thus $x$  cannot 
be an Euler elements of $\fsl_3(\R)$. 
We conclude that, there exists no Euler element 
$x \in \cE(\fsl_3(\R))$ for which $(h_1,x)$ is orthogonal.

We shall see in Theorem~\ref{thm:automaticsym}(b) below that this never 
happens for symmetric Euler elements, but $h_1$ is not symmetric. 
It corresponds to $h_1$ for the root system $A_2$ 
in the notation of Section~\ref{app:b.2}. 
\end{ex}

\begin{ex} \mlabel{ex:sl4} For $\g = \fsl_4(\R)$, the Euler element 
\[ h_1 := \frac{1}{4}\pmat{
3 & 0  & 0 & 0 \\ 
0 & -1  & 0 & 0 \\ 
0 & 0  & -1 & 0 \\ 
0 & 0  & 0 & -1} \quad \mbox{ satisfies } \quad 
 \sigma_{h_1}\pmat{a & b \\ c &d} 
= \pmat{a & -b \\ -c &d},\] 
where we write matrices as $2 \times 2$-block matrices according to the 
partition $4 = 1 + 3$. 
Up to conjugacy under the centralizer of $h_1$, 
the symmetric matrices in $\Fix(-\sigma_{h_1})$ are represented by 
\[ x = \pmat{
0 & 0 & 0 & a \\
0 & 0 & 0 & 0 \\
0 & 0 & 0 & 0 \\
a & 0 & 0 & 0}.\] 
They all have three different eigenvalues and 
\red{$\ad x$ has five eigenvalues,} 
so that they are not 
Euler elements. We conclude that there exists no Euler element 
$x \in \cE(\fsl_4(\R))$ for which $(h_1,x)$ is orthogonal.

This is different for the symmetric Euler element 
\[ h_2 := \frac{1}{2}\pmat{
1 & 0  & 0 & 0 \\ 
0 & 1  & 0 & 0 \\ 
0 & 0  & -1 & 0 \\ 
0 & 0  & 0 & -1} \quad \mbox{ with } \quad  \sigma_{h_2}\pmat{a & b \\ c &d} 
= \pmat{a & -b \\ -c &d},\] 
where we write matrices as $2 \times 2$-block matrices according to the 
partition $4 = 2 + 2$. 
Up to conjugacy under the centralizer of $h_2$, 
the symmetric matrices in $\Fix(-\sigma_{h_2})$ are represented by 
\[ x = \pmat{
0 & 0 & a & 0 \\
0 & 0 & 0 & b \\
a & 0 & 0 & 0 \\
0 & b & 0 & 0},\] 
and, for $a = b = \shalf$, these are Euler elements orthogonal to $h_2$. 
\end{ex}

\subsection{Euler elements in simple real Lie algebras}
\mlabel{app:b.2}

In this section we take a systematic look at Euler elements 
in simple real Lie algebras. In particular we determine which of them 
are symmetric and show that pairs of orthogonal ones 
generate $\fsl_2$-subalgebras (Theorem~\ref{thm:automaticsym}). 
For the classification of $3$-gradings of simple Lie algebras, 
we refer to \cite{KA88}, the concrete list of the $18$ types 
in \cite[p.~600]{Kan98} which is also listed below, 
and Kaneyuki's lecture notes \cite{Kan00}. 

Let $\g$ is a real semisimple Lie algebra. 
An involutive automorphism $\theta \in \Aut(\g)$ is called a {\it 
Cartan involution} if its eigenspaces 
\[ \fk := \g^\theta = \{ x \in\g \: \theta(x) = x \} \quad \mbox{ and }\quad 
\fp := \g^{-\theta} = \{ x \in\g \: \theta(x) = -x \} \] 
have the property that they are orthogonal with respect to $\kappa$, 
which is  negative definite on $\fk$ and 
positive definite on $\fp$. Then 
\begin{equation}
  \label{eq:cartandec}
\g = \fk \oplus \fp 
\end{equation}
is called a {\it Cartan decomposition}. 
Cartan involutions always exist and two such involutions are conjugate 
under the group $\Inn(\g)$ of inner automorphism, so they produce 
isomorphic decompositions (\cite[Thm.~13.2.11]{HN12}). 

If $\g = \fk \oplus \fp$ is a Cartan decomposition, then 
$\fk$ is a maximal compactly embedded subalgebra of~$\g$, 
$x \in \g$ is elliptic 
if and only if its adjoint orbit $\cO_x =\Inn(\g)x$ intersects $\fk$, 
and $x \in \g$ is hyperbolic if and only if 
$\cO_x \cap \fp \not=\eset$. 

For the finer structure theory, and also for classification 
purposes, one starts with a Cartan involution $\theta$ and 
fixes a maximal abelian subspace $\fa \subeq \fp$. As $\fa$ is abelian, 
$\ad \fa$ is a commuting set of diagonalizable operators, hence 
simultaneously diagonalizable. 
For a linear functional $0 \not=\alpha \in \fa^*$, the simultaneous eigenspaces 
\[ \g_\alpha :=  \{ y \in \g \: (\forall x \in \fa) \ [x,y] = \alpha(x)y\} \] 
are called {\it root spaces} and 
\[ \Sigma := \Sigma(\g,\fa) := \{ \alpha \in \fa^* \setminus \{0\}  \: 
\g_\alpha \not=0\} \] 
is called the set of {\it restricted roots}. 
We pick a set 
\[ \Pi := \{ \alpha_1, \ldots, \alpha_n \} \subeq \Sigma \] 
of {\it simple roots}. This is a subset with the property that every 
root $\alpha \in \Sigma$ is a linear combination 
$\alpha = \sum_{j =1}^n n_j \alpha_j$, where the coefficients 
are either all in $\Z_{\geq 0}$ or in $\Z_{\leq 0}$. The convex cone 
\[ \Pi^\star := \{ x \in \fa \: (\forall \alpha \in \Pi) \ \alpha(x) \geq 0\} \] 
is called the {\it positive (Weyl) chamber corresponding to $\Pi$}. 

We have the {\it root space decomposition} 
\[ \g = \g_0 \oplus \bigoplus_{\alpha \in \Sigma} \g_\alpha 
\quad \mbox{ and }\quad 
\g_0 = \fm \oplus \fa, \quad \mbox{ where } \quad 
\fm = \g_0 \cap \fk.\] 
Now $\theta(\g_\alpha) = \g_{-\alpha}$, and for 
a non-zero element $x_\alpha \in \g_\alpha$, the 
$3$-dimensional subspace spanned by $x_\alpha, \theta(x_\alpha)$ and 
$[x_\alpha, \theta(x_\alpha)] \in \fa$ is a Lie subalgebra 
isomorphic to $\fsl_2(\R)$. In particular, it contains 
a unique element $\alpha^\vee \in \fa$ with $\alpha(\alpha^\vee) = 2$. 
Then 
\[ r_\alpha \: \fa \to \fa, \quad r_\alpha(x) := x - \alpha(x) \alpha^\vee \] 
is a reflection, and the subgroup 
\[ \cW := \la r_\alpha \: \alpha \in \Sigma \ra \subeq \GL(\fa) \]
is called the {\it Weyl group}. 
Its action on $\fa$ provides a good description of the adjoint 
orbits of hyperbolic elements: Every hyperbolic 
element in $\g$ is conjugate to a unique element in $\Pi^\star$ and, 
for $x \in \fa$, the intersection  
$\cO_x \cap \fa = \cW x$ is the Weyl group orbit 
(\cite[Thm.~III.10]{KN96}). 

{\bf From now on we assume that $\g$ is simple.} 
Then $\Sigma$ is an irreducible root system, hence of one of the following types: 
\[ A_n,  \qquad  
 B_n, \qquad  
 C_n,  \qquad  
 D_n, \qquad  
 E_6, E_7, E_8,\ \  F_4,\ \  G_2
\quad \mbox{ or } \quad BC_n, n \geq 1\] 
(cf.\ \cite{Bo90a}).
If $\g$ is a complex simple Lie algebra, then it is also simple 
as a real Lie algebra, and a Cartan decomposition takes the form 
\[ \g = \fk \oplus i \fk,\] 
where $\fk \subeq \g$ is a compact real form. Then $\fa = i \ft$, where 
$\ft \subeq \fk$ is maximal abelian. In particular, the restricted root system 
$\Sigma(\g,\fa)$ coincides with the root system of the complex Lie algebra~$\g$. 
This leads to a one-to-one correspondence between isomorphy classes of simple complex 
Lie algebras and the irreducible reduced root systems. 
If $\g$ is not complex, then 
\red{neither the isomorphy class of $\g$ nor of $\g_\C$} is determined by 
the root system $\Sigma(\g,\fa)$. For instance all Lie algebras 
$\so_{1,n}(\R)$ have the restricted root system $A_1$ with $\dim \fa = 1$, but 
their complexifications $\so_{n+1}(\C)$ have the root systems 
$B_k$ for $n = 2k$ and $D_k$ for $n = 2k-1$. 

The adjoint orbit of an Euler element in $\g$ contains  
a unique $h \in \Pi^\star$. \red{For any Euler element 
$h \in \Pi^\star$, we have $\alpha(h) \in \{0,1\}$ for $\alpha \in \Pi$ 
because the values of the roots on $h$ are the eigenvalues of $\ad h$. 
If such an element exists, then the irreducible root system 
$\Sigma$ must be reduced. Otherwise, for any root $\alpha$ with 
$2\alpha \in \Sigma$, we must have $\alpha(h) = 0$ because 
$\ad x$ has only three eigenvalues. As the set of such roots 
generates the same linear space as $\Sigma$, this leads to the 
contradiction $h = 0$. }
This excludes the non-reduced simple root systems of type~$BC_n$.

To see how many possibilities we have for Euler elements in $\fa$, 
we recall that $\Pi$ is a linear basis of $\fa$, so that, for each $j \in \{1,\ldots, n\}$, there exists a uniquely determined element 
\begin{equation}
  \label{eq:hj}
h_j \in \fa, \quad \mbox{  satisfying  } \quad \alpha_k(h_j) =
\begin{cases}
  1 & \text{ for } \ j = k \\ 
  0 & \text{ otherwise}.
\end{cases}
\end{equation}

A simple Lie algebra $\g = \fk \oplus \fp$ is called {\it hermitian} 
if the center $\fz(\fk) = \{ x \in \fk \:  [x,\fk] = \{0\}\}$ 
of a maximal compactly embedded subalgebra $\fk$ is 
non-zero. For hermitian Lie algebras, 
the restricted root system $\Sigma$ 
is either of type $C_r$ or $BC_r$ (cf.\ Harish Chandra's Theorem 
\cite[Thm.~XII.1.14]{Ne99}), 
and we say that $\g$ is 
{\it of tube type} if the restricted root system is of type $C_r$. 

The following theorem lists for each irreducible root system $\Sigma$ 
the possible Euler elements in the positive chamber $\Pi^\star$. 
Since every adjoint orbit in $\cE(\g)$ has a unique 
representative in $\Pi^\star$, this classifies the 
$\Inn(\g)$-orbits in $\cE(\g)$ for any non-compact simple real Lie algebra. 
For semisimple algebras $\g = \g_1 \oplus \cdots \oplus \g_k$, an 
element $x = (x_1, \ldots, x_n)$ is an Euler element if and only if its 
components $x_j \in \g_j$ are Euler elements, and its orbit is 
\[ \cO_x = \cO_{x_1} \times \cdots \times \cO_{x_k}.\] 
Therefore it suffices to consider simple Lie algebras, 
and for these the root system $\Sigma$ is irreducible. 
As every complex simple Lie algebra $\g$ 
is also a real simple Lie algebra, our discussion also 
covers  complex Lie algebras.

\begin{theorem} \mlabel{thm:classif-symeuler}
Suppose that $\g$ is a non-compact simple 
real Lie algebra, with restricted root system 
$\Sigma \subeq \fa^*$ of type $X_n$. 
We follow the conventions of the tables in {\rm\cite{Bo90a}}
for the classification of irreducible root systems and the enumeration 
of the simple roots $\alpha_1, \ldots, \alpha_n$. 
Then every Euler element $h \in \fa$ on which 
$\Pi$ is non-negative is one of  $h_1, \ldots, h_n$, and for 
every irreducible root system, the Euler elements among the $h_j$ are 
 the following: 
\begin{align} 
&A_n: h_1, \ldots, h_n, \quad 
\ \ B_n: h_1, \quad 
\ \ C_n: h_n, \quad \ \ \ D_n: h_1, h_{n-1}, h_n, \quad 
E_6: h_1, h_6, \quad 
E_7: h_7.\label{eq:eulelts2}
\end{align}
For the root systems $BC_n$, $E_8$, $F_4$ and $G_2$ no Euler element exists 
(they have no $3$-grading). 
The symmetric Euler elements are 
\begin{equation}
  \label{eq:symmeuler}
A_{2n-1}: h_n, \qquad 
B_n: h_1, \qquad C_n: h_n, \qquad 
D_n: h_1, \qquad 
D_{2n}: h_{2n-1},h_{2n}, \qquad 
E_7: h_7.  
\end{equation}
\end{theorem}

\begin{prf} 
Writing the highest root in $\Sigma$ with respect to the simple system 
$\Pi$  as $\alpha_{\rm max} = \sum_{j = 1}^n c_j \alpha_j$, 
we have $c_j \in \Z_{>0}$ for each $j$. 
If $h \in \Pi^\star$ is an Euler element, then $\Pi(h) \subeq \{0,1\}$, 
and $1 = \alpha_{\rm max}(h) = \sum_{j = 1}^n c_j \alpha_j(h)$ implies
that at most one value $\alpha_j(h)$ can be $1$,  and then the others are~$0$, 
i.e., $h = h_j$ for some $j \in \{1,\ldots, n\}$. 
Moreover, $h_j$ is an Euler element if and only if $c_j = 1$. 
Consulting the tables on the irreducible root systems 
in \cite{Bo90a}, we obtain the Euler elements listed in 
\eqref{eq:eulelts2}. 

To determine the symmetric ones, let $w_0 \in \cW$ be 
the longest element of the Weyl group, which is uniquely determined by 
$w_0^*\Pi = - \Pi$ for the dual action of $\cW$ on $\fa^*$. 
Then $h_j' := w_0(-h_j)$ is the Euler element in 
the positive chamber representing the orbit $\cO_{-h_j}$. 
Therefore $h_j$ is symmetric if and only if 
$-h_j \in \cW h_j$, which is equivalent to $h_j' = h_j$. 
Using the description of $w_0$ and the root systems in \cite{Bo90a},  
now leads to 
\begin{align}
& A_{n-1}: h_j' = h_{n-j}, \quad 
B_n: h_1' = h_1, \quad C_n: h_n' = h_n, \\ 
& D_n: h_1' = h_1, h_n' =
\begin{cases}
  h_{n-1} & \text{ for } n \ \text{ odd},\\ 
  h_n & \text{for } n \ \text{ even}, 
\end{cases}\\
&E_6: h_1' = h_6, \quad E_7: h_7' = h_7. 
\end{align}
Hence the symmetric Euler elements are given by 
the list~\eqref{eq:symmeuler}.
\end{prf}

This theorem requires some interpretation. So let us first 
see what it says about complex simple Lie algebras~$\g$. In \eqref{eq:eulelts2} 
we see that only if $\g$ is not of type $E_8, F_4$ or $G_2$, the Lie algebra $\g$ 
contains an Euler element. As Euler elements correspond to $3$-gradings 
of the root system and these in turn to hermitian real forms $\g^\circ$, 
where $ih_j \in \fz(\fk^\circ)$ generates the center of a maximal 
compactly embedded subalgebra $\fk^\circ$ 
\red{(\cite[Thm.~A.V.1]{Ne99})}. We thus obtain the 
following possibilities. In Table 1, we write 
$\g^\circ$ for the hermitian real form, $\g$ for the complex Lie algebra, 
$\Sigma$ for its root system, and $h_j$ for the corresponding Euler element:\\[2mm]
\begin{tabular}{||l|l|l|l|l||}\hline
{} $\g^\circ$ \mbox{(hermitian)} & $\Sigma(\g^\circ, \fa^\circ)$  & $\g = (\g^\circ)_\C$ & $\Sigma(\g,\fa)$ & 
{\rm Euler element} \\ 
\hline\hline 
$\su_{p,q}(\C), 1 \leq p\leq q$ & $BC_p (p < q)$, $C_p (p=q)$ & $\fsl_{p+q}(\C)$ & $A_{p+q-1}$ & 
$h_p$ \\ 
 $\so_{2,2n-1}(\R), n > 1$ & $C_2$ & $\so_{2n+1}(\C)$ & $B_n$ & $h_1$ \\ 
$\sp_{2n}(\R)$ & $C_n$ & $\fsp_{2n}(\C)$ & $C_n$ & $h_n$ \\ 
 $\so_{2,2n-2}(\R), n > 2$ & $C_2$ & $\so_{2n}(\C)$ & $D_{n}$ & $h_1$ \\ 
 $\so^*(2n)$ & $BC_m (n = 2m+1)$, $C_m (n = 2m)$ & $\so_{2n}(\C)$  & $D_{n}$ 
& $h_{n-1}, h_n$ \\ 
$\fe_{6(-14)}$ & $BC_2$ & $\fe_6$ & $E_6$ & $h_1 = h_6'$ \\ 
$\fe_{7(-25)}$ & $C_3$ & $\fe_7$ & $E_7$ & $h_7$ \\ 
\hline
\end{tabular} \\[2mm] {\rm Table 1: Simple hermitian Lie algebras $\g^\circ$}\\

In this correspondence, those hermitian simple Lie algebras 
corresponding to symmetric Euler elements are of particular interest. 
Comparing with the list of hermitian simple Lie algebras 
of tube type (cf.~\cite[p.~213]{FK94}), we see that they 
correspond precisely the $3$-gradings specified by symmetric Euler elements, 
as listed in~\eqref{eq:symmeuler}. 
Since the Euler elements $h_{n-1}$ and $h_n$ for the root system of type 
$D_n$ are conjugate under a diagram automorphism, they correspond to 
isomorphic hermitian real forms. \\[2mm] 
\begin{tabular}{||l|l|l|l|l||}\hline
{} $\g^\circ$ \mbox{(hermitian)}  & $\Sigma(\g^\circ, \fa^\circ)$ & $\g = (\g^\circ)_\C$ & $\Sigma(\g,\fa)$ & {\rm symm.\ Euler element}\  $h$ \\ 
\hline\hline 
$\su_{n,n}(\C)$ & $C_n$ & $\fsl_{2n}(\C)$ & $A_{2n-1}$ & $h_n$ \\ 
 $\so_{2,2n-1}(\R), n > 1$ & $C_2$ & $\so_{2n+1}(\C)$ & $B_n$ & $h_1$ \\ 
$\sp_{2n}(\R)$ & $C_n$ & $\fsp_{2n}(\C)$ & $C_n$ & $h_n$ \\ 
 $\so_{2,2n-2}(\R), n > 2$ & $C_2$ & $\so_{2n}(\C)$ & $D_{n}$ & $h_1$ \\ 
 $\so^*(4n)$ & $C_n$ & $\so_{4n}(\C)$  & $D_{2n}$ & $h_{2n-1},  h_{2n}$ \\ 
$\fe_{7(-25)}$ & $C_3$ & $\fe_7$ & $E_7$ & $h_7$ \\ 
\hline
\end{tabular} \\[2mm] {\rm Table 2: Simple hermitian Lie algebras $\g^\circ$ 
of tube type}\\

In our context hermitian simple Lie algebras are of particular interest. 
We therefore collect some of their main properties in the following  proposition.
\begin{prop} \mlabel{prop:herm} 
For a simple real Lie algebra, the following assertions hold: 
  \begin{itemize}
  \item[\rm(a)] $\g$ is hermitian if and only if there exists a 
closed convex $\Inn(\g)$-invariant cone $C \not= \{0\},\g$. 
  \item[\rm(b)] A simple hermitian Lie algebra contains an Euler 
element if and only if it is of tube type, and in this case $\Inn(\g)$ 
acts transitively on $\cE(\g)$. 
  \end{itemize}
\end{prop}

\begin{prf} (a) is a consequence of the Kostant--Vinberg Theorem 
(cf.\ \cite[Lemma~2.5.1]{HO96}). 

\nin (b) Since the restricted root system of a hermitian simple Lie algebra 
is of type $C_r$ or $BC_r$, and the first case characterizes the algebras of tube 
type, the assertion follows from 
Theorem~\ref{thm:classif-symeuler} because 
$C_r$ only permits one class of Euler elements. 
\end{prf}

There are many types of simple $3$-graded Lie algebras that are neither 
complex nor hermitian of tube type; for instance the Lorentzian algebras 
$\so_{1,n}(\R)$. We refer to 
\cite[p.~600]{Kan98} or \cite{Kan00}. 
for the list of all $18$ types which is reproduced below. \\[2mm]
\begin{tabular}{||l|l|l|l|l||}\hline
& $\g$  & $\Sigma(\g,\fa)$  & $h$ & $\g_1(h)$  \\ 
\hline\hline 
1 & $\fsl_n(\R)$ & $A_{n-1}$ & $h_j, 1 \leq j \leq n-1$ & $M_{j,n-j}(\R)$  \\ 
2 & $\fsl_n(\H)$ & $A_{n-1}$ & $h_j, 1 \leq j \leq n-1$ & $M_{j,n-j}(\H)$  \\ 
3 & $\su_{n,n}(\C)$ & $C_{n}$ & $h_n$ & $\Herm_n(\C)$  \\ 
4 & $\sp_{2n}(\R)$ & $C_{n}$ & $h_n$ & $\Sym_n(\R)$   \\ 
5 & $\fu_{n,n}(\H)$ & $C_{n}$ & $h_n$ & $\Aherm_n(\H)$  \\ 
6  & $\so_{p,q}(\R)$ & $ B_p\ (p<q),\ D_p\ (p = q)$ & $h_1$ & $\R^{p+q-2}$   \\ 
7  & $\so^*(4n)$ & $C_n$ & $h_n$ & $\Herm_n(\H)$   \\ 
8  & $\so_{n,n}(\R)$ & $C_n$ & $h_n$ & $\Alt_n(\R)$   \\ 
9 & $\fe_6(\R)$ & $E_6$ & $h_1=h_6' $ & $M_{1,2}(\bO_{\rm split})$   \\ 
10 & $\fe_{6(-26)}$ & $A_2$ & $h_1$ & $M_{1,2}(\bO)$   \\ 
11 & $\fe_7(\R)$ & $E_7$ & $h_7 $ & $\Herm_3(\bO_{\rm split})$   \\ 
12 & $\fe_{7(-25)}$ & $C_3$ & $h_3$ & $\Herm_3(\bO)$   \\ 
13 & $\fsl_n(\C)$ & $A_{n-1}$ & $h_j, 1 \leq j \leq n-1$ & $M_{j,n-j}(\C)$  \\ 
14 & $\sp_{2n}(\C)$ & $C_{n}$ & $h_n$ & $\Sym_n(\C)$   \\ 
15a & $\so_{2n+1}(\C)$ & $ B_{n}$ & $h_1$ & $\C^n$   \\ 
15b & $\so_{2n}(\C)$ & $ D_{n}$ & $h_1$ & $\C^n$   \\ 
16 & $\so_{2n}(\C)$ & $ D_{n}$ & $h_{n-1}, h_n$ & $\Alt_n(\C)$   \\ 
17 & $\fe_6(\C)$ & $E_6$ & $h_1 = h_6' $ & $M_{1,2}(\bO)_\C$   \\ 
18 & $\fe_7(\C)$ & $E_7$ & $h_7 $ & $\Herm_3(\bO)_\C$   \\ 
\hline
\end{tabular} \\[2mm] {\rm Table 3: Simple $3$-graded Lie algebras}\\

\begin{rem} As $h \in \fa$ implies $\theta(h) = -h$, the Cartan 
involution $\theta$ always maps $h$ into $-h$, but this only 
implies that $h$ is symmetric if $\theta \in \Inn(\g)$.
This is the case if $\g$ is hermitian, so that in these Lie algebras 
all Euler elements are symmetric. 
\end{rem}

We conclude this section with some finer results concerning orthogonality 
and symmetry of Euler elements. 

\begin{thm} \mlabel{thm:automaticsym} 
If $\g$ is simple and $h \in \cE(\g)$, then 
the following assertions hold: 
\begin{itemize}
\item[\rm(a)] If $x \in \cE(\g)$ is such that 
$(h,x)$ is orthogonal, then 
\begin{itemize}
\item[\rm(i)] $h$ and $x$ are symmetric,  
\item[\rm(ii)] the Lie algebra generated by $h$ and $x$ is isomorphic to 
$\fsl_2(\R)$, and 
\item[\rm(iii)] $\sigma_x(h) = -h$, so that  $(x,h)$ is also orthogonal. 
\end{itemize}
\item[\rm(b)] There exists an Euler element 
$x$ such that $(h,x)$ is orthogonal if and only if $h$ is symmetric. 
\end{itemize}
\end{thm}

\begin{prf}  
(a) We split the proof into the two cases, according to whether $\g$ 
is a complex Lie algebra or not. We then reduce the second case to the first one. \\
\nin {\bf Case 1: $\g$ is complex:} A simple complex Lie algebra $\g$ 
contains an Euler element, i.e., it possesses 
a $3$-graded root system, if and only if it has a real form 
$\g^\circ$ which is hermitian, i.e., $\g = (\g^\circ)_\C = \g^\circ \oplus i \g^\circ$. 
This follows for example by comparing the list of 
irreducible root systems for which 
Euler elements exist (see \eqref{eq:eulelts2}) 
with the classification of hermitian simple Lie algebras 
$\g^\circ$ (see \cite[Thm.~A.V.1]{Ne99} and Table~$1$). 
In this case the real Lie algebra $\g^\circ$ has a Cartan decomposition 
$\g^\circ = \fk^\circ \oplus \fp^\circ$ and the center $\fz(\fk^\circ)$ is 
one-dimensional 
 and generated by an element $z$ with 
$\Spec(\ad z) = \{0,\pm i\}$ \red{(\cite[Thm.~A.V.1]{Ne99})}. 
Then $h = i z$ is an Euler element in 
the complexification $\g$ for which $\fk^\circ = \ker(\ad z) \cap \g^\circ$ and 
$[z,\g^\circ] = \fp^\circ$, where $\ad z\res_{\fp^\circ}$ is a complex structure on the real 
vector space~$\fp^\circ$. The corresponding Euler 
involution $\sigma_h = e^{\pi i \ad h} = e^{\pi \ad z} \in \Aut_\C(\g)$ 
thus restricts to the Cartan involution on $\g^\circ$, corresponding to the 
decomposition $\fk^\circ \oplus \fp^\circ$. Accordingly, we obtain 
\[ \fh := \Fix(\sigma_h) = (\fk^\circ)_\C \quad \mbox{ and } \quad 
\fq := \Fix(-\sigma_h) = (\fp^\circ)_\C.\] 

A Cartan decomposition of $\g$ is obtained by 
$\fk = \fk^\circ + i \fp^\circ$ and 
$\fp = \fp^\circ + i \fk^\circ$. If $\ft \subeq \fk^\circ$ is a maximal abelian 
Lie subalgebra, then $\fa := i \ft \subeq \fp$ is a maximal abelian subspace 
which contains $h = i z \in i \fz(\fk^\circ) \subeq i\ft$.  
The orthogonality of the pair $(h,x)$ means that 
$x \in \fq = \Fix(-\sigma_h)$. By \cite[Cor.~III.9]{KN96}, 
$x \in \cE(\g) \cap \fq$ is conjugate under 
the centralizer of $h$ to an element in $\fq \cap \fp = \fp^\circ$.  
Fixing a maximal abelian subspace $\fa^\circ \subeq \fp^\circ$, 
we may therefore assume that $x$ is an Euler element for the corresponding 
restricted root system $\Sigma^\circ := \Sigma(\g^\circ, \fa^\circ)
\subeq (\fa^\circ)^*$, which is of 
type~$C_r$ or $BC_r$ (cf.~\cite[Thm.~XII.1.14]{Ne99}). 
As we have already observed above, 
the existence of an Euler element $x \in \fa^\circ$ 
implies that the restricted root system $\Sigma^\circ$ is reduced, 
which excludes the case $BC_r$. Therefore $\g^\circ$ is of tube type 
(cf.\ Proposition~\ref{prop:herm}) and Table~$2$ thus implies that 
$h$ is symmetric.

The fact that $\g^\circ$ is of tube type implies that 
$x \in \fa^\circ$ corresponds to the unique Euler element 
$h_r$ for the restricted root system $\Sigma^\circ$ of type 
$C_r$ (see \eqref{eq:eulelts2}). From \eqref{eq:symmeuler} it now 
follows that $x$ is symmetric (see also Proposition~\ref{prop:herm}). 
This proves (i).

To verify (ii) and (iii), we observe that 
the root system $C_r$ contains the maximal subset
$\{2 \eps_1, \ldots, 2 \eps_r\}$ of strongly orthogonal roots, 
i.e., neither sums nor differences of these roots are roots.
The multiplicities of these restricted roots 
are $1$ (\cite[Thm.~XII.1.14]{Ne99}), and 
\[ \fs 
:= \bigoplus_{j = 1}^r (\g^\circ_{2\eps_j} + \g^\circ_{-2\eps_j} 
+ \R (2\eps_j)^\vee) 
= \fa^\circ \oplus \bigoplus_{j = 1}^r (\g^\circ_{2\eps_j} + \g^\circ_{-2\eps_j}) 
\cong \fsl_2(\R)^r \]
(cf.~\cite[Lemma~XII.1.11]{Ne99}, \cite[p.~12]{Ta79}). 
As the roots $2\eps_j$ all take the value $1$ on the Euler element 
$x \in \fa^\circ$, 
we have $x = \shalf \sum_{j =1}^r (2\eps_j)^\vee$, which is the diagonal 
element in $\fsl_2(\R)^r$, corresponding to  $\pmat{\shalf & 0 \\ 0 &-\shalf}$.
Likewise, $ih$ is contained in $\fs \cong \fsl_2(\R)^r$ 
as the diagonal element corresponding to 
$\pmat{0 & -\shalf \\ \shalf &0}$. As the Lie subalgebra of 
$\gl_2(\C)$, generated by 
\[ \pmat{0 & -\shalf \\ \shalf &0} \quad \mbox{ and } \quad 
\pmat{\shalf & 0 \\ 0 &-\shalf} \] 
is isomorphic to $\fsl_2(\R)$, the same holds for the 
real Lie subalgebra of $\g$ generated by $h$ and $x$. Now (ii) and (iii) 
follow from Lemma~\ref{lem:sl2}.

\nin {\bf Case 2: $\g$ is not complex:} 
Then $\g_\C$ is a simple complex Lie algebra to which all arguments 
in Case 1 apply. In particular, 
the real Lie subalgebra $\fs$ spanned by $h$, $x$ and $[h,x]$ 
is isomorphic to $\fsl_2(\R)$. This proves (ii) and (iii). 
As $\fs \subeq \g$ and all Euler elements in $\fsl_2(\R)$ 
are symmetric, we also obtain (i). 

\nin (b) If  there exists an Euler element $x$ for which 
$(h,x)$ is orthogonal, then (a)(i) implies that $h$ is a symmetric 
Euler element. 
Suppose, conversely, that $h$ is a symmetric Euler element. 
For a Cartan involution $\theta$ with $\theta(h) = -h$, 
we choose a maximal abelian subspace $\fa \subeq \fp = \Fix(-\theta)$ containing 
$h$ and choose in 
the subset 
\[ \Sigma_1 := \{ \alpha \in \Sigma(\g,\fa) \: \alpha(h) = 1\} \]
a maximal set $\{ \gamma_1, \ldots, \gamma_r\}$ of strongly 
orthogonal roots \red{(cf.\ \cite[p.13]{Ta79} or \cite[p.~134]{Kan00}). 
From these references we further infer the existence of elements }
$e_j \in \g_{\gamma_j}$ such that, 
for each $j$, the subalgebra 
$\fs_j := \Spann_\R \{e_j, \sigma_h(e_j), [e_j, \sigma_h(e_j)]\}$ is isomorphic 
to $\fsl_2(\R)$. We normalize $e_j$ in such a way that, 
for $x_j := [e_j, \sigma_h(e_j)]$, we have $\gamma_j(x_j) = 1$. 
\red{Then loc.\ cit.\ further implies that }
\[ \fa_\fq := \fa \cap \fq = \Spann \{ x_j \: j =1,\ldots, r\} 
\quad \mbox{ for } \quad \fq := \g^{-\sigma_h} \] 
is maximal abelian in $\fq_\fp$. 
Since $h$ is a symmetric Euler element and the root system 
$\Sigma(\g,\fa)$ is irreducible, $h$ corresponds to some 
$h_j$ in the list \eqref{eq:symmeuler}. 
The restricted root system $\Sigma(\g,\fa_\fq)$ is always of type $C_r$. 
The explicit description 
of the restricted roots in \cite[p.~596]{Kan98} now implies that 
$x := \sum_{j = 1}^r x_j \in \fa_\fq$ is an Euler element. 
By construction, it satisfies $\sigma_h(x) = -x$, so that 
$(h,x)$ is orthogonal. 
This completes the proof. 
\end{prf}

\begin{cor} \label{cor:gfinite}Let $\g$ be a finite dimensional Lie algebra 
and $(h,x)$ be orthogonal Euler elements such that $h$ is also symmetric. 
Then the following assertions hold: 
\begin{itemize}
\item[\rm(a)] There exists a Levi complement containing $h$ and $x$. 
\item[\rm(b)] The Lie algebra generated by $h$ and $x$ is isomorphic 
to $\fsl_2(\R)$. 
\item[\rm(c)] $(x,h)$ is also orthogonal. 
\end{itemize}
\end{cor} 

\begin{prf} By Proposition~\ref{prop:1.1}(i), 
there exists a Levi decompositions 
$\g = \fr \rtimes \fs$ with $h \in \fs$. We then have 
for $\fq := \Fix(-\sigma_h)$ the decompositions 
\[ \fq := \g_1(h) \oplus \g_{-1}(h) = \fq_\fr \oplus \fq_\fs \quad \mbox{ with }\quad 
\fq_\fr = \fq \cap \fr \quad \mbox{ and }\quad 
\fq_\fs = \fq \cap \fs,\] 
and 
$x \in \fq$ is an Euler element, hence in particular hyperbolic. 
Let $\fa_\fr \subeq \fq_\fr$ be a maximal hyperbolic subspace, 
i.e., $\fa_\fr$ is abelian, consists of $\ad$-diagonalizable elements 
and is maximal with respect to this property. 
Then $\fa_\fr \subeq [h,\fr] \subeq [\g,\fr]$ consists also of $\ad$-nilpotent 
elements, hence is central. As $\ad h\res_\fq$ is injective, 
it follows that $\fa_\fr= \{0\}$. 
By \cite[Prop.~III.5]{KN96}, $\fq_\fs$ contains a maximal 
hyperbolic subspace $\fa$ of $\fq$ and $x$ is conjugate under 
$\Inn_\g(\fh)$ to an element of $\fa \subeq \fq_\fs$. This proves(a). 

\nin (b) In view of (a), we may w.l.o.g.\ assume that $\g$ is semisimple,  
and by Theorem~\ref{thm:automaticsym}, which applies to each simple ideal,
 even that 
$\g \cong \fsl_2(\R)^r$ for some $r \in \N$. 
As $\Aut(\fsl_2(\R))\cong \PGL_2(\R)$ acts transitively on the set of 
orthogonal pairs of Euler elements in $\fsl_2(\R)$ 
(Example~\ref{ex:1.4}), we may further assume that 
\[ h = (h_0, \cdots, h_0) \quad \mbox{ and }\quad 
x = (x_0, \cdots, x_0)\quad \mbox{ for }\quad 
h_0 = \pmat{\shalf & 0 \\ 0 &-\shalf}, \quad 
x_0 := \pmat{0 & -\shalf \\ \shalf &0},\] 
so that the Lie subalgebra generated by $x$ and $h$ is the diagonal 
in $\fsl_2(\R)^r$, hence isomorphic to~$\fsl_2(\R)$.   

\nin (c) follows directly from (b) and Lemma~\ref{lem:sl2}.
\end{prf}

\section{Covariant nets of real subspaces}
\mlabel{sec:4}

In this section we develop an axiomatic setting for covariant 
nets of standard subspaces parametrized by $G^\up$-orbits in $\cG_E(G)$. 

\subsection{Standard subspaces}

Here we collect some fundamental notions concerning real subspaces 
of a complex Hilbert space $\cH$ with scalar product 
$\langle\cdot,\cdot\rangle$, linear in the second argument. 
We call a closed real subspace $\sH \subeq \cH$ \textit{cyclic} if $\sH+i\sH$ is dense in $\cH$, \textit{separating} if $\sH\cap i\sH=\{0\}$, and \textit{standard} 
if it is cyclic and separating. The symplectic ``complement'' 
of a real subspace $\sH$ is defined by the symplectic form $\Im \langle\cdot,\cdot\rangle$, namely
\[ \sH'=\{\xi\in\cH:(\forall \eta \in \sH)\ 
 \Im\langle\xi,\eta \rangle=0 \}.\] 
Note that $\sH$ is separating if and only if $\sH'$ is cyclic, hence $\sH$ is standard if and only if $\sH'$ is standard.
For a standard subspace $\sH$, we 
define the {\it Tomita operator} as the closed antilinear involution
\[S_\sH:\sH+i\sH \to \sH + i \sH,\quad 
\xi+i\eta\mapsto \xi-i\eta. \] 
The polar decomposition $S_\sH=J_\sH\Delta_\sH^{\frac12}$ defines an 
antiunitary involution $J_\sH$ and the modular operator~$\Delta_\sH$. 
For the modular group  $(\Delta_\sH^{it})_{t \in \R}$, 
we then have 
\[  J_\sH\sH=\sH' \quad \mbox{ and }\quad 
 \Delta^{it}_\sH\sH=\sH \qquad \mbox{ for every  } \quad 
t\in \R\]
(\cite[Thm.~3.4]{Lo08}). This construction 
leads to a one-to-one correspondence between Tomita operators and 
standard subspaces: 

\begin{proposition}\label{prop:11}{\rm (\cite[Prop.~3.2]{Lo08})} 
The map $\sH\mapsto S_\sH$ 
is a bijection between the set of standard subspaces of $\cH$ 
and the set of closed, densely defined, antilinear involutions on $\cH$. 
Moreover, polar decomposition $S = J\Delta^{1/2}$ 
defines a one-to-one correspondence between such involutions and 
pairs $(\Delta, J)$, where $J$ is a conjugation and $\Delta >0$ 
selfadjoint with $J\Delta J =\Delta^{-1}$. 
\end{proposition}
The modular operators of symplectic complements satisfy the following relations 
\[ S_{\sH'}=S_{\sH}^*, 
\qquad \Delta_{\sH'}=\Delta_\sH^{-1}, \qquad J_{\sH'}=J_{\sH}.\] 

From Proposition~\ref{prop:11} we easily deduce: 
{\begin{lemma}\label{lem:sym}{\rm(\cite[Lemma 2.2]{Mo18})}
Let $\sH\subset\cH$ be a standard subspace  and $U\in\AU(\cH)$ 
be a unitary operator. Then $U\sH$ is also standard and 
$U\Delta_\sH U^*=\Delta_{U\sH}^{\eps(U)}$ and $UJ_\sH U^*=J_{U\sH}$.
\end{lemma}}

\begin{lemma}\mlabel{inc}{\rm(\cite[Cor.~2.1.8]{Lo08})}
Let $\sH\subset \cH$ be a standard subspace, and $\sK\subset \sH$ 
be  a closed, real linear subspace of $\sH$. 
If $\Delta_\sH^{it}\sK=\sK$ for all $t\in\R$, then $\sK$ 
is a standard subspace of $\cK:= \overline{\sK+i\sK}$ 
and $\Delta_\sH |_\sK$  is the modular operator of $\sK$ on $\cK$. 
If, in addition, $\sK$ is a cyclic subspace of $\cH$, then $\sH=\sK$.
\end{lemma}

The following theorem relates 
positive generators and inclusions of real subspaces.

\begin{theorem}{\rm(\cite[Thms.~3.15, 3.17]{Lo08}, \cite[Thm.~3.2]{BGL02})} 
\label{Borch} 
Let $\sH\subset\cH$ be a standard subspace and $U(t)=e^{itP}$ 
be a unitary one-parameter group on $\cH$ with  a generator $P$.
\begin{itemize}
\item[\rm(a)] 
If $\pm P > 0$ and $U(t)\sH\subset \sH$ for all $t\geq 0$, then 
\begin{equation}\label{eq:DeltaJ} \Delta_\sH^{-is/2\pi}U(t)\Delta_\sH^{is/2\pi} 
= U(e^{\pm s}t)\quad \mbox{ and } \quad 
J_{\sH}U(t)J_{\sH}=U(-t) \quad \mbox{ for all } \quad 
t,s \in \R.
\end{equation}
\item[\rm(b)] If $\Delta_\sH^{-is/2\pi}U(t)\Delta_\sH^{is/2\pi} 
= U(e^{\pm s}t)$ for $s,t \in \R$, then  the following are equivalent:
\begin{enumerate} 
\item[\rm(1)] $U(t)\sH \subset \sH$ for $ t\geq 0$;
\item[\rm(2)] $\pm P$ is positive.
\end{enumerate}
\end{itemize}
\end{theorem}

Part {\rm(a)}  is also called the One-particle Borchers Theorem.  
Borchers originally proved it for von Neumann algebras with a cyclic and separating vectors. Part (b) is in \cite{BGL02}. 

With the notation introduced in Examples~\ref{ex:models}(b),  we have 
seen  that any couple $(U,\sH)$ of a one-parameter group 
$(U_t)_{t \in \R}$ 
with positive (resp.~negative) generator and a standard subspace $\sH$ satisfying the assumptions of Theorem~\ref{Borch}(a) defines a unitary, positive energy representation of the affine group $\Aff(\R) \cong \RR\rtimes\RR^\times$ implemented by 
\[ U(\zeta(t)) = U(t), \quad U(\delta(t))=\Delta_\sH^{-\frac{it}{2\pi}},\quad 
U(r_0)=J_\sH \quad \mbox{ for } \quad t \in \R.\]
A representation of $\Aff(\R)$ can also  be obtained by looking at some peculiar relative positions of standard subspaces: The half-sided modular inclusions.
\begin{defn}
An inclusion $\sK \subeq \sH$ of standard subspaces of 
$\cH$ is called a {\it $\pm$half-sided modular inclusion} ($\pm$HSMI) if 
\[ \Delta_\sH^{-it} \sK \subeq \sK \quad \mbox{ for } \quad \pm t \geq 0.\] 
\end{defn}

\begin{theorem}{\rm(\cite[Cor.~3.6.6.]{Lo08}, \cite[Thm.~3.15]{NO17})}
\label{thm:HSMI}
 $\sK \subeq \sH$ is a positive 
half-sided modular inclusion if and only if there exists an 
(anti-)unitary positive energy representation $(U,\cH)$ of 
$\Aff(\R) \cong \R \rtimes \R^\times$ 
with  $U(\delta(t))=\Delta_\sH^{-\frac{it}{2\pi}}$, $U(r_0)=J_\sH$, $U(\delta_{(1,\infty)}(t))=\Delta_\sK^{-\frac{it}{2\pi}}$, $U(r_1)=J_\sK$. In this picture
\[ \sK = \sN((1,1).W_0) \quad \mbox{ and } \quad 
\sH = \sN(W_0),\] where $W_0=(\lambda,r_0)$ corresponds to 
the half-line $(0,\infty)$, the translations satisfy 
\[ \sK = U(1)\sH \quad \mbox{ and } \quad 
{U(1-e^{t})=\Delta_\sK^{-it/2\pi}\Delta_\sH^{it/2\pi}}
\quad \mbox{ for } \quad t\in\R. \] 
\end{theorem}

As a consequence, negative
half-sided modular inclusions $\sK \subeq \sH$ are in in 1-1 correspondence with (anti-)unitary negative energy representation $(U,\cH)$ of 
$\Aff(\R) \cong \R \rtimes \R^\times$ with 
\[ \sK = \sN((-1,1).W_0') = U(-1)\sN(W_0)' 
\quad \mbox{ and } \quad 
\sH = \sN(W_0') = \sN(W_0)'\] 
and with  $U(\delta(-t))=\Delta_\sH^{-\frac{it}{2\pi}}$, $U(r_0)=J_\sH$, $U(\delta_{(-\infty,-1)}(t))=\Delta_\sK^{-\frac{it}{2\pi}}$, $U(r_{-1})=J_\sK$.

\begin{corollary} {\rm(\cite[Corollary 2.4.3.]{Lo08}).}  
If $\sK\subset \sH$ is +HSMI, then $\sH'\subset \sK'$ is -HSMI 
\end{corollary}

\subsection{The axiomatics of abstract covariant nets} 

Hereafter we will make the following assumption on the group $G$.
\begin{assumption}\label{as:group}
{We assume that $\cG_E(G) \not=\eset$ and 
write $G=G^\up\rtimes\{e,\sigma\}$ for some  Euler involution~$\sigma$.}
\end{assumption}

\begin{ex} Note that $G^\down$ may contain involutions which are 
not Euler. 

We consider the graded Lie group $G := \SO_{1,n}(\R)$ with the 
identity component $G^\up = \SO_{1,n}(\R)^\up$. 
For $n \geq 2$, the Lie algebra $\g = \so_{1,n}(\R)$ is simple, 
$\theta(x) = - x^\top$ is a Cartan involution, 
and $\fa := \so_{1,1}(\R) \subeq \fp$ (acting on the first two components) 
is a maximal abelian subspace. 
As the corresponding restricted root system is of type $A_1$, 
our classification scheme (see 
\eqref{eq:symmeuler}  in Theorem~\ref{thm:classif-symeuler}) implies that 
all Euler elements in $\g$ are conjugate to the one corresponding to the 
boost generator 
\[ h(x_0, \ldots, x_n) = (x_1, x_0, 0, \ldots, 0).\] 
Accordingly, an involution $\sigma \in G$ is Euler if and only if 
$\sigma$ or $-\sigma$ is the orthogonal reflection in a $2$-dimensional 
Lorentzian plane. 

However, $G^\down$ contains all reflections of the type 
\[ \tau(x) = (\eps_0 x_0, \ldots,\eps_n x_n) \quad \mbox{ with } \quad 
\eps_j \in \{\pm 1 \}\quad \mbox{ satisfying } \quad  \prod_{j = 0}^n \eps_j = 1.\] 
In particular neither $\Fix(\tau)$ nor $\Fix(-\tau)$ must have 
dimension~$2$. 
\end{ex}

{We now present the analogs of the one-particle Haag--Kastler axioms and further fundamental properties in our general setting. }

\begin{defn}\label{def:net} 
Let {$G = G^\up \rtimes \{e,\sigma\}$ be as above, 
$C \subeq \g$ be a closed convex $\Ad^\eps(G)$-invariant cone in $\g$, 
and fix a $G^\up$-orbit $\cW_+ = G^\up.W \subeq \cG_E(G)$.}
Let $(U,\cH)$ be a unitary representation of $G^\up$  and 
\begin{equation}
  \label{eq:net}
\sN \: \cW_+ \to \Stand(\cH) 
\end{equation}
be a map, also called a {\it net of standard subspaces}. 
In the following we denote this data as $(\cW_+,U,\sN)$.
We consider the following properties: 
\begin{itemize}
\item[\rm(HK1)] {\bf Isotony:} $\sN(W_1) \subeq \sN(W_2)$ for $W_1 \leq W_2$. 
\item[\rm(HK2)] {\bf Covariance:} $\sN(gW) = U(g)\sN(W)$ for 
$g \in {G}^\up$, $W \in \cW_+$. 
\item[\rm(HK3)] {\bf Spectral condition:} 
$ C \subeq C_U := \{ x \in \g \: -i \partial U(x) \geq 0\}.$ 
We then say that $U$ is {\it $C$-positive}.  
\item[\rm(HK4)]{\bf Central twisted locality:} 
{For $\alpha \in Z(G^\up)^-$ 
and $W \in \cW_+$ with $W^{'\alpha}\in\cW_+$, there exists a 
unitary $Z_\alpha\in U(G^\uparrow)'$ satisfying 
\begin{equation}
  \label{eq:ZJ}
Z_\alpha^2 = U(\alpha) \quad \mbox{ and } \quad 
J_{\sN(W)} Z_\alpha J_{\sN(W)}=Z_\alpha^{-1}, 
\end{equation}}
such that 
\begin{equation}\label{eq:tcl}\sN(W^{'\alpha}) \subset Z_\alpha \sN(W)'.
\end{equation} 
Moreoever, such an $\alpha$ exists.
\end{itemize}
When $Z_\alpha$ is trivial, for instance when 
$\partial(G^\uparrow_{\underline W})=\{e\}$, then the central twisted locality reduces to the more familiar locality relation.
\begin{itemize}
\item[\rm(HK4$_e$)]{\bf Locality:} 
{If $W \in \cW$ is such that $W'\in\cW_+$, then $\sN(W') \subset \sN(W)'.$}
\end{itemize}

Concerning (HK3), note that $C_U$ is pointed if and only if $\ker(U)$ is discrete. 
Therefore the assumption that $C$ is pointed is compatible with the 
possible existence of representations with discrete kernel 
satisfying (HK3). Furthermore, if $C=\{0\}$, then (HK3) trivially holds.

The following property will be central in our discussion 
because it connects the modular groups of standard subspaces to the 
unitary representation $U$ of~$G^\up$. 

\begin{itemize}
\item[\rm(HK5)] {\bf Bisognano--Wichmann (BW) property:} 
$U(\lambda_W(t)) = \Delta_{\sN(W)}^{-it/2\pi}$ for 
$W \in \cW_+, t \in \R$.
\end{itemize} 
We will see {in Proposition~\ref{prop:bwdual}} 
that a consequence of {(HK1-5)} is the following 
{stronger form of (HK4):}
\begin{itemize}
\item[(HK6)] \textbf{Central Haag Duality}: $\sN(W^{'\alpha})=Z_\alpha \sN(W)'$ 
{for $\alpha \in Z(G^\up)^-$, $W \in \cW_+$ with $W^{'\alpha}\in\cW_+$ 
and $Z_\alpha$ as in \eqref{eq:ZJ}.}
\end{itemize}
If the representation $U$ extends antiunitarily to $G$ we can further 
require: 
\begin{itemize}
\item[(HK7)] \textbf{G-covariance:} 
For any $\alpha \in Z(G^\up)^-$ such that $W^{'\alpha}\in\cW_+$, there exists an (anti-)unitary extension $U^\alpha$  of $U$ from $G^\up$ to $G$ such that the following condition is satisfied:
\begin{equation}
  \label{eq:twitcovar}
 \sN(g *_\alpha W) = U^\alpha(g) \sN(W) \quad \mbox{ for } \quad 
g \in G,
 \end{equation}
where $*_\alpha$ is the $\alpha$-twisted action \eqref{eq:twistact} of 
$G$ on $\cW_+$ defined in Lemma~\ref{lem:twistact}(e). 
\end{itemize}
It is enough to provide an extension $U^\alpha$ w.r.t.~one $\alpha \in Z(G^\up)^-$ such that $W^{'\alpha}\in\cW_+$. \blue{ All the other extensions come as described in Lemma~\ref{lem:twistact}(f).  }
The modular conjugation of {standard subspaces} 
can have a geometric meaning when the extension $U^\alpha$ from (HK7) has 
the following specific form:
\begin{itemize}
\item[(HK8)] \textbf{Modular reflection:} 
$U^\alpha(\sigma_W)= Z_\alpha J_{\sN(W)}$ for 
$\alpha \in Z(G^\up)^-$, $W \in \cW_+$ with $W^{'\alpha}\in\cW_+$ 
and  $Z_\alpha$ as in  \eqref{eq:ZJ}. 
\end{itemize}

\end{defn}

In the next sections we will show that there exist {nets 
of standard subspaces} satisfying all the above assumptions. 
It is the analog of the BGL construction in this general setting.

\subsubsection{Wedge isotony and half-sided modular inclusions} 

Taking the wedge modular inclusion defined in 
Section~\ref{sect:wedgeinc} into account, we now prove that
isotony can be deduced from covariance, the 
Bisognano--Wichmann property and the $C$-spectral condition. On 
specific models this has been checked in  \cite{BGL02, Lo08}.

\begin{proposition}\label{prop:isotony}
Let $(\cW_+,\sN,U)$ be  a net of  standard subspaces. 
Then the spectral condition {\rm(HK3)}, 
the BW property {\rm(HK5)} and 
$G^\up$-covariance {\rm(HK2)} imply isotony {\rm(HK1)}. 
\end{proposition}

\begin{prf} Let $W_0 = (h,\tau) \in \cG_E$ and $\sH_0 = \sN(W_0)$. 
By covariance, the net $\sN$ is isotone if and only if 
\[ \cS_{W_0} = G^\up_{W_0} \exp(C_+) \exp(C_-) 
\subeq \cS_{\sH_0} := \{ g \in G^\up \: U(g)\sH_0 \subeq \sH_0\}.\] 
As the stabilizer $G^\up_{W_0}$ stabilizes $\sH_0$ by covariance, 
isotony is equivalent to $\exp(x) \in \cS_{\sH_0}$ for every 
$x \in C_+ \cup C_-$. 

By the spectral condition (HK3), we have $\mp i \partial U(x) \geq 0$. 
Therefore {Theorem~\ref{Borch}}  
shows that isotony is equivalent to 
\begin{equation}
  \label{eq:borrel}
 U^{\sH_0}(e^s) U(\exp tx) U^{\sH_0}(e^{-s})
= U(\exp e^{\pm s} tx) 
\quad \mbox{ for } \quad 
s,t \in \R, x \in C_\pm.
\end{equation}
By the BW property (HK5), 
$U^{\sH_0}(e^s)= \Delta_{\sH_0}^{-is/2\pi} = U(\exp s h)$, 
so that $[h,x]= \pm x$ for $x \in C_\pm$ implies 
\eqref{eq:borrel}. 
\end{prf}

\subsubsection{The Brunetti--Guido--Longo (BGL) construction} 

We have seen in the introduction to Section~\ref{sec:1} that 
each standard subspace $\sH$ specifies a homomorphism 
\begin{equation}
  \label{eq:uv-rep}
 U^\sH \: \R^\times \to \AU(\cH)\quad \mbox{ by } \quad 
U^\sH(e^t) := \Delta_\sH^{-it/2\pi}, \quad 
U^\sH(-1) := J_\sH,
\end{equation}
and that this leads to a bijection 
\[ \Phi \: \Hom_{\rm gr}(\R^\times,\AU(\cH)) \to \Stand(\cH), \quad 
U^\sH \mapsto \sH \] 
between continuous (anti-)unitary representations of the graded 
Lie group $\R^\times$ and standard subspaces (\cite[Prop.~3.2]{NO17}). 
By Lemma~\ref{lem:sym}, 
$\Phi$ is equivariant with respect to the 
natural action of $\AU(\cH)$ on $\Stand(\cH)$ and the action 
\eqref{eq:symact} on $\Hom_{\rm gr}(\R^\times,\AU(\cH))$. 

 Now every (anti-)unitary representation $U \:  G \to \AU(\cH)$ defines 
by composition a natural $G$-equivariant map 
\[ \cG \smapright{\Psi^{-1}} \Hom_{\rm gr}(\R^\times,G) \ssmapright{U \circ} 
 \Hom_{\rm gr}(\R^\times,\AU(\cH)), \qquad 
W \mapsto U \circ \gamma_W. \] 
Combining this with $\Phi$ leads to the 
so-called Brunetti--Guido--Longo (BGL) construction:

\begin{defn}(Brunetti--Guido--Longo (BGL) net)  
If $(U,G)$ is an (anti-)unitary representation,  
then we obtain a $G$-equivariant map 
$\sN_U \:  \cG \to \Stand(\cH)$ 
determined for $W = (k_W, \sigma_W)$ by 
\begin{equation}
  \label{eq:bgl}
J_{\sN_U(W)} = 
U(\sigma_W) \quad \mbox{ and } \quad \Delta_{\sN_U(W)}^{-it/2\pi} 
= U(\exp t k_W) 
\quad \mbox{ for } \quad t \in \R.
\end{equation}
This means that, with respect to Definition \ref{def:1.1}, $U^{\sN_U(W)} = U \circ \gamma_W$ for $W \in \cG$ 
(see \cite{BGL02}, \cite[Prop.~5.6]{NO17}). 

The BGL net associates to every wedge $W\in\cG$ a standard subspace $\sN_U(W)$. We shall denote with $(\cW_+,\sN_U,U)$ the restriction of the BGL net to 
the $G^\up$-orbit $\cW_+ \subeq \cG_E(G)$.
\end{defn}

\begin{theorem}
The restriction of the BGL net $\sN_U$ 
associated to an (anti-)unitary $C$-positive  
representation $U$ of $G=G^\uparrow\rtimes\{e,\sigma\}$ 
to a $G^\up$-orbit $\cW_+\subeq \cG_E$ satisfies all the axioms 
{\rm(HK1)-(HK3)} and {\rm (HK5)}. 
\end{theorem}

\nin We shall see in Proposition \ref{prop:BGLpart2} 
 that the twisted locality (HK4), Central Haag Duality (HK6) 
 and (HK7-8) are also satisfied.

\begin{proof} Let $\cW_+ \subeq \cG_E(G)$ be a $G^\up$-orbit. 
By construction, the restriction of the BGL net $\sN_U$ to $\cW_+$ 
satisfies (HK2) and by  construction it satisfies  (HK5).
By Proposition~\ref{prop:isotony}, isotony (HK1) 
follows from the Spectral Condition (HK3), which is the $C$-positivity of~$U$. 
\end{proof}

As a last remark in this section we stress that two (anti-)unitary extensions of a unitary representation $(U,\cH)$ of $G^\up$ are unitarily equivalent, 
\red{but the corresponding BGL nets depend} 
on the choice of the (anti-)unitary extension. 
The following proposition makes this dependence explicit 
and provides a natural parameter space.

\begin{prop} \mlabel{prop:1.9} {\rm(The space of (anti-)unitary extensions)} 
Fix $(h,\tau) \in \cG$, 
let $U \: G \to \AU(\cH)$ be an (anti-)unitary representation 
and let $\cM := U(G^\up)'$. Then the following assertions hold: 
\begin{itemize}
\item[\rm(i)] All (anti-)unitary representations 
$(\tilde U,\cH)$ extending $U\res_{G^\up}$ are of the form 
$\tilde U = T U T^{*}$ for some $T\in \U(\cM)$. 
{The corresponding BGL nets are related by 
\begin{equation}
  \label{eq:bglrel}
 \sN_{\tilde U}(W) = T\sN_U(W) \quad \mbox{ for } \quad W \in \cG.
\end{equation}}
\item[\rm(ii)] {\rm(Parametrization of (anti-)unitary extensions)} 
Let $J := U(\tau)$, $\tau \in G^\downarrow$. For every 
\[ N \in \U(\cM)^{-} := \{ M \in \U(\cM) \:  J M J = M^{-1} \}, \] 
there exists a unique (anti-)unitary extension $\tilde U$ of $U\res_{G^\up}$ with 
$\tilde U(\tau) = N J$, and we thus obtain a bijection between 
the set $\U(\cM)^-$ and the set of 
(anti-)unitary extensions of $U\res_{G^\up}$ to~$G$.
\end{itemize}
\end{prop}

\begin{prf} (i) follows from Proposition~\ref{prop:antiuniext} 
and the assertion on the BGL nets is an immediate consequence of the 
definitions. 

\nin (ii) Let $T \in \U(\cM)$, so that 
$\tilde U = T U T^{-1} \: G \to \AU(\cH)$ 
is an (anti-)unitary extension of $U\res_{G^\up}$ with 
$\tilde J := \tilde U(\tau) = T J T^{-1}$. 
Since $\tilde U$ and $U$ extend the same representation of $G^\up$, 
\[ N := \tilde J J = \tilde J J^{-1} \in \U(\cM).\] 
This element satisfies 
$J N J = J \tilde J = N^{-1}$, so that $N \in \U(\cM)^-$ and 
$\tilde J = NJ$. 

If, conversely, $N \in \U(\cM)^-$, 
then Lemma~\ref{lem:symspace} 
implies the existence of an $X = - X^* \in \cM$ with 
$N = e^{2X}$ and $J X J = - X$. For 
$T := e^X \in \U(\cM)$ and $\tilde J := T J T^{-1}$, we then have 
\[ \tilde J J =  T J T^{-1} J  = T^2 = e^{2X} = N.\] 
Therefore the manifold $\U(\cM)^-$ parametrizes the 
(anti-)unitary extensions of~$U\res_{G^\up}$.   
\end{prf}

\subsubsection{Twisted Locality}

{ We have seen in Section~\ref{sect:alftwist} that 
it can happen that $W'\notin\cW_+ = G^\up.W$. 
One can anyway attach to $W'$ a real subspace by the BGL-net and  by construction obtain the relation $\sH(W')=\sH(W)'$. On the other hand one can define 
natural complementary wedges $W^{'\alpha}$ 
indexed by central elements~$\alpha$. 
In this section we will see that in the BGL construction, 
the complementary wedge subspaces satisfy the central Haag duality condition (HK6), hence the twisted locality relation (HK4).}
We start with a lemma on standard subspaces.

\begin{lem} \label{lem:twistsub}Let $\sH\subset\cH$ be a standard subspace, and $U\in \U(\cH)$ be a unitary operator  commuting with $\Delta_\sH$ and satisfying $J_\sH UJ_\sH = U^{-1}$. Let $\sH_1$ be the standard subspace defined by $(\Delta_\sH, UJ_\sH).$ There exists a unitary square root $Z$ of $U$ commuting with $\Delta_\sH$ such 
that $J_\sH ZJ_\sH = Z^{-1}$ and $Z\sH=\sH_1$. 
The standard subspace $\sH_1$ does not depend on this choice of $Z$.
\end{lem}

\begin{proof}The existence of the square root and the commutation relation with the modular conjugation and the modular operator follows by Lemma~\ref{lem:app.1}. Then 
\[ Z (J_\sH \Delta_\sH^{1/2}) Z^{-1} 
= ZJ_\sH Z^{-1} \Delta_\sH^{1/2} 
= Z^2 J_\sH \Delta_\sH^{1/2} 
= U J_\sH \Delta_\sH^{1/2} \] 
implies that $\sH_1 = Z \sH$. It is clear 
that $\sH_1$ does not depend on the choice of $Z$.
\end{proof}

In order to conclude (HK6), hence the central locality condition on a 
BGL net $\sN_U$, we will 
need an analogous statement relating complementary wedge subspaces. 

\begin{prop} \mlabel{prop:4.14} 
{Let $(U,\cH)$ be an (anti-)unitary representation of the graded group 
\break $G=G^\uparrow\rtimes\{e,\sigma\}$ 
and $\alpha \in Z(G^\up)^-$}. 
Then the commutant $U(G^\up)'$ contains 
a unitary square root $Z_\alpha$ of $U(\alpha)$ satisfying 
\begin{equation}
  \label{eq:rootcond}
{U(g) Z_\alpha U(g)^{-1} = Z_\alpha^{-1}\quad \mbox{ for every } 
\quad g\in G^\down.}
\end{equation}
\end{prop}

\begin{prf}{First we note that $U(\alpha) \in \cM := U(G^\up)'$.}
We fix $\sigma_0 \in \Inv(G^\down)$ 
and observe that conjugation with $U(\sigma_0)$ defines an 
antilinear isomorphism $\beta$ of $\cM$. 
As $\beta(U(\alpha)) = U(\alpha)^{-1}$ follows from 
$\alpha \in Z(G^\up)^-$, we find with 
Lemma~\ref{lem:app.1}(c) in the appendix, 
a unitary square root $Z_\alpha$ of $U(\alpha)$ satisfying 
\begin{equation}
U(\sigma_0) Z_\alpha U(\sigma_0) = \beta(Z_\alpha)= Z_\alpha^{-1}. 
\end{equation} 
For any other  $\sigma \in G^\down$ we have
$\sigma = \sigma_0g$ with $g \in G^\up$, so that 
\[ 
U(\sigma) Z_\alpha U(\sigma) 
= U(\sigma) Z_\alpha U(\sigma)^{-1}  
= U(\sigma_0) U(g) Z_\alpha U(g)^{-1} U(\sigma_0) 
= U(\sigma_0) Z_\alpha U(\sigma_0) 
= Z_\alpha^{-1}. \qedhere\] 
\end{prf}

{We are now ready to verify that the BGL net is compatible 
with the twistings appearing in (HK4), (HK6) and (HK7).}

\begin{proposition} \label{prop:BGLpart2} 
For every (anti-)unitary 
representation $(U,\cH)$ of $G$, the BGL net $\sN_U$ satisfies 
{ {\rm(HK4)} and {\rm(HK6)}. 
Moreover, for $\alpha \in Z(G^\up)^-$, 
$W \in \cW_+$ with $W^{'\alpha} \in \cW_+$ and $Z_\alpha \in U(G^\up)'$ satisfying \eqref{eq:ZJ}, 
the (anti-)unitary extension 
$(U^\alpha, \cH)$ of $U\res_{G^\up}$ to $G$, 
determined by $U^\alpha(\sigma_W) := Z_\alpha U(\sigma_W)$, 
satisfies {\rm(HK7)} and {\rm(HK8)}.}
\end{proposition}

\begin{proof} {Let $\alpha \in Z(G^\up)^-$ and $W = (x,\sigma)
\in \cW_+$ be such that 
$W^{'\alpha} = (-x, \alpha\sigma)\in \cW_+$. 
Proposition~\ref{prop:4.14} implies the 
existence of $Z_\alpha \in U(G^\up)'$ satisfying \eqref{eq:ZJ}. 
Then 
\[ \Delta_{\sN_U(W^{'\alpha})}^{-it/2\pi} 
= U(\exp(- tx)) \quad \mbox{ and } \quad 
J_{\sN_U(W^{'\alpha})} = U(\alpha\sigma) 
= Z_\alpha^2 J_{\sN_U(W)} 
= Z_\alpha J_{\sN_U(W)} Z_\alpha^{-1} \] 
imply that 
$\sN_U(W^{'\alpha}) = Z_\alpha \sN_U(W)'.$ 
This shows that (HK6), hence also (HK4) are satisfied. 
We also have 
\[\sN_U(\sigma *_\alpha W)  = \sN_U(W^{'\alpha}) 
= Z_\alpha \sN_U(W)' = Z_\alpha U(\sigma) \sN_U(W).\] 
Since $\sN_U$ is $G$-equivariant on $\cG$, this leads for 
$g \in G^\up$ to 
\begin{align*}
\sN_U(g\sigma *_\alpha W) 
&= \sN_U(g.(\sigma *_\alpha W)) 
= U(g) \sN_U(\sigma*_\alpha W)
= U(g) Z_\alpha \sN_U(W)' \\
&= U(g) Z_\alpha U(\sigma) \sN_U(W) 
= U(g) U^\alpha(\sigma) \sN_U(W)
= U^\alpha(g\sigma) \sN_U(W).  
\end{align*}
This proves (HK7). As $J_{\sN_U(W)} = U(\sigma_W)$ by definition, we also 
have 
\[ U^\alpha(\sigma_W) = Z_\alpha U(\sigma_W) = Z_\alpha J_{\sN_U(W)},\] 
so that $U^\alpha$ also satisfies~(HK8). 
}
\end{proof}

\begin{rem} (a){If $U\res_{G^\up}$ is irreducible, then 
$U(Z(G^\up)) \subeq \T \1$, so that, we find 
for any $\alpha \in Z(G^\up)$ that 
$U(\alpha) = \zeta \1$ with $|\zeta| = 1$. We may thus put 
$Z_\alpha := z\1$ for any complex number $z$ with $z^2 = \zeta$. 
In this case $J Z_\alpha J = Z_\alpha^*$ holds for any antiunitary operator~$J$.}

\nin (b) {Let $(U,\cH)$ be an (anti-)unitary representation of~$G$.}
For any other square root $Z$ of $U(\alpha)$ satisfying the same 
requirements as $Z_\alpha$, the unitary operator $Z^{-1} Z_\alpha$ is an involution 
commuting with $U(  G)$, so that it leaves all standard subspaces 
$\sN(W)$ {of the BGL net} invariant. 

\nin (c) If $\alpha \in Z(G^\up)$ satisfies $\alpha^\sigma = \alpha$ 
for $\sigma \in G^\down$, then 
$\alpha$ acts trivially on $\cG(  G)$ and, by covariance of $\sN$, 
leaves all standard subspaces $\sN(W)$ invariant. 
This happens in particular if $\alpha^2 = e$. Then also 
$\alpha \in {Z(G^\up)^-}$, 
so that $\alpha$-twisted complements are useful in the context 
of fermionic theories. Here $U(\alpha)$ is an involution and one choice 
of a square root of $U(\alpha)$ is given by 
\begin{equation}\label{eq:ss}
 Z_{\alpha} :=  \frac{\1 + i U(\alpha)}{1+i}.
 \end{equation}
\end{rem}

Given a net satisfying (HK1)-(HK5), the 
commutation relation among twist operators and  the wedge modular operators immediately hold.

\begin{proposition} \mlabel{prop:cov}
Let $(\cW_+,U, \sN)$ be a 
$G$-covariant net satisfying 
{\rm(HK1)-(HK5)}, suppose that $U$ extends to an (anti-)unitary 
representation of $G$, and let $Z_\alpha \in U(G^\up)'$ as in \eqref{eq:ZJ}. Then, for every $W \in \cW_+$, we have 
\[ Z_\alpha \Delta_{\sN(W)} Z_\alpha^{-1} = \Delta_{\sN(W)}.\]
\end{proposition}

The latter proposition allows to conclude that (HK6) is a 
consequence of (HK1)-(HK5). 
\begin{proposition}\label{prop:bwdual}
Let $(\cW_+,\sN,U)$ be a net of standard subspace satisfying 
{\rm(HK1)}-{\rm(HK5)}. Then it also satisfies central Haag 
duality {{\rm(HK6)}}:
\[ \sN(W^{'\alpha})=Z_\alpha \sN(W)'
\quad \mbox{ for } \quad {\alpha \in Z(G^\up)^-, W \in \cW_+, 
W^{'\alpha} \in \cW_+. }
 \] 
In particular, the right hand side 
does not depend on the choice of $Z_\alpha$.
\end{proposition}

\begin{proof}
{By (HK5), the unitary operator $Z_\alpha \in U(G^\up)'$ 
commutes with the modular operator of $\sN(W)$, by Proposition 
\ref{prop:cov}. 
Therefore the two standard subspaces 
$\sN(W^{'\alpha})$ and $Z_\alpha \sN(W)'$ have the same modular 
operator. }
By twisted locality $\sN(W^{'\alpha}) \subeq Z_\alpha\sN(W)'$, so that 
Lemma~\ref{inc} implies that they coincide. 
\end{proof}

\begin{remark}
Let $(\cW_+,\sN,U)$ be  a net of standard subspaces with a unitary 
$C$-positive representation $(U,\cH)$ 
of $G^\uparrow$. Let $W_0=(x,\sigma)\in\cW_+\subset \cG_E$ and $\sH_0 := \sN(W_0)$. 
We claim that (HK1-3) imply that 
\[ \tilde U(\sigma) := J_{\sH_0} \] 
defines an (anti-)unitary extension of $U\res_{G^\up(W_0)}$ to the graded 
subgroup $G(W_0) = G^\up(W_0)\rtimes\{e,\sigma\}$ of~$G$. 
In fact, $J_{\sH_0}$ commutes with $G^\uparrow_{W_0}$ 
by Lemma \ref{lem:sym}. Further, the $C$-positivity 
and Theorem~\ref{Borch}(b) imply that it also has the correct 
commutation relation with $\exp(C_\pm)$, hence also with $G^\up(W_0)$. 
 We shall  see in Section~\ref{sect:sl2int}, when we actually obtain an extension to the full group~$G$. 
\end{remark}

\begin{example} (The Poincar\'e case) 
Let $\uG := \cP_+ := \R^{1,3} \rtimes \SO_{1,3}(\R)_0^\uparrow$ 
be the proper Poincar\'e group and 
\[  G = \tilde \cP_+ = \R^{1,d-1} \rtimes {\Spin_{1,3}(\R)_0}\] 
be its simply connected covering. 
We write $\lambda_\uW$ for the one-parameter group lifting 
the boost group $\Lambda_\uW$ associated to a wedge $\uW \in \cW 
= \uG.W_1$ 
(see e.g. \cite{Mo18}). For 
$\uG^\up$, a wedge is defined by a pair $\uW = (x,r_x)$, where $x$ generates 
$\Lambda_{\uW}$ and $r_x = e^{\pi i x}$ 
is the spacetime reflection in the direction of the wedge. 
Since $Z=Z( G)= \{ \pm \1\}$ is a $2$-element group, 
a wedge $W \in \uline\cG$ has two lifts 
which belong to two different 
$  G^\up$-orbits in $\cG(  G)$. 
To see this, we note that $Z = Z^-$  and $Z_2 = \{e\}$. 
For the second equality we \red{use the isomorphism} 
$\Spin_{1,3}(\R)$ with $\SL_2(\C)$ 
and note that the centralizer of any Euler element~$x$, 
which may be assumed to be 
$x = \pmat{\shalf & 0 \\ 0 & -\shalf}$,  
is connected and isomorphic to the multiplicative group $\C^\times$, 
on which the involution $\sigma_x$ acts trivially. 
Therefore the central elements $\partial(g) =  g^\sigma g^{-1}$, 
$g \in   G^\up_{(x,\sigma_x)}$, are all trivial, which leads to 
$Z_2 = \{e\}$. 

For $\alpha := -\1$, the twisted complement of 
$W = (k_W, \sigma_W)$ is $W^{'-1}=(-k_W, -\sigma_W)$. 
Any lift $\tilde r \: \R \to   G^\up$ of a rotation one-parameter 
group $\rho \: \R \to \SO_2(\R) \into \SO_{1,3}(\R)$ in $\uG^\up$ 
satisfying $\Ad(\rho(\pi)) k_W = - k_W$ now satisfies $\tilde \rho(2\pi) = -\1$. 
This shows that, $W' = (-k_W, \sigma_W) \not\in   G^\up.W$, but that 
$W^{'-1} = (-k_W, -\sigma_W) \not\in   G^\up.W$. 

Let $(U,\cH)$ be an irreducible  unitary positive 
energy representation of $  G^\up$ for which $U(-1) \not= \1$, 
then $U(-1) = - \1$ by Schur's Lemma. 
For the BGL net $\sN \: \cG(  G) \to \Stand(\cH)$ we therefore 
have $\sN(W^{'-1})=i\sN(W)'$ and $Z_\alpha = i\1$ is a suitable twist 
operator (cf.~\cite[Thm.~2.8]{Mo18}).
\end{example}

\begin{example}(Finite coverings of the M\"obius group) 
Consider the $n$-fold covering of the M\"obius group 
$  G^\up:=\Mob^{(n)} \subeq   G:=\Mob^{(n)}_2$, 
where $\uG = \Mob_2$ (cf.~Example~\ref{ex:models}(e)). 
This group is obtained 
from $\tilde\Mob_2$ by factorization of the subgroup $n Z(\tilde \Mob)$.
Then $Z := Z(G^\up) \cong \Z_n$ is a cyclic group of order~$n$. 
Let $\alpha := \tilde \rho(2\pi) \in  Z$ be a generator, 
where $\tilde\rho \: \R\to G^\up$ is the lift of the rotation 
group. 

Let $(U,\cH)$ be an (anti-)unitary 
representation of~$G$ whose restriction to $G^\up$ is 
irreducible.
Then, by Schur's Lemma, 
$U(\alpha^n) = U(\tilde\rho(\pi n))$ is an involution in $\T\1$, 
hence $\pm \1$. 
We now define {\it $n$-twisted local nets of real subspaces} as follows: 
\begin{itemize}
\item \textbf{$n$ is even.} As $\beta^\tau = \beta^{-1}$ for $\beta \in  Z$, we have 
$ Z^- =  Z$ and $ Z_1\cong \Z_{n/2}$ is a subgroup 
of index~$2$. As for $\tilde\Mob_2$, we have $ Z_2 =  Z_1$. 
We therefore obtain for every {Euler couple} $W = (x,\sigma){\in \cG_E(G)}$ 
two $  G^\up$-orbits $  G^\up.{(\pm x,\sigma)}$ covering $\uG^\up.\uW \subeq 
\cG_E(\uline G)$. 
Choosing $  G^\up.(x,\sigma)$, one obtains with the 
BGL construction a net of real subspaces $I\mapsto \sN(I)$, 
 where $I$ denotes an interval of length smaller than $2\pi$ 
in the $\frac n2$-covering $\bS^1_{(n/2)} {\simeq \R/\pi n\Z}$ of $\bS^1$. 
We can realize the net on intervals in $\bS^1_{(n/2)}$ because 
$U(\tilde\rho(n\pi))\sN(I)={\pm}\sN(I)=\sN(I)$. 
For the central element $\alpha=\tilde\rho(-2\pi)\in Z$, 
twisted complements look as follows. 
For $I=(a,b)\subset \R/\pi n\Z$ with $b-a < 2\pi$, we have 
$I^{'\alpha}= I^c$, where $I^c=(b-2\pi,a)$ is the ``complement'' obtained by 
conformal reflection on the left endpoint, cf.~\eqref{eq:left}. All the other twisted complement, belonging to the same orbit, are obtained by covariance.

{The locality relation then is given by 
$$\sN(I^{'\alpha})=\omega^k \sN(I)', \qquad k\in\mathbb Z,$$
where $\alpha=\tilde\rho(2\pi k)$ and $\omega\in \T$ satisfies 
$\omega^2 \1= U(\tilde\rho(2\pi))$. 
Since $U$ is irreducible and $Z$ is a cyclic group of order $n$, $U(\tilde\rho(2\pi))$ is an $n$-th 
root of the unity, hence $\omega^{2n} = 1$ and $Z_\alpha=\omega^k\1$.}

\item \textbf{$n$ is odd.}
Then $ Z^- =  Z_1$ implies that $  G^\up$ acts transitively 
on the inverse images of $\uG^\up$-orbits in $\uline\cG$. 
Fixing the orbit $G^\up.(x,\sigma)$, we  have by the 
BGL construction a net of real subspaces $I\mapsto \sN(I)$, 
 where, again, $I$ is an interval of length smaller than $2\pi$ 
in the $n$-fold covering of $\bS^1$.
Here the locality relation is 
$$\sN(I^{'\alpha})=\omega^k \sN(I)', \qquad k\in\mathbb Z,$$
where $\alpha=\tilde\rho(2\pi k)$ and  $\omega^{2n} = 1$, $I^{'\alpha}$ and $Z_\alpha$ are as above. 

\end{itemize}
\end{example}

\subsection{New models}

{Theorem~\ref{thm:classif-symeuler} provides 
 the list of restricted root systems for real simple 
Lie algebras  containing (symmetric) Euler elements, 
hence supporting (symmetric) Euler wedges. 
Any such Lie algebra $\g$ is the Lie algebra 
of a simply connected Lie group $G^\up$. 
Then  \eqref{eq:eul} defines an Euler involution on the group $G^\up$, 
so that we obtain the extension to $G=G^\uparrow \rtimes\{\id,\sigma\}$. }

{Such a Lie group $G^\up$ has many unitary representations, 
possibly with positive energy if the Lie algebra $\g$ is hermitian. 
By unitary induction, one can construct 
a unitary representation of $G^\up$ from 
a unitary representation of a subgroup, for instance from 
a covering of $\PSL_2(\R)\subset G^\up$ \cite{ma52}. 
It is always possible,  to extend a unitary representation $(U,\cH_U)$ 
of $G^\uparrow$ to an (anti-)unitary representation of $G$ 
by doubling the Hilbert space, if the representation  does not extend on 
 $\cH_U$ itself. 
Indeed,  we can choose any conjugation $C$ on $\cH_U$ and observe that the 
representation defined by 
$\tilde U(g) = U(g) \oplus CU(\sigma g\sigma)C$ on $\cH_U \oplus \cH_U$ 
extends to~$G$ by 
$U(\sigma)=\left(\begin{array}{cc}0&C\\C&0\end{array}\right)$. 
By the BGL construction 
there exists a (twisted-)local one-particle net satisfying (HK1-8).}

As a consequence we have the theorem:
\begin{theorem}\label{thm:newmod}
Let $\fg$ be a simple real Lie algebra containing an Euler element, 
i.e.,  whose restricted root system space occurs in 
{\rm Theorem \ref{thm:classif-symeuler}}. 
Then there exists a graded Lie group $G=G^\up\cup G^\down$ 
as in {\rm Section~\ref{sect:G}} with an (anti-)unitary representations $U$, 
and these in turn define twisted $G$-covariant 
BGL-nets $(\cW_+,U,\sN)$.
\end{theorem}

This theorem shows, for instance, that it is possible to associate a covariant homogeneous net of standard subspace to a Lie algebra $\fg$ 
with restricted root system~$E_7$.  
The subgroups $G_{\pm 1} = \exp(\g_{\pm 1}(x)) \subeq G^\up$ are closed, 
and if $G_0 := \{ g \in G^\up \: \Ad(g)x = x\}$, then so is 
${\bf P} := G_0 G_{-1}$. Then $M := G^\up/{\bf P}$ is a homogeneous space 
whose tangent space in the base point can be identified 
with the eigenspace $\g_1(x) = \ker(\ad x- \1)$. 
If $\g$ is simple hermitian \red{of tube type} 
and $C \subeq \g$ is a pointed generating 
invariant cone, then $C_+ := C \cap \g_1(x)$ 
defines a $G^\up$-invariant causal structure on $M$. 
The so-obtained manifolds include the Jordan spaces-times of 
G\"unaydin \cite{Gu93, Gu00, Gu01} and the simple spacetime manifolds 
in the sense of Mack--de Riese \cite{MdR07}. 
If the rank of the restricted root system $\Sigma$ of $(\g,\fa)$ 
is $2$, then $M$ is a Lorentzian manifold, but in general it is not. 
As a consequence of Proposition \ref{prop:herm} and Table 2, there 
exists a real form with a non-trivial positive cone 
\red{i.e., $\g$ is hermitian of tube type}, 
for every root system appearing in Theorem \ref{thm:classif-symeuler}.  Thus models with a proper notion of positive energy appear can be associated to every root system supporting \red{symmetric} Euler elements.

Recently, in \cite{NO20}  it has been shown that irreducible (anti-)unitary representations 
$(U,\cH)$ of $G$ which are of positive energy in the sense that 
$-i \partial U(y) \geq 0$ for $y \in C$, lead to 
$G$-covariant nets 
$(\sV_\cO)$ of real subspace of $\cH$, indexed by open subsets 
$\cO \subeq M$. If $\cO\not =\eset$, 
then $\sV_\cO$ is generating, and it is standard if $\cO$ is not 
``too big''. In particular, the open subset $\cO = \exp(C_+^0){\bf P} 
\subeq M$ corresponds to a standard subspace with the 
Bisognano--Wichmann property for which the modular group 
is represented by the one-parameter group $(\exp tx)_{t \in \R}$ of $G$ 
(see \cite[\S 5.2]{NO20}).

\subsection{The $\SL_2$-problem, symmetry extension} 
\label{sect:sl2int}

In Section~\ref{app:b.1} we have seen that the existence of orthogonal 
Euler wedges corresponds to the existence of an $\fsl_2$-subalgebra containing 
both Euler elements. In this section we will discuss when we can extend a covariant net of standard subspaces $(\cW_+,\sN, U)$ of Euler wedges 
satisfying (HK1)-(HK5) to a $G$-covariant net. 

We first look at the (anti-)unitary extensions of unitary 
representations of $\widetilde\SL_2(\R)$. In $\fsl_2(\R)$, we consider the two Euler elements 
\begin{equation}\label{eq:handk}
h := \frac{1}{2} \pmat{ 1 & 0 \\ 0 & -1} \quad \mbox{ and }\quad 
 k := \frac{1}{2} \pmat{ 0 & 1 \\ 1 & 0}.
 \end{equation} 
Let $(U,\cH)$ be a unitary representation 
of the group $G := \tilde\SL_2(\R)$ and consider the two selfadjoint 
operators 
\[ H := -2\pi i \partial U(h) \quad \mbox{ and } \quad 
K := -2\pi i \partial U(k).\] 

\begin{thm}
  \mlabel{thm:sl2-ext} 
Every continuous unitary representation of $\tilde\SL_2(\R)$ 
extends to an (anti-)unitary representation of the group 
\[\tilde \GL_2(\R) := \tilde\SL_2(\R) \rtimes \{\1, \tau_G\},\] 
where $\tau_G$ is the involutive automorphism of $\tilde\SL_2(\R)$ 
induced by the Lie algebra automorphism 
\[ \tau\pmat{ a& b \\ c &d} = \pmat{ a& -b \\ -c &d}\] 
corresponding to the Euler element~$h$.
\end{thm}

In \cite{GL95, Lo08} this  theorem was  proved for $\SL_2(\R)$-representations of the  principal and discrete series. Here the argument does not depend on the family of the representation. 

\begin{prf} Since $\tilde\SL_2(\R)$ is a type $I$ group, 
every unitary representation has a unique direct integral decomposition 
into irreducible 
unitary representations. 
This reduces the problem to the irreducible case. We have
to show that $U \circ \tau_G \cong U^*$ (the dual representation). 
Let 
\[ u := [h,k] = \frac{1}{2} \pmat{ 0 & -1 \\ 1 & 0}.\]
Then $h,k,u$ is a basis of $\fsl_2(\R)$ and 
\[ \omega := h^2 + k^2 - u^2 \in \cU(\fsl_2(\R))\] 
is a Casimir element, so that 
\[ \partial U(\omega) =  c\1 \quad \mbox{ for some }  \quad c \in \R.\] 
The antilinear extension $\oline\tau$ of $\tau$ to 
$\fsl_2(\R)$ satisfies $\oline\tau(iu) = iu$ and 
the operator $i\partial U(u)$ is selfadjoint and diagonalizable. 
We have 
\[ \partial U^*(u) 
= -\partial U(u) 
= \partial U(\tau(u)),\] 
so that $U^* \circ \tau_G$ is an irreducible with the 
same $u$-weights and the same Casimir eigenvalue~$c$. 
Below we argue that $U$ is uniquely determined by any pair $(\mu,c)$, 
where $\mu$ is an eigenvalue of $i\partial U(u)$ 
\red{occurring in the representation} (\cite{Sa67}, \cite{Lo08}),
and this implies that $U \circ \tau_G \cong U^*$.

To see that $U$ is determined by the pair $(\mu,c)$, we first recall that 
$\cH$ decomposes into one-dimensional eigenspaces of $i\partial U(u)$ and, 
by irreducibility, it is generated by  any eigenvector $\xi_\mu$ of eigenvalue $\mu$. 
Let $\cU(\g)$ denote the complex enveloping algebra of $\g$. Then 
$V_\mu := \cU(\g)\xi_\mu$ is a dense subspace consisting of analytic vectors, 
so that the representation 
$U$ is determined by the $\g$-representation on this space. In $\cU(\g)$ the centralizer 
$\cC_u$ of $u$ is generated by $u$ and the Casimir element. Therefore $\xi_\mu$ is a 
$\cC_u$-eigenvector and the corresponding homomorphism 
$\chi \: \cC_u \to \C$ is determined by $\chi(u) = \mu$ and $\chi(\omega) = c$. 
It is now easy to verify that these two values determine the $\cU(\g)$-module 
structure on $V_\mu$, hence the unitary representation~$U$. 
\end{prf}

\begin{rem}   
Here the determination of the representation is obtained by considering 
in the enveloping algebra $\cU(\fsl_2(\R))$, the centralizer 
subalgebra $\C[\omega,u]$ of $u$. Any 
cyclic weight vector $\xi_{\mu,c}$ defines a character $\chi$ of 
this subalgebra by $\chi(iu)  = \mu$ and $\chi(\omega) = c$, and 
$\cU(\g)\xi_{\mu,c}$ is isomorphic to the quotient of $\cU(\g)$ 
by the left module generated by $\mu \1- iu$ and $\omega - c\1$. 
\end{rem}

Now, we consider the positive selfadjoint operator 
\[ \Delta_h := e^{-H} = e^{2\pi i \partial U(h)}.\] 
By Theorem~\ref{thm:sl2-ext}, $U$ extends to an (anti-)unitary 
representation of $\tilde\GL_2(\R)$, and we put 
\[ J := U(\tau_G), \quad S := J \Delta_h^{-1/2} 
= J e^{-\pi i \partial U(h)}
\quad \mbox{ and }\quad \sV := \Fix(S).\]

\begin{lem} \mlabel{lem:standpair} 
For a unitary operator $T \in \U(\cH)$, 
the following assertions are equivalent: 
  \begin{itemize}
  \item[\rm(a)] $S  T S \subeq T$ holds on a dense subspace. 
  \item[\rm(b)] $T^{-1} \sV \cap \sV$ is standard. 
  \end{itemize}
If these conditions are satisfied, then {\rm(a)} holds on 
$T^{-1} \sV \cap \sV$.
\end{lem}

\begin{prf} If (b) holds, then any $\xi \in T^{-1}\sV \cap \sV$ satisfies 
$S T S \xi = S T \xi = T \xi,$ so that (a) holds. 
  
Conversely, assume that 
\[ \cD := \{ \xi \in \cD(S) \: S T S \xi = T \xi \} \] 
is dense in $\cH$. For any $\xi \in \cD$ we then have 
$T\xi \in \cR(S) = \cD(S)$ and 
\[ STS(S\xi) = ST \xi = S(STS\xi) = T(S\xi),\] 
so that $\cD$ is $S$-invariant. This implies that 
$\cD = (\cD \cap \sV) + i (\cD \cap \sV),$ 
so that $\cD \cap \sV$ is standard. 
For $\xi \in \sV$, we have $T\xi \in \sV$ if and only if 
$\xi \in \cD$, so that 
$\cD \cap \sV = T^{-1} \sV = \sV.$ 
This proves the lemma. 
\end{prf}

\begin{prop} \mlabel{prop:sl2standcond}
The following assertions are equivalent: 
  \begin{itemize}
  \item[\rm(a)] $\Delta_h^{-1/2} e^{itK} \Delta_h^{1/2} \subeq e^{-itK}$ 
holds for every $t \in \R$ on a dense subspace of $\cH$. 
  \item[\rm(b)] $S  e^{itK} S \subeq e^{itK}$ 
holds for every $t \in \R$ on a dense subspace of $\cH$. 
  \item[\rm(c)] $e^{-itK} \sV \cap \sV$ is standard for every $t \in \R$. 
  \end{itemize}
If these conditions are satisfied, then 
{\rm(a)} holds on $J(e^{-it K}\sV \cap \sV)$ and {\rm(b)} 
on $e^{-it K}\sV \cap \sV$. 
\end{prop}

\begin{prf} (a) $\Leftrightarrow$ (b): From $\tau(k) = -k$ it follows that 
\[ J U(\exp tk)J = U(\tau_G(\exp tk)) = U(\exp(-tk)),\] 
so that conjugating with $J$ translates (a) into (b). 

\nin (b) $\Leftrightarrow$ (c) follows from Lemma~\ref{lem:standpair}.  
\end{prf}
From \cite[Thm.~1.1, Cor.~1.3(c)]{GL95} one can deduce that 
the equivalent conditions 
in Proposition~\ref{prop:sl2standcond} are satisfied 
for  principal series representations and lowest and highest weight representations, {but it is not known for complementary series representations. }

{The following theorem  shows that an  
isotone, central twisted local $G^\up$-covariant net of standard subspaces satisfying the BW property  extends is actually $G$-covariant. The argument  needs the  density property described in Proposition \ref{prop:sl2standcond} for $\widetilde\SL_2(\R)$. 
The extension is done by (HK8).The proof generalizes the argument in \cite{GL95}.}

\begin{theorem} {\rm(Extension Theorem)} \label{thm:ext}
Let $G=G^\uparrow\rtimes\{\id,\sigma\}$ be a graded Lie group, 
where $\sigma$ is  an Euler involution. 
Let $(U,\cH)$ be a unitary $C$-positive representation of $G^\uparrow$, 
$\cW_+\subeq \cG_E(G)$ be a {$G^\up$-orbit}, and 
$(\cW_+, \sN,U)$ be a net of standard subspaces satisfying {\rm(HK1-4)} 
and  the BW property {\rm(HK5)}. 
If $h_1, \ldots, h_n$, $n \geq 2$, is a pairwise orthogonal family 
of Euler elements generating the 
Lie algebra~$\g$, and the conditions in {\rm Proposition~\ref{prop:sl2standcond}} 
hold for the representations of the connected subgroups corresponding to the 
$\fsl_2$-subalgebras generated by $h_1$ and $h_j$ for $j =2,\ldots, n$, 
 then $U$ extends to an (anti-)unitary representation of $G$ such that 
$G$-covariance {\rm(HK7)}  and  modular reflection {\rm(HK8)} hold.
\end{theorem}

\begin{proof} Let $(\cW_+,\sN,U)$ be a 
net of standard subspaces satisfying (HK1-5).
The Bisognano--Wichmann property (HK5) implies  
Central Haag Duality (HK6) by Proposition \ref{prop:bwdual}. 
Let  $H_j:=-i\partial U(h_j)$ be the selfadjoint generators of the 
unitary one-parameter group corresponding to~$h_j$.
By Corollary~\ref{cor:gfinite}, every pair $(h_1,h_j)$ 
generates a subalgebra isomorphic to $\fsl_2(\RR)$ and the generators 
$H_1$ and $H_j$ integrate to a representation  of $\tilde\SL_2(\R)$.  
Consider the Euler wedges 
$W_1, W_j\in \cW_+$ 
associated to $h_1$ and $h_j$, respectively. 

We claim that Proposition \ref{prop:sl2standcond} implies that 
$U(\sigma_{h_1}) := J_{\sN(W_1)}$ 
associated to the standard subspace $\sN(W_1)$ 
extends the $\tilde\SL_2(\R)$-representation to 
an (anti-) unitary representation of $\tilde\PGL_2(\R)$. 
Indeed, by Proposition~\ref{prop:sl2standcond}(b) we have that
\begin{equation}\label{eq:ext1}
\Delta_{h_1}^{-1/2} e^{itH_j} \Delta_{h_1}^{1/2} \subset J_{\sN(W_1)} e^{itH_j} J_{\sN(W_1)} 
\end{equation} 
on the dense domain $J_{\sN(W_1)}(e^{-itH_j} \sV \cap \sV)$ 
with $\sV=\Fix(J_{\sN(W_1)}\Delta_{\sN(W_1)}^{1/2}) = \sN(W_1)$, 
cf.~condition (c) in Proposition~\ref{prop:sl2standcond}. 
On the previous domain we then have 
\[{ \Delta_{h_1}^{-1/2} e^{itH_j} \Delta_{h_1}^{1/2} 
\subset U(\sigma_{h_1})e^{itH_j} U(\sigma_{h_1})}.\] 
With Proposition~\ref{prop:sl2standcond}(a) we can now conclude that 
\begin{equation}\label{eq:ext2}
U(\sigma_{h_1})e^{itH_j} U(\sigma_{h_1})=e^{-itH_j}\quad \mbox{ for }  \quad t\in\RR
\end{equation}
because both sides are  bounded operators which coincide on a 
dense subspace. 
Since the Lie algebra $\g$ is generated by $h_1,\ldots,h_n,$ we obtain 
\begin{equation}\label{eq:ext3}
U(\sigma_{h_1})U(g)U(\sigma_{h_1})=U(\sigma_{h_1}g\sigma_{h_1})
\quad \mbox{ for all } \quad g\in G^\uparrow.\end{equation}
In particular, $U$ defines an (anti-)unitary representation of~$G$. 
Pick $\alpha \in Z(G^\up)^-$ such that (HK4) is satisfied 
and consider the twisted representation of $G$ defined by 
$U^\alpha(\sigma_{h_1}):=Z_\alpha J_{\sN(W_1)} = Z_\alpha U(\sigma_{h_1})$. 
Since $\sN$ coincides with the restriction to $\cW_+$ of  
the BGL net of the (anti-)unitary representation $U$ of $G$, 
the representation $U^\alpha$ satisfies (HK7) and (HK8) 
by Proposition~\ref{prop:BGLpart2}.
\end{proof}

Note that the density property as well as the existence of orthogonal wedges are sufficient but not necessary to have a $G$-covariant action: 
Consider the BGL net associated to the unique irreducible positive energy representation $U$ of the $G = \Aff(\R)$ on the real line. 
Then the standard 
subspaces $\sN_U(a,\infty)$ and $\sN_U(-\infty, b)$  are associated to positive and negative half-lines and satisfy (HK1)-(HK5). There are no-orthogonal wedges in this 
case but the extension to an  (anti-)unitary 
representation of $G$ is given by
$$U(\sigma_W)=J_{\sN_U(W)}.$$
We further remarks that in this case $\sigma_W$ does not preserve the wedge family $\cW_+$.

For the Poincar\'e group, with the identification of wedge regions 
and Euler elements (see \eqref{eq:abs-conc-wedge-poincare}), 
the axial wedges 
\[ W_j=\{(t,x)\in\RR^{1+d}:|t|<x_j\}, \quad j=1,\ldots,d,\] 
define a family of {orthogonal wedge regions, namely wedge regions associated to orthogonal Euler elements.} 
Considering wedges as subsets of Minkowski spaces one can define further regions by wedge intersection. Spacelike cones are particularly important: they are defined, up to translations by finite intersection of  wedges obtained by 
Lorentz transforms of $W_1$. Analogously one can define, by intersecting wedge subspaces, subspace associated to any spacelike cone. In principle this can also be  trivial, but if they are standard, the cyclicity assumption of \ref{prop:sl2standcond}(c) is ensured, cf.~\cite{GL95}. 

Consider $G=\widetilde\Mob\rtimes\{\id,\tilde\tau\}$. Let $(\cW_+,U,\sN)$ be a net of standard subspaces satisfying (HK1)-(HK5). Let $\tilde I_\supset \subeq \R$ 
be an interval with 
$q(\tilde I_\supset)=I_\supset$ where the latter is the right semicircle 
with endpoints $(-i,i)\subset\bS^1$. Then the dilation generators $\tilde \delta_\cap$ and $\tilde \delta_\supset$ define orthogonal Euler elements generating $\widetilde{\Mob}$. Considering the wedges $W_\cap=(x_\cap, \sigma_\cap)$ and  $W_\supset=(x_\supset, \sigma_\supset)$ with $W_\cap=\tilde\rho(\pi/2)W_\supset$,
 the intersection is again a wedge interval $\tilde I=\tilde I_\cap\cap\tilde I_\supset$. In particular, by isotony,  $\sN(\tilde I_\cap)\cap\sN(\tilde I_\supset)\supset \sN(\tilde I)$ is standard and condition (c) in Proposition \ref{prop:sl2standcond} holds.

\appendix
\section{Toolbox} 
\mlabel{app:a}

\begin{prop} \mlabel{prop:antiuniext}  
{\rm(\cite[Thm.~2.11(a)]{NO17})} 
If $(U,\cH)$ is a unitary representation of $G^\up$, then any two 
(anti-)unitary extensions $(\tilde U_j, \cH)$, $j =1,2$, of $U$ to $G$ 
are unitarily 
equivalent, i.e., there exists $\Gamma \in U(G^\up)'$ with 
\[ \Gamma \circ \tilde U_1(g) = \tilde U_2(g) \circ \Gamma \quad \mbox{ for } \quad 
g \in G.\] 
\end{prop}

\begin{lem} \mlabel{lem:symspace} 
Let $\cM \subeq B(\cH)$ be a von Neumann algebra 
and $J \in \Conj(\cH)$ such that $J \cM J = \cM$. 
Then the exponential function of the Banach symmetric space 
\[ \U(\cM)^{J,-} := \{ U \in \U(\cM) \: J U J = U^{-1} \}  \] 
is surjective, i.e., for every $U \in \U(\cM)^{J,-}$ there exists 
an element $X = - X^* \in \cM$ with $J X J = - X$ such that $U = e^X$. 
\end{lem}

\begin{prf} We consider the antilinear automorphism 
\[ \alpha \: \cM \to \cM, \quad 
\alpha(M) := J M J \] 
of the von Neumann algebra $\cM$. 
Let $\cN \subeq \cM$ be the abelian von Neumann algebra generated 
by a fixed element~$U \in \U(\cM)^{J,-}$. 
Then $\alpha(U) = U^{-1} = U^*$ implies that 
$\alpha(\cN) = \cN$ with $\alpha(A) = A^*$ for every $A \in \cN$. 
Any spectral resolution of $U$ in $\cN$ and any bounded measurable function 
$f \: \T \to i\R$ with $e^{f} = \id_\T$ yields an element 
$X := f(U)\in \cN$ with $X^* = -X$ and $e^{X} = U$. 
Then $J X J = \alpha(X) = X^* = - X.$
\end{prf}

The following lemma is \cite[Lemma~A.1]{NO17}: 

\begin{lemma} \mlabel{lem:app.1}
Let $\cM\subeq \cH$ be a von Neumann algebra, 
$\alpha \: \cM \to \cM$ a real-linear weakly continuous automorphism and $U \in \U(\cM)$ be a unitary element. Then 
the following assertions hold: 
\begin{itemize}
\item[\rm(a)] If $\alpha$ is complex linear and $\alpha(U)= U$, 
then there exists a $V \in \U(\cM)$ with $\alpha(V) = V$ and $V^2 = U$.
\item[\rm(b)] If $\alpha$ is complex linear and $\alpha(U)= U^{-1}$ 
with $\ker(U+\1) = \{0\}$, then there exists a $V \in \U(\cM)$ with 
$\alpha(V) = V^{-1}$ and $V^2 = U$.
\item[\rm(c)] If $\alpha$ is antilinear and $\alpha(U)= U^{-1}$, 
then there exists a $V \in \U(\cM)$ with $\alpha(V) = V^{-1}$ and $V^2 = U$.
\item[\rm(d)] If $\alpha$ is antilinear and $\alpha(U)= U$  
with $\ker(U+\1) = \{0\}$, then there exists a $V \in \U(\cM)$ with 
$\alpha(V) = V$ and $V^2 = U$.
\end{itemize}
\end{lemma}

\end{document}